\documentclass[10pt,journal,compsoc]{IEEEtran}
\ifCLASSOPTIONcompsoc
  \usepackage[compress]{cite}
\else
  \usepackage{cite}
\fi

\usepackage{graphicx}
\usepackage[caption=false,font=footnotesize,hangindent=5pt]{subfig}
\usepackage{array,amsmath,amsfonts,amsthm,amssymb}
\usepackage[linesnumbered,ruled,vlined]{algorithm2e}
\usepackage[hyphens]{url}
\usepackage{color,comment,xparse}
\usepackage{threeparttable,tabularx,multirow}
\usepackage{tikz}
\usetikzlibrary{shapes.geometric,positioning,fit,calc}

\graphicspath{{./}{figure/}}
\NewDocumentCommand{\E}{o m}{\mathbb{E}\IfValueTF{#1}{_{#1}}{}[#2]}
\newcommand{\minitab}[1]{\begin{tabular}{@{}c@{}}#1\end{tabular}}

\newcommand{\header}[1]{\smallskip\noindent\textbf{#1}}
\newcommand{\indr}[1]{\mathbf{1}\left( #1 \right)}
\newcommand{\bullethdr}[1]{\noindent\textbullet\,\emph{#1}}
\newcommand{\VS}[1]{VS\textsuperscript{#1}}
\newcommand{\RW}[1]{RW\textsuperscript{#1}}

\newtheorem{example}{Example}

\newtheorem{theorem}{Theorem}

\newtheorem{lemma}{Lemma}

\newcommand{\xslant}{1}
\newcommand{\yslant}{-0.5}

\begin{document}

\title{Sampling Online Social Networks by \\Random Walk with Indirect Jumps}

\author{Junzhou Zhao, Pinghui Wang, John~C.S.~Lui, Don~Towsley, and~Xiaohong~Guan
  \IEEEcompsocitemizethanks{
    \IEEEcompsocthanksitem J. Zhao and J. C.S. Lui are with the
    Department of Computer Science and Engineering, The Chinese University
    of Hong Kong, Hong Kong, China.
    \IEEEcompsocthanksitem D. Towsley is with the School of Computer
    Science, University of Massachusetts at Amherst, MA 01003, USA.
    \IEEEcompsocthanksitem P. Wang and X. Guan are with the
    MOE Key Lab for Intelligent Networks and Network
    Security, Xi'an Jiaotong University, Xi'an 710049, China.
  }
}

\IEEEtitleabstractindextext{%
\begin{abstract}
  Random walk-based sampling methods are gaining popularity and importance in
characterizing large networks.
While powerful, they suffer from the slow mixing problem when the graph is loosely
connected, which results in poor estimation accuracy.
Random walk with jumps (RWwJ) can address the slow mixing problem but it is
inapplicable if the graph does not support uniform vertex sampling (UNI).
In this work, we develop methods that can efficiently sample a graph without the
necessity of UNI but still enjoy the similar benefits as RWwJ.
We observe that many graphs under study, called target graphs, do not exist in
isolation.
In many situations, a target graph is related to an auxiliary graph and a
bipartite graph, and they together form a better connected {\em two-layered
  network structure}.
This new viewpoint brings extra benefits to graph sampling: if directly sampling a
target graph is difficult, we can sample it indirectly with the assistance of the
other two graphs.
We propose a series of new graph sampling techniques by exploiting such a
two-layered network structure to estimate target graph characteristics.
Experiments conducted on both synthetic and real-world networks demonstrate the
effectiveness and usefulness of these new techniques.

\end{abstract}

}

\maketitle

\IEEEdisplaynontitleabstractindextext

\IEEEpeerreviewmaketitle

\section{Introduction}
\label{sec:intro}

Online social networks (OSNs) such as Facebook and Twitter have attracted much
attention in recent years because of their ever-increasing popularity and
importance in our daily lives.
An OSN not only provides a platform for people to connect with their friends, but
also offers an opportunity to study various user characteristics, which are
valuable in many applications such as understanding human
behaviors~\cite{Leskovec2010a,Zhang2013d,Backstrom2014} and inferring user
preferences~\cite{LiEtAl2016b,Han2014}.
Exactly measuring user characteristics requires the complete OSN data.
For third parties who do not possess the data, they can only rely on public APIs
to crawl the OSN.
However, commercial OSNs are typically unwilling to grant third parties full
permission to access the data due to user privacy and business secrecy.
They often impose barriers to limit third parties' large-scale
crawling~\cite{Mondal2012}, e.g., by limiting the API requesting
rate~\cite{WeiboLimits,TwitterLimits}.
As a result, collecting the complete data of a large-scale OSN is practically
impossible.

To address this challenge, sampling methods have been developed, i.e., a small
fraction of OSN users are sampled and used to estimate the OSN user
characteristics.
In the literature, random walk based sampling methods have gained
popularity~\cite{Massoulie2006,Avrachenkov2010,Ribeiro2010,Gjoka2011,Ribeiro2012b,Lee2012b,Xu2014}.
In a typical random walk sampling, a walker is launched over a graph, which
continuously moves from a node to one of its neighbors selected uniformly at
random, to obtain a collection of node samples.
These samples can be used to obtain unbiased estimates of nodal or topological
properties of the graph.
Because a random walk only explores neighborhood of a node during sampling, it is
suitable for crawling and sampling large-scale OSNs.

\subsection{Problems in Random Walk Based Sampling}

While random walk sampling is powerful, if a graph is loosely connected, e.g.,
consists of communities, it will suffer from {\em slow
  mixing}~\cite{Sinclair1989}, i.e., requires a long ``burn-in'' period to reach
steady state, which results in the need of a large number of samples in order to
achieve desired estimation accuracy.
Previous studies have found that mixing times in many real-world networks are
larger than expected~\cite{Mohaisen2010}.

To overcome the slow mixing problem, an effective approach is to incorporate
uniform node sampling (UNI) into random walk sampling, and enable the walker to
jump to other parts of the graph while walking, aka the {\em random walk with
  jumps} (RWwJ)~\cite{Avrachenkov2010,Ribeiro2012b,Xu2014}.
In UNI, a node is independently sampled uniformly at random from the graph, and in
practice, if users in an OSN have unique numerical IDs, then UNI is conducted by
generating random numbers in the ID space and including those valid IDs as UNI
samples.
RWwJ then leverages UNI to perform jumps on a graph.
Specifically, at each step of RWwJ, the walker jumps with a probability determined
by the node where it currently resides, to a node sampled by UNI.
By incorporating UNI into random walk sampling, the walker can jump out of a
community or disconnected component of a graph, and avoid being trapped, thereby
reducing the mixing time~\cite{Avrachenkov2010}.

The main problem of using RWwJ to sample an OSN is that, {\em some OSNs may not
  support UNI at all} because user IDs are not numerical, or {\em the UNI is
  resource intensive} because the valid IDs are sparsely populated.
For example, in Pinterest~\cite{Pinterest}, a user's ID is an arbitrary length
string, which hence makes UNI practically impossible.
In MySpace and Flickr, although the user IDs are numerical, the fractions of valid
user IDs are only about $10\%$ and $1.3\%$, respectively~\cite{Ribeiro2012b}; in
other words, one has to generate about $10$ (or $77$) random numbers (and verify
them by querying OSN APIs) to obtain {\em one} valid user ID in MySpace (or
Flickr).
In some situations, the valid ID space could become extremely sparse.

\begin{example}[Sampling Weibo users in a city]\label{eg:weibo}
  Suppose we want to measure user characteristics in Sina Weibo~\cite{Weibo},
  which is a popular OSN in China.
  Rather than measuring all the Weibo users, we are only interested in users who
  checked in\footnote{Sina Weibo provides a check-in service~\cite{WeiboPlace}
    that allows users to share location information with their friends, e.g., the
    restaurant she took lunch, the hotel she lived during travel.
    The service is similar to the function in Foursquare and other location-based
    OSNs.}
  venues in a specified city.
  For example, users who shared check-in information at tourist spots, hotels, and
  restaurants of a city could be used to evaluate the city's internationality,
  economic index, etc.
  Suppose the users who checked in the city account for about $0.1\%$ of all Weibo
  users.
  We also know that each Weibo user has a unique $10$-digit numerical ID, and the
  fraction of valid IDs is about $10\%$\footnote{A Weibo user ID is in the range
    $[1000000000,6200000000]$, as of May 2017.
    About $10\%$ of the IDs in this range represent valid users.
  }.
\end{example}

In the above example, when conducting UNI, we expect that a randomly generated
number is a valid user ID, and the corresponding user checked in the city.
This happens with probability $10^{-4}$, and as a result, we have to try $10^4$
times on average to obtain one valid UNI sample!
Without the efficiency of conducting UNI on a graph, we cannot perform jumps, and
hence RWwJ is inapplicable.
This raises the following problem we want to solve in this work:
\begin{quote}
  If we cannot perform jumps on a graph, can we conduct random walk sampling that
  still has the similar benefits as RWwJ?
\end{quote}

\subsection{Overview of Our Approach}

In this work, we design a series of graph sampling techniques that can efficiently
sample a network without the necessity of UNI, but still enjoy the similar
benefits as RWwJ.
The main idea behind our method is to leverage a ``{\em two-layered network
  structure}'' to perform ``{\em indirect jumps}'' on the graph under study, and
indirect jumps can bring similar benefits as the direct jumps in RWwJ.
We first use Example~\ref{eg:weibo} to briefly explain what we mean by two-layered
network structure, and then this discovery immediately motivates us to design an
indirect sampling method, which enables us to perform indirect jumps on a graph.

In Example~\ref{eg:weibo}, directly applying UNI on the user network is
inefficient because of the sparsity of user ID space, i.e., a randomly generated
number is very likely to be an invalid user ID, or the user just lies outside of
the city.
Since directly sampling users by UNI is difficult, we propose to sample users in
an indirect manner.
We notice that besides the user network, we are actually also provided with a
space consisting of venues on a map, as illustrated in Fig.~\ref{f:lbsn}.
If we can sample venues in the city by UNI (or its variants), then we can sample
users indirectly because venues and users are related by their check-in
relationships.
The check-ins tell us which user checked in which place, and for a given venue, we
can query the users who checked in this venue, and hence easily {\em obtain a user
  sample from a venue sample}.
Sampling venues in an area is indeed possible by leveraging the APIs provided by
many location-based OSNs (LBSNs).
Many LBSNs provide APIs for querying venues within an area specified by a
rectangle region with southwest and northeast corners latitude-longitude
coordinates given~\cite{WeiboSearch,FoursquareSearch}.
This function can be used to design efficient sampling methods for sampling venues
in an area on a map~\cite{Li2012,Li2014,Wang2014}.
For example, we can efficiently sample a venue in the city specified by a
rectangle region, and the probability of obtaining this venue sample is
calculable.
Note that a user sample obtained from a venue sample is no longer uniformly
distributed.
Because if a user checked in many venues in the city, the user is likely to be
heavily sampled.
But such bias can be easily removed by a reweighting strategy, which we will
elaborate in Section~\ref{sec:methods}.

\begin{figure*}[htp]
  \centering
  \subfloat[user network and venues on a map\label{f:lbsn}]{\begin{tikzpicture}[
  every node/.style={minimum size=2mm, inner sep=0},
  txt/.style={inner sep=1pt, font=\footnotesize, anchor=west, align=center, fill=white},
  cblue/.style={circle, draw, thin, fill=cyan!20, scale=0.8},
  ]
	\begin{scope}[
		yshift=-1.8cm,
		every node/.append style={yslant=\yslant,xslant=\xslant},
		yslant=\yslant,xslant=\xslant]

		\draw[black, dashed, thin] (0,0) rectangle (2.5,2.5);

    \node[txt] at (0,-.15) {user network};

    \node[rectangle] at (1.25,1.25) {\includegraphics[width=2.5cm]{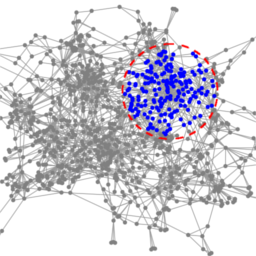}};
	\end{scope}

  \draw[thick,red,dashed] (2.75,-.2) -- (2.6,-1.8);
  \draw[thick,red,dashed] (3.3,-.2) -- (3.9,-1.8);

	\begin{scope}[
		every node/.append style={yslant=\yslant,xslant=\xslant},
		yslant=\yslant,xslant=\xslant]
		\fill[white,fill opacity=.5] (.25,.25) rectangle (2.25,2.25); 
		\draw[black, dashed, thin] (.25,.25) rectangle (2.25,2.25);

    \node[txt] at (.15,.05) {venues on map};

    \node[rectangle] at (1.25,1.25) {\includegraphics[width=2cm]{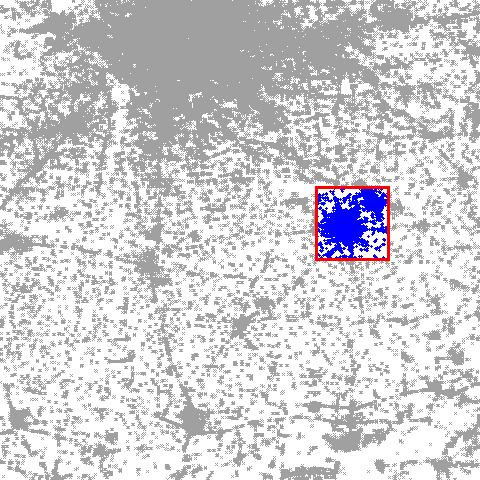}};

	\end{scope}

  \draw[-latex,thick](4.5,-.2) node[txt]{venues\\in the city}
    to[out=180,in=90] (3.1,-.1);
  \draw[-latex,thick](4.4,-1.4) node[txt]{users who\\checked in\\the city}
    to[out=180,in=90] (3.2,-1.5);

\end{tikzpicture}

}
  \subfloat[accounts sharing between two OSNs\label{f:pin_fb}]{\begin{tikzpicture}[
  every node/.style={minimum size=2mm, inner sep=0},
  txt/.style={inner sep=1pt, font=\footnotesize, anchor=west, align=center, fill=white},
  pu/.style={circle, draw, thick, minimum size=3pt},
  fu/.style={rectangle, draw, thick, minimum size=3pt},
  ]

	\begin{scope}[
		yshift=-1.8cm,
		every node/.append style={yslant=\yslant,xslant=\xslant},
		yslant=\yslant,xslant=\xslant]

		\draw[black, dashed, thin] (0,0) rectangle (2.5,2.5);

    \node[txt] at (0,-.15) {Pinterest};

    \node[pu,red] (p1) at (.5, .5) {};
    \node[pu,blue] (p2) at (.9, 1.2) {};
    \node[pu,green] (p3) at (1.7, 1.5) {};
    \node[pu,cyan] (p4) at (2, 2.1) {};
    \node[pu,magenta] (p5) at (1.7, 1) {};

    \node[pu]     (x0) at (.2, .2) {};
    \node[pu]     (x1) at (.3, 1) {};
    \node[pu]     (x2) at (.8, .4) {};
    \node[pu]     (x3) at (.4, 1.5) {};
    \node[pu]     (x4) at (.8, 2.2) {};
    \node[pu]     (x5) at (1.2, 1.8) {};
    \node[pu]     (x6) at (1.3, .3) {};
    \node[pu]     (x7) at (1.4, .8) {};
    \node[pu]     (x8) at (2, 1.3) {};
    \node[pu]     (x9) at (2, .6) {};

    \draw[thick] (p1) -- (p2) -- (p3) -- (p4);
    \draw[thick] (p2) -- (p5) -- (p3) -- (x5);
    \draw[thick] (x0) -- (p1) -- (x1) -- (p2) -- (x6);
    \draw[thick] (x2) -- (p1);
    \draw[thick] (x3) -- (p2) -- (x4);
    \draw[thick] (p5) -- (x7);
    \draw[thick] (x8) -- (p5) -- (x9);
	\end{scope}

	\begin{scope}[
		every node/.append style={yslant=\yslant,xslant=\xslant},
		yslant=\yslant,xslant=\xslant]

    \coordinate (c1) at (.5, .5);
    \coordinate (c2) at (.9, 1.2);
    \coordinate (c3) at (1.7, 1.5);
    \coordinate (c4) at (2, 2.1);
    \coordinate (c5) at (1.7,1);
  \end{scope}

  \draw[thick,dashed,red] (c1) -- (p1);
  \draw[thick,dashed,blue] (c2) -- (p2);
  \draw[thick,dashed,green] (c3) -- (p3);
  \draw[thick,dashed,cyan] (c4) -- (p4);
  \draw[thick,dashed,magenta] (c5) -- (p5);

	\begin{scope}[
		every node/.append style={yslant=\yslant,xslant=\xslant},
		yslant=\yslant,xslant=\xslant]

		\fill[white,fill opacity=.7] (.25,.25) rectangle (2.25,2.25); 
		\draw[black, dashed, thin] (.25,.25) rectangle (2.25,2.25);

    \node[txt] at (.25,.1) {Facebook};

    \node[fu,red] (f1) at (c1) {};
    \node[fu,blue] (f2) at (c2) {};
    \node[fu,green] (f3) at (c3) {};
    \node[fu,cyan] (f4) at (c4) {};
    \node[fu,magenta] (f5) at (c5) {};

    \node[fu] (y1) at (1.2,.6) {};
    \node[fu] (y3) at (1.6,1.9) {};
    \node[fu] (y4) at (1.2,1.8) {};
    \node[fu] (y5) at (.8,2) {};
    \node[fu] (y6) at (2,1.3) {};
    \node[fu] (y7) at (1.9,.5) {};
    \node[fu] (y8) at (.5,1.4) {};

    \draw[thick] (f1) -- (f2) -- (f3) -- (f4);
    \draw[thick] (y8) -- (f2) -- (f5) -- (f3) -- (y6) -- (f5) -- (y7);

    \draw[thick] (f1) -- (y1) -- (f2) -- (y4) -- (y3);
    \draw[thick] (y5) -- (y4) -- (f3) -- (y3) -- (f4);
	\end{scope}

  \node[txt] at (4.2,-1) {account\\sharing};

\end{tikzpicture}

}
  \subfloat[item network and tag/category network\label{f:cluster}]{\begin{tikzpicture}[
  every node/.style={minimum size=2mm, inner sep=0},
  txt/.style={inner sep=1pt, font=\footnotesize, anchor=west, align=center,fill=white},
  tag/.style={rectangle,draw, minimum size=3pt,thick},
  ]

	\begin{scope}[
		yshift=-1.8cm,
		every node/.append style={yslant=\yslant,xslant=\xslant},
		yslant=\yslant,xslant=\xslant]

		\draw[black, dashed, thin] (0,0) rectangle (2.5,2.5);

    \node[txt] at (0,-.15) {item network};

    \node[rectangle] at (1.25,1.25) {\includegraphics[width=25mm]{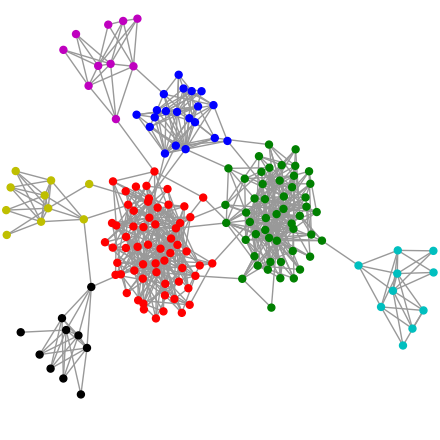}};

    \coordinate (c_k) at (.4, .6);
    \coordinate (c_r) at (.9, 1.1);
    \coordinate (c_y) at (.2, 1.4);
    \coordinate (c_g) at (1.5, 1.3);
    \coordinate (c_b) at (1, 2.02);
    \coordinate (c_c) at (2.3, 1.1);
    \coordinate (c_m) at (.6, 2.02);
	\end{scope}

	\begin{scope}[yslant=\yslant,xslant=\xslant]
    \coordinate (t_k) at (.4, .6);
    \coordinate (t_r) at (.9,1.1);
    \coordinate (t_y) at (.5, 1.2);
    \coordinate (t_g) at (1.7, 1.1);
    \coordinate (t_b) at (1.3, 1.7);
    \coordinate (t_c) at (1.8, 1.6);
    \coordinate (t_m) at (.65, 2);
  \end{scope}

  \draw[thick,dashed,black] (c_k) -- (t_k);
  \draw[thick,dashed,red] (c_r) -- (t_r);
  \draw[thick,dashed,yellow] (c_y) -- (t_y);
  \draw[thick,dashed,green] (c_g) -- (t_g);
  \draw[thick,dashed,blue] (c_b) -- (t_b);
  \draw[thick,dashed,cyan] (c_c) -- (t_c);
  \draw[thick,dashed,magenta] (c_m) -- (t_m);

	\begin{scope}[
		every node/.append style={yslant=\yslant,xslant=\xslant},
		yslant=\yslant,xslant=\xslant]
		\fill[white,fill opacity=.75] (.25,.25) rectangle (2.25,2.25);
		\draw[black, dashed, thin] (.25,.25) rectangle (2.25,2.25);

    \node[tag,black] (a_k) at (t_k) {};
    \node[tag,red] (a_r) at (t_r) {};
    \node[tag,yellow] (a_y) at (t_y) {};
    \node[tag,green] (a_g) at (t_g) {};
    \node[tag,blue] (a_b) at (t_b) {};
    \node[tag,cyan] (a_c) at (t_c) {};
    \node[tag,magenta] (a_m) at (t_m) {};

    \draw[thick,dashed] (a_k) -- (a_r) -- (a_y);
    \draw[thick,dashed] (a_b) -- (a_r) -- (a_m);
    \draw[thick,dashed] (a_c) -- (a_b) -- (a_g);

    \node[txt] at (.2,.05) {tags/categories};
	\end{scope}

  \node[txt] at (3.5,-1) {tagging /\\categorizing};
\end{tikzpicture}

}
  \caption{Examples of two-layered network structures}
  \label{fig:two_layered}
\end{figure*}
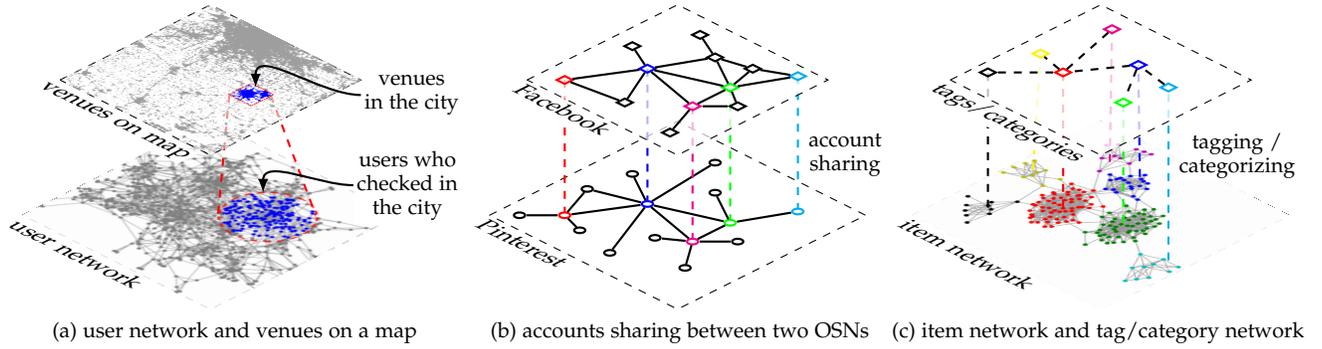

An important lesson learned from solving the problem in Example~\ref{eg:weibo} is
that, the {\em two-layered network structure}, consisting of the user network
layer and the venues layer, can help us to obtain samples of one layer when
sampling another layer is easy.
Hence, this enables us to conduct ``indirect jumps'' on the user network with the
help of venue sampling.
We further find that the two-layered network structure is not unique to the
problem in Example~\ref{eg:weibo}, but is pervasive in a wide range of graph
sampling problems, and more examples will be presented in
Section~\ref{sec:two_layered}.
Hence, it is necessary to develop some unified graph sampling techniques that can
leverage the two-layered network structure to address these graph sampling
problems.

In general, there are \emph{three graphs} related to the two-layered network
structure: (1) a {\em target graph}, whose characteristics are of interest to us
and need to be estimated, e.g., the user sub-network in Example~\ref{eg:weibo};
(2) an {\em auxiliary graph}, which is easier to be sampled than the target graph,
e.g., the venues can be thought of as nodes in an auxiliary graph, and
Example~\ref{eg:weibo} is a special case where the auxiliary graph has an empty
edge set; and (3) a {\em bipartite graph} that connects nodes in the target and
auxiliary graphs.
When directly sampling the target graph is difficult, we can turn to sample the
auxiliary graph, and the bipartite graph bridges the two sample spaces and allows
us to sample the target graph in an indirect manner.
This thus enables us to perform indirect jumps on the target graph, and allows us
to develop random walk sampling methods with indirect jumps that have the similar
benefits as RWwJ.

\subsection{Contributions}

We make three contributions in this work:
\begin{itemize}
\item We discover the pervasiveness and usefulness of a ``two-layered network
  structure'', that exists in many real-world applications, and can be exploited
  to efficiently sample a graph in an indirect manner if directly sampling the
  graph is difficult.
\item We design three new sampling techniques by leveraging such a two-layered
  network structure.
  These new techniques enable us to conduct random walk sampling that has the
  similar benefits as RWwJ.
\item We conduct extensive experiments on both synthetic and real-world networks
  to validate our proposed methods.
  The experimental results demonstrate the effectiveness of our designed sampling
  techniques.
\end{itemize}

\subsection{Outline}

The reminder of this paper will proceed as follows.
In Section~\ref{sec:preliminaries}, we provide some preliminaries about graph
sampling.
In Section~\ref{sec:two_layered}, we formally define the two-layered network
structure along with more examples.
In Section~\ref{sec:methods} we elaborate three new sampling methods.
In Section~\ref{sec:experiment}, we conduct experiments to validate our methods.
Section~\ref{sec:related_work} reviews some related literature, and
Section~\ref{sec:conclusion} concludes.

\section{Preliminaries}
\label{sec:preliminaries}

In this section, we provide some preliminaries about the graph sampling problem,
and review a random walk based sampling method named random walk with jumps (RWwJ).

\subsection{Graph Sampling}

An OSN can be modeled as an undirected graph\footnote{For Facebook, the friendship
  network is an undirected graph; for Twitter, because the followees and followers
  of a user are known once the user is collected, hence we can build an undirected
  graph of the Twitter follower network on-the-fly.}
$G=(U,E)$, where $U$ is a finite set of nodes representing users, and $E\in U\times
U$ is a set of edges representing relations between users.
We assume that the graph $G$ has no self-loops and no multiple edges connecting two
nodes.
Also, the graph size $|U|=n$ may be not known in advance.

Let $f\colon U\mapsto\mathbb{R}$ be any desired {\em characteristic function} that
maps a node in the graph to a real number.
The goal of measuring the characteristic of graph $G$ is to estimate
\[
  \theta\triangleq \frac{1}{n}\sum_{u\in U}f(u),
\]
which is the aggregated nodal characteristic of the graph.
For example, in an OSN, we let $f(u)=1$ if user $u$ is female, and otherwise
$f(u)=0$, then $\theta$ represents the fraction of female users in the OSN.

The goal of graph sampling is to design an algorithm for collecting node samples
$S$ from graph $G$, constrained by a budget $|S|\leq B\ll n$, and for providing
unbiased estimate of $\theta$ with low statistical error.

\subsection{Random Walk with Jumps}

Random walk with jumps (RWwJ)~\cite{Avrachenkov2010} is a popular graph sampling
method that can address the slow mixing issue of a simple random walk when the
graph has community structures.
RWwJ generally works as follows: A walker starts from a node in the graph, and at
each step, it moves to a neighbor selected uniformly at random, or jumps to a node
uniformly sampled from the graph, and the probability of jumping is determined by
the node where the walker currently resides; this process continues until enough
samples are collected.

An easier way to think about RWwJ is that, we modify the structure of the original
graph by connecting every node in the graph to a virtual {\em jumper node}, with
edge weight $\alpha\geq 0$; then a simple random walk on this modified graph is
equivalent to RWwJ.
Figure~\ref{fig:rwwj} illustrates RWwJ on a loosely connected graph.
Comparing the modified graph with the original graph, we can find that the
modified graph always has larger {\em graph conductance} than the original graph,
and because larger graph conductance usually implies faster mixing of a random
walk~\cite{Sinclair1989}, hence, RWwJ has the advantage of faster mixing than a
simple random walk on poorly connected graphs~\cite{Avrachenkov2010}.

\begin{figure}[htp]
  \centering
  \begin{tikzpicture}[
  every node/.style={minimum size=2mm, inner sep=0},
  txt/.style={font=\footnotesize, anchor=west, align=center, fill=white},
  tu/.style={circle, draw, thick, minimum size=4pt},
  jumper/.style={circle, draw, thick, blue, minimum size=5pt},
  je/.style={densely dotted,thick,blue},
  ]

	\begin{scope}[
    scale=.85,
		every node/.append style={yslant=\yslant,xslant=\xslant},
		yslant=\yslant,xslant=\xslant
    ]

		\draw[black, dashed, thin] (0,0) rectangle (2.5,2.5);

    \node[tu] (x0) at (.8, 1.2) {};
    \node[tu] (x1) at (.3, 1) {};
    \node[tu] (x2) at (1, .6) {};
    \node[tu] (x3) at (1.7, 1.5) {};
    \node[tu] (x4) at (.8, 1.8) {};
    \node[tu] (x5) at (1.5, 1.9) {};
    \node[tu] (x6) at (1.3, 1.2) {};
    \node[tu] (x7) at (2, 2) {};
    \node[tu] (x8) at (2.2, 1.3) {};
    \node[tu] (x9) at (1.8, 1) {};
    \node[tu] (x10) at (.3, 1.5) {};

    \draw[thick] (x1) -- (x0) -- (x2) -- (x1) -- (x10) -- (x4);
    \draw[thick] (x7) -- (x5) -- (x3) -- (x9) -- (x8) -- (x3) -- (x7) -- (x8);
    \draw[thick] (x2) -- (x6) -- (x4) -- (x0) -- (x6) -- (x3);
    \draw[thick] (x0) -- (x10);

    \draw[very thick, blue, dashed] (1.8,.5) -- (1.1,2.1);
	\end{scope}

	\begin{scope}[
		xshift=4.5cm,
    scale=.85,
		every node/.append style={yslant=\yslant,xslant=\xslant},
		yslant=\yslant,xslant=\xslant]

		\draw[black, dashed, thin] (0,0) rectangle (2.5,2.5);

    \node[tu] (x0) at (.8, 1.2) {};
    \node[tu] (x1) at (.3, 1) {};
    \node[tu] (x2) at (1, .6) {};
    \node[tu] (x3) at (1.7, 1.5) {};
    \node[tu] (x4) at (.8, 1.8) {};
    \node[tu] (x5) at (1.5, 1.9) {};
    \node[tu] (x6) at (1.3, 1.2) {};
    \node[tu] (x7) at (2, 2) {};
    \node[tu] (x8) at (2.2, 1.3) {};
    \node[tu] (x9) at (1.8, 1) {};
    \node[tu] (x10) at (.3, 1.5) {};

    \draw[thick] (x1) -- (x0) -- (x2) -- (x1) -- (x10) -- (x4);
    \draw[thick] (x7) -- (x5) -- (x3) -- (x9) -- (x8) -- (x3) -- (x7) -- (x8);
    \draw[thick] (x2) -- (x6) -- (x4) -- (x0) -- (x6) -- (x3);
    \draw[thick] (x0) -- (x10);

    \node[jumper] (j) at (-.3,2.8) {};
    \foreach \i in {0,1,2,3,4,5,6,7,8,9,10} \draw[je] (x\i) -- (j);

	\end{scope}

  \node[txt,above right = 5mm and 1mm of x1] {$\color{blue}\alpha$};
  \node[txt, right=5pt of j] {jumper node};
  \node[txt] at (1.3,-1.3) {original graph};
  \node[txt] at (5.6,-1.3) {modified graph};

\end{tikzpicture}

  \caption{RWwJ is viewed as a simple random walk on the modified graph.}
  \label{fig:rwwj}
\end{figure}
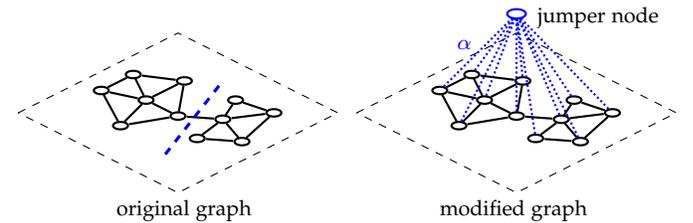

In RWwJ, the probability transition matrix of the underlying Markov chain is given
by
\[
  P_{ij}^{\text{RWwJ}} =
  \begin{cases}
    \frac{\alpha/n+1}{d_i+\alpha}, & (i,j)\in E, \\
    \frac{\alpha/n}{d_i+\alpha}, & (i,j)\notin E. \\
  \end{cases}
\]
That is, if $(i,j)\in E$, the walker starting from $i$ could walk to $j$ (in one
step) through the edge $(i,j)$ with probability $\frac{1}{d_i + \alpha}$; or jump
to $j$ through UNI with probability $\frac{\alpha}{d_i+\alpha}\cdot\frac{1}{n} =
\frac{\alpha/n}{d_i+\alpha}$; thus the transition probability on edge $(i,j)$ is
$\frac{\alpha/n+1}{d_i+\alpha}$.
If $(i,j)\notin E$, the walk starting from $i$ can only walk to $j$ (in one step)
by jumping with probability $\frac{\alpha/n}{d_i+\alpha}$.

When RWwJ reaches the steady state, a node $u\in U$ is sampled with probability
proportional to $d_u+\alpha$.
If we let $S$ denote the samples collected by RWwJ, an asymptotically unbiased
estimator of $\theta$ is given by
\begin{equation}\label{est:rwwj}
  \hat\theta^\text{RWwJ}=
  \frac{1}{Z^{\text{RWwJ}}}\sum_{s\in S}\frac{f(s)}{d_s+\alpha},
\end{equation}
where $Z^{\text{RWwJ}}\triangleq\sum_{s\in S}1/(d_s+\alpha)$.
We can understand the unbiasedness of Estimator~\eqref{est:rwwj} by leveraging the
ratio form of the {\em Law of Large Numbers} of Markov chains.

\begin{lemma}[{Law of Large Numbers~\cite[p.427--428]{Meyn2009}}]\label{lemma:lln}
  Let $S$ be a sample path obtained by a Markov chain defined on state space $U$
  with stationary distribution $\pi$.
  For any function $f,g\colon U\mapsto\mathbb{R}$, and let $F_S(f)\triangleq
  \sum_{s\in S}f(s)$, $\E[\pi]{f}\triangleq\sum_{u\in U}\pi_uf(u)$.
  It holds that
  \begin{align}\label{eq:LLN}
    \lim_{|S|\rightarrow\infty}\frac{1}{|S|}F_S(f) &= \E[\pi]{f} \quad a.s., \\
    \lim_{|S|\rightarrow\infty}\frac{F_S(f)}{F_S(g)} &= \frac{\E[\pi]{f}}{\E[\pi]{g}} \quad a.s..
  \end{align}
  Here, ``a.s.'' denotes ``almost sure'' convergence, i.e., the event of interest
  happens with probability one.
\end{lemma}

Therefore, in Estimator~\eqref{est:rwwj}, replacing $f(s)/(d_s+\alpha)$ by $f(s)$,
and $1/(d_s+\alpha)$ by $g(s)$, we obtain that $\hat\theta^{\text{RWwJ}}$ converges
to $\E[\pi]{f}/\E[\pi]{g}=\theta$, almost surely.

Although RWwJ can address the slow mixing problem, it requires UNI to perform
jumps on a graph.
If the OSN does not support UNI, or UNI is inefficient, RWwJ becomes inapplicable.
In this work, we introduce a two-layered network structure that exists in many
real-world applications, and we will show that such a structure can be leveraged
to design random walk sampling methods having the similar benefits as RWwJ even
though we cannot conduct UNI on the graph.

\section{Two-Layered Network Structure}
\label{sec:two_layered}

In this section, we first formally describe the two-layered network structure we
discovered in Example~\ref{eg:weibo}.
Then we provide more examples to demonstrate the pervasiveness of such a
structure.

\subsection{Definition}

We use three undirected graphs to describe a two-layered network structure:
$G(U,E)$, $G'(V,E')$, and $G_b(U,V,E_b)$, where $U,V$ are two sets of nodes, and
$E\subseteq U\times U, E'\subseteq V\times V, E_b\subseteq U\times V$ are three
sets of edges.
More specifically,
\begin{itemize}
\item $G(U,E)$ is the {\em target graph}, whose characteristic $\theta$ is of
  interest to us and needs to be measured.
  For example, the user social network in Example~\ref{eg:weibo} can be treated as
  the target graph.
\item $G'(V,E')$ is an {\em auxiliary graph}, which can be more efficiently
  sampled than the target graph.
  In Example~\ref{eg:weibo}, we can construct an auxiliary graph where the nodes
  represent the venues in the city, and the edge set is left empty
  (i.e., $E'=\emptyset$).
\item $G_b(U,V,E_b)$ is a {\em bipartite graph} that connects nodes in the target
  and auxiliary graphs.
  In Example~\ref{eg:weibo}, the bipartite graph is formed by users, venues and
  their check-in relationships.
\end{itemize}

An example of such a two-layered network structure is illustrated in
Fig.~\ref{fig:two_layered}.
The target graph consists of two disconnected components, however, if we view the
three graphs as a whole, they form a better connected graph than the target graph
itself.
Hence, it is possible to sample target graph efficiently with the help of the
other two graphs.
With this intuition in mind, we will see in next section that we indeed can design
efficient sampling methods by leveraging this two-layered network structure.

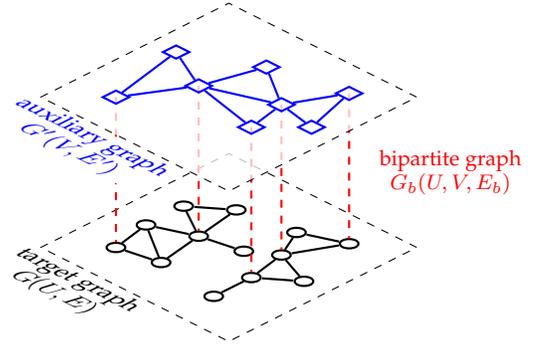
\begin{figure}[htp]
\centering
\begin{tikzpicture}[
  every node/.style={minimum size=2mm, inner sep=0, font=\tiny},
  txt/.style={font=\footnotesize, anchor=west, align=center, fill=white},
  tu/.style={circle, draw, thick, minimum size=5pt},
  au/.style={rectangle, draw=blue, thick, minimum size=5pt},
  be/.style={dashed,thick,red},
  ae/.style={thick,blue},
  ]

	\begin{scope}[
		yshift=-2cm,
		every node/.append style={yslant=\yslant,xslant=\xslant},
		yslant=\yslant, xslant=\xslant]

		\draw[black, dashed, thin] (0,0) rectangle (2.5,2.5);

    \node[txt] at (-.1,-.35) {target graph\\$G(U,E)$};

    \node[tu] (p1) at (0.5, 0.5) {};
    \node[tu] (p2) at (0.9, 1.2) {};
    \node[tu] (p3) at (1.7, 1.5) {};
    \node[tu] (p4) at (2.0, 2.1) {};
    \node[tu] (p5) at (1.8, 1.0) {};

    \node[tu] (x1) at (0.4, 1.0) {};
    \node[tu] (x2) at (1.0, 0.6) {};
    \node[tu] (x3) at (0.4, 1.5) {};
    \node[tu] (x4) at (0.8, 1.8) {};
    \node[tu] (x5) at (1.5, 1.9) {};
    \node[tu] (x6) at (1.4, 1.3) {};
    \node[tu] (x7) at (1.8, 0.5) {};
    \node[tu] (x8) at (2.2, 1.3) {};

    \draw[thick] (p2) -- (x2) -- (x1);
    \draw[thick] (x7) -- (p5) -- (x8) -- (p3) -- (p4) -- (x5) -- (p3) -- (p5);
    \draw[thick] (x2) -- (p1) -- (x1) -- (p2) -- (x6)
                 (x3) -- (x4) -- (p2) -- (x3);
	\end{scope}

	\begin{scope}[
		every node/.append style={yslant=\yslant,xslant=\xslant},
		yslant=\yslant,xslant=\xslant]

    \coordinate (c1) at (0.5, 0.5);
    \coordinate (c2) at (0.9, 1.2);
    \coordinate (c3) at (1.7, 1.5);
    \coordinate (c4) at (2.0, 2.1);
    \coordinate (c5) at (1.8, 1.0);
  \end{scope}

  \draw[be] (c1) -- (p1);
  \draw[be] (c2) -- (p2);
  \draw[be] (c3) -- (p3);
  \draw[be] (c4) -- (p4);
  \draw[be] (c5) -- (p5);

	\begin{scope}[
		every node/.append style={yslant=\yslant,xslant=\xslant},
		yslant=\yslant,xslant=\xslant]

		\fill[white,fill opacity=.75] (0,0) rectangle (2.5,2.5); 
		\draw[black, dashed, thin] (0,0) rectangle (2.5,2.5);

    \node[txt,text=blue] at (-.1,-.35) {auxiliary graph\\$G'(V,E')$};

    \node[au] (f1) at (c1) {};
    \node[au] (f2) at (c2) {};
    \node[au] (f3) at (c3) {};
    \node[au] (f4) at (c4) {};
    \node[au] (f5) at (c5) {};

    \node[au] (y1) at (0.3, 1.5) {};
    \node[au] (y2) at (1.1, 1.9) {};
    \node[au] (y3) at (2.2, 1.4) {};

    \draw[ae] (f2) -- (y1) -- (f1) -- (f2) -- (f3) -- (f4) -- (y3);
    \draw[ae] (f3) -- (y2) -- (f2) -- (f5) -- (f3) -- (y3);
	\end{scope}

  \node[txt,text=red] at (4.5,-1) {bipartite graph\\$G_b(U,V,E_b)$};

\end{tikzpicture}

\caption{Illustration of the two-layered network structure.}
  \label{fig:two_layered}
\end{figure}

\subsection{More Examples}

The two-layered network structure is not unique to Example~\ref{eg:weibo}, but
exists in a wide range of real-world applications.
In the following, we provide more examples.

\begin{example}[Accounts sharing between two OSNs\label{eg:osn}]
  Many OSNs now support using an existing OSN's accounts to login another OSN.
  For example, Facebook users can login Pinterest using their Facebook accounts.
  This naturally forms a two-layered network structure consisting of Facebook and
  Pinterest.
  Suppose we want to measure Pinterest, then we can let target graph represent
  Pinterest, auxiliary graph represent Facebook, and bipartite graph represent
  their account sharing relations.
\end{example}

Figure~\ref{f:pin_fb} illustrates Example~\ref{eg:osn}.
Note that Pinterest does not support UNI, hence RWwJ is inapplicable.
Instead, using the techniques developed in this work, we will be able to leverage
Facebook to sample Pinterest.

\begin{example}[Amazon item network and categories\label{eg:clusters}]
  Items in Amazon are related with each other to form an item network.
  Each item also belongs to one or more categories.
  Meanwhile, Amazon provides a complete category list to facilitate customers to
  quickly navigate to the items they are looking for.
  This forms a two-layered network structure consisting of items and categories.
  Suppose we want to measure the item network, then we can let target graph
  represent the item network, auxiliary graph represent the category list, and
  bipartite graph represent the affiliation relations between items and
  categories.
\end{example}

Figure~\ref{f:cluster} illustrates Example~\ref{eg:clusters}.
Note that categories could also be tags and they may also form a tag network.
Items are very likely to form clusters, and hence easily trap a random walker.
If we can leverage the category information, and help a random walker to jump out
of clusters, we can sample the item network in a more efficient way.

\section{Sampling Design}
\label{sec:methods}

In this section, we leverage the two-layered network structure and design three
new sampling techniques to sample and characterize the target graph.

\subsection{Indirectly Sampling Target Graph by Vertex Sampling on
Auxiliary Graph (\VS{A})}

The first method assumes that vertex sampling is easier to conduct on the
auxiliary graph than on the target graph, as is the case in
Example~\ref{eg:weibo}.
We present a sampling method \VS{A} (and its two implementations \VS{A}-I and
\VS{A}-II) to indirectly sample the target graph under this setting.
The basic idea of \VS{A} is illustrated in Fig.~\ref{fig:vsa}.

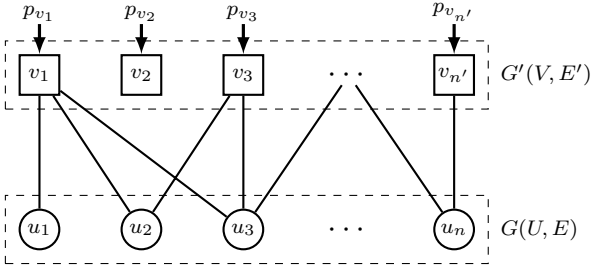
\begin{figure}[htp]
\centering
\begin{tikzpicture}[
  every node/.style={inner sep=1pt, minimum size=15pt, font=\footnotesize},
  und/.style={draw, thick, circle},
  vnd/.style={draw, thick, rectangle},
  att/.style={left, minimum size=0},
  box/.style={draw, dashed, inner sep=5pt},
  arr/.style={very thick, -latex},
  ]

  \node(u1) [und] {$u_1$};
  \node(u2) [und, right=.8 of u1] {$u_2$};
  \node(u3) [und, right=.8 of u2] {$u_3$};
  \node(p1) [att, right=.8 of u3] {\large $\cdots$};
  \node(un) [und, right=.8 of p1] {$u_n$};
  \node[box, fit={(u1) (un)}] {};
  \node[att, right=3mm of un] {$G(U,E)$};

  \node(v1) [vnd, above=1.5 of u1] {$v_1$};
  \node(v2) [vnd, right=.8 of v1] {$v_2$};
  \node(v3) [vnd, right=.8 of v2] {$v_3$};
  \node(p2) [att, right=.8 of v3] {\large $\cdots$};
  \node(vn) [vnd, right=.8 of p2] {$v_{n'}$};
  \node[box, fit={(v1) (vn)}] {};
  \node[att, right=3mm of vn] {$G'(V,E')$};

  \draw[thick]
    (u2)--(v1)--(u1)
    (u2)--(v3)--(u3)--(v1)
    (vn)--(un)--(p2)--(u3);

  \foreach \v/\a in {v1/1,v2/2,v3/3,vn/n'}{
    \node[att,above=0.4 of \v] (tmp) {$p_{v_{\a}}$};
    \draw[arr] (tmp)--(\v);
  }
\end{tikzpicture}

\caption{Illustration of \VS{A}. Edges in $G$ and $G'$ are omitted.}
\label{fig:vsa}
\end{figure}

\header{\VS{A}-I.}
Assume that a node $v\in V$ is sampled with probability $p_v\propto a_v >0$ in
auxiliary graph $G'$.
For example, if auxiliary graph $G'$ supports UNI, then $a_v\equiv 1, \forall v\in
V$.
The simplest way to implement \VS{A} is as follows: We first sample a node $v\in
V$ in $G'$, and then sample a neighbor of $v$ in $G_b$ uniformly at random,
denoted by $u$.
Obviously, $u\in U$, and we collect $u$ as a sample.
We refer to this simple sampling method as \VS{A}-I, and will show that samples
collected by \VS{A}-I can indeed yield unbiased estimate of $\theta$.
The detailed design of \VS{A}-I is described as follows.

\header{Sampling design.}
\VS{A}-I repeats the following two steps until sample collection $S$ reaches
budget $B$.
\begin{itemize}
\item Sample a node $v$ from auxiliary graph $G'$;
\item If $v$ has neighbors in bipartite graph $G_b$, sample a neighbor $u$
  uniformly at random, and put $u$ into samples $S$.
\end{itemize}

\header{Estimator.}
In \VS{A}-I, we can see that a node $u\in U$ is sampled with probability
\begin{equation}\label{eq:vsa1-pu}
  p_u\propto b_u\triangleq\sum_{v\in V_u}\frac{a_v}{d_v^{(b)}},
\end{equation}
where $V_u\subseteq V$ is the set of neighbors of $u$ in $G_b$, and $d_v^{(b)}$ is
the degree of $v$ in $G_b$.
Then, we propose to use the following estimator to estimate $\theta$:
\begin{equation}\label{eq:est_vsa1}
  \hat{\theta}^{\VS{A}\text{-I}}
  = \frac{1}{Z^{\VS{A}\text{-I}}}\sum_{u\in S}\frac{f(u)}{b_u},
\end{equation}
where $Z^{\VS{A}\text{-I}}\triangleq\sum_{u\in S}1/b_u$.
The following theorem guarantees its unbiasedness.

\begin{theorem}\label{thm:vsa1}
  Estimator~\eqref{eq:est_vsa1} is asymptotically unbiased.
\end{theorem}

\begin{proof}
  \VS{A}-I can be viewed as sampling $U$ with replacement according to
  distribution $\{p_u\}_{u\in U}$.
  This can be further viewed as generating samples according to a Markov chain
  which has a probability transition matrix with all rows the same vector
  $[p_u]_{u\in U}$, and $\pi_u=p_u,\forall u\in U$.
  This allows us to leverage Lemma~\ref{lemma:lln}, and obtain that
  \begin{align*}
    \lim_{B\rightarrow\infty}\hat{\theta}^{\VS{A}\text{-I}}
    &=\frac{\E{f(u)/b_u}}{\E{1/b_u}}
    =\frac{\sum_{u\in U}p_uf(u)/b_u}{\sum_{u\in U}p_u/b_u} \\
    &=\frac{1}{n}\sum_{u\in U}f(u)
    =\theta \qquad a.s.
  \end{align*}
This thus completes the proof.
\end{proof}

\VS{A}-I has one drawback, i.e., to correct the bias of each sample $u\in S$, we
require $b_u$, which further requires $a_v$ for each neighbor of $u$ in $G_b$ by
Eq.~\eqref{eq:vsa1-pu}.
This is not an issue if we are conducting UNI (or its variants) on the auxiliary
graph, as we have known $a_v, \forall v\in V$ before conducting UNI.
But in some cases where more complicated vertex sampling methods are applied on
auxiliary graph, this condition may not be met: we know $a_v$ only if $v$ is
sampled, otherwise $a_v$ is not known in advance.
And this is actually the case we met in Example~\ref{eg:weibo}: we know the
probability of obtaining a venue sample only if the venue is sampled (this should
become clear when we describe the venue sampling method in
Section~\ref{sec:experiment}).
To address this problem, we propose another sampling method \VS{A}-II.

\header{\VS{A}-II.}
When a node $v\in V$ is sampled in auxiliary graph, we collect all of its
neighbors in the bipartite graph as samples; we repeat this process until enough
samples are collected.
We use these samples to estimate $\theta$.
The detailed design of \VS{A}-II is described as follows.

\header{Sampling design.}
\VS{A}-II repeats the following steps to obtain two sample collections $S$ and
$S'$ from $G$ and $G'$ respectively.
Samples in $S$ are used to estimate $\theta$.
\begin{itemize}
\item Sample a node $v$ from auxiliary graph $G'$;
\item If $v$ has neighbors in bipartite graph $G_b$, put $v$ into samples $S'$,
  and put all the neighbors of $v$ in $G_b$ into samples $S$.
\end{itemize}

\header{Estimator design for \VS{A}-II.}
We propose to estimate $\theta$ using the following estimator:
\begin{equation}\label{eq:est_vsa2}
  \hat\theta^{\text{\VS{A}-II}}
  =\frac{1}{Z^{\text{\VS{A}-II}}}
  \sum_{v\in S'}\frac{1}{a_v}\sum_{u\in U_v}\frac{f(u)}{d_u^{(b)}},
\end{equation}
where $U_v\subseteq U$ is the set of neighbors of $v$ in $G_b$, and
$Z^{\text{\VS{A}-II}}\triangleq\sum_{v\in S'}1/a_v\sum_{u\in U_v}1/d_u^{(b)}$.
The following theorem guarantees its unbiasedness.

\begin{theorem}
  Estimator~\eqref{eq:est_vsa2} is asymptotically unbiased.
\end{theorem}
\begin{proof}
  Using the similar idea as we proved Theorem~\ref{thm:vsa1}, we have
  \begin{align*}
    \E{\frac{1}{a_v}\sum_{u\in U_v}\frac{f(u)}{d_u^{(b)}}}
    &= \sum_{v\in V}\frac{p_v}{a_v}\sum_{u\in U_v}\frac{f(u)}{d_u^{(b)}}
    = c\sum_{v\in V}\sum_{u\in U_v}\frac{f(u)}{d_u^{(b)}} \\
    &= c\sum_{u\in U}d_u^{(b)}\frac{f(u)}{d_u^{(b)}}
    = c\sum_{u\in U}f(u)
    = cn\theta
  \end{align*}
  where $c\triangleq p_v/a_v$ is a constant.
  The third equation holds because each inside item is added exactly $d_u^{(b)}$
  times before we merge the two sums into one sum.
  Similarly,
  \[
    \E{\frac{1}{a_v}\sum_{u\in U_v}\frac{1}{d_u^{(b)}}}
    = \sum_{v\in V}\frac{p_v}{a_v}\sum_{u\in U_v}\frac{1}{d_u^{(b)}}
    = c\sum_{v\in V}\sum_{u\in U_v}\frac{1}{d_u^{(b)}}
    = cn.
  \]
  By Lemma~\eqref{lemma:lln}, we thus obtain
  \[
    \lim_{B\rightarrow\infty}\hat\theta^{\text{\VS{A}-II}}
    =\frac{\E{1/a_v\sum_{u\in U_v}f(u)/d_u^{(b)}}}
    {\E{1/a_v\sum_{u\in U_v}1/d_u^{(b)}}} =\theta
    \quad a.s.
  \]
\end{proof}

\header{Remark.}
It is important to know that \VS{A} (either \VS{A}-I or \VS{A}-II) can provide
unbiased estimate of target graph characteristic under the condition that {\em
  every node in the target graph is connected to nodes in the auxiliary graph}.
If a node $u$ is not connected to any node in $G'$, $u$ cannot be indirectly
sampled by \VS{A}.
This will result in biased estimates, and it is difficult to correct the bias.
In Example~\ref{eg:weibo}, since we are only interested in users who share their
check-ins in Weibo, therefore Example~\ref{eg:weibo} satisfies this condition.

\subsection{Random Walk on Target Graph Incorporating with Vertex
  Sampling on Auxiliary Graph (\RW{T}\VS{A})}

In some situations, $d_u^{(b)}=0$ for some $u\in U$, such as the case in
Example~\ref{eg:osn}, where some Pinterest users may not have Facebook accounts at
all, and these users cannot be sampled by \VS{A} (and as a result, \VS{A} can not
provide unbiased estimates of Pinterest user characteristics).
To address this issue, we propose a second sampling method \RW{T}\VS{A}, which
combines random walk sampling on the target graph with vertex sampling on the
auxiliary graph.

The basic idea of \RW{T}\VS{A} is that, we launch a random walk on the target
graph, and at each step allow the walker to jump with a probability dependent on
the node where the walker currently resides.
This is similar to RWwJ on the target graph $G$, but with the major difference that
in \RW{T}\VS{A} the walker jumps to a node in $G$ by jumping first to a node in
$G'$, and then randomly selecting one of its neighbors in $G_b$ (similar to
\VS{A}-I).
We refer to this as an {\em indirect jump}, and show in experiments that indirect
jumps in \RW{T}\VS{A} bring similar benefits as the direct jumps in RWwJ.
An additional advantage of using random walk on the target graph is that it better
characterizes highly connected nodes than uniform node sampling as random walks are
biased towards high degree nodes in $G$.
We depict \RW{T}\VS{A} in Fig.~\ref{fig:rwtvsa}, where each node in $G$ is
connected to a virtual jumper node to conduct indirect jumps, through doing vertex
sampling over auxiliary graph $G'$.

\begin{figure}[htp]
\centering
\begin{tikzpicture}[
  every node/.style={inner sep=1pt, minimum size=15pt, font=\footnotesize},
  und/.style={draw, thick, circle},
  vnd/.style={draw, thick, rectangle},
  att/.style={left, minimum size=0},
  arr/.style={thick, -latex},
  jumper/.style={und, dashed},
  box/.style={draw, dashed, inner sep=5pt},
  ]

  \node(u1) [und] at (0,0) {$u_1$};
  \node(u2) [und, right=.8 of u1] {$u_2$};
  \node(u3) [und, right=.8 of u2] {$u_3$};
  \node(p1) [att, right=.8 of u3] {\large $\cdots$};
  \node(un) [und, right=.8 of p1] {$u_n$};
  \node(j) [jumper, above=0.9 of u3] {$j$};
  \node[box, fit={(u1) (un)}] {};
  \node[att, right=3mm of un] {$G(U,E)$};
  \node[att, right=3mm of j]  {jumper node};

  \node(v1) [vnd] at (0,3) {$v_1$};
  \node(v2) [vnd, right=.8 of v1] {$v_2$};
  \node(v3) [vnd, right=.8 of v2] {$v_3$};
  \node(p2) [att, right=.8 of v3] {\large $\cdots$};
  \node(vn) [vnd, right=.8 of p2] {$v_{n'}$};
  \node[box, fit={(v1) (vn)}] {};
  \node[att, right=3mm of vn] {$G'(V,E')$};

  \foreach \v/\a in {v1/1, v2/2, v3/3, vn/n'}{
    \node[att, above=0.4 of \v] (tmp) {$p_{v_{\a}}$};
    \draw[arr, very thick] (tmp)--(\v);
  }

  \foreach \i in {1,2,3,n}{
    \draw[thick, densely dotted] (u\i) to
    node[att,inner sep=0, font=\scriptsize] {$w_{u_\i}$} (j);
  }

  \draw[blue, arr, dashed] (j) .. controls (1.8,1.8) and (2,2.5) .. (v2);
  \draw[blue, arr, dashed] (v2) .. controls (.5,2.5) and (.5,.8) .. (u1);
  \node[att, text=blue] at (2.4,2) {(i)};
  \node[att, text=blue] at (.5,1.5) {(ii)};

\end{tikzpicture}

\caption{{\bf Illustration of \RW{T}\VS{A} and indirect jump}.
  Each node $u$ in $G$ is virtually connected to a jumper node $j$ with weight
  $w_u$.
  An indirect jump is performed by: {\bf (i)}~randomly sampling a node $v$ in
  $G'$, and {\bf (ii)}~randomly choosing a neighbor of $v$ in $G_b$ as the target
  node to jump to.}
\label{fig:rwtvsa}
\end{figure}
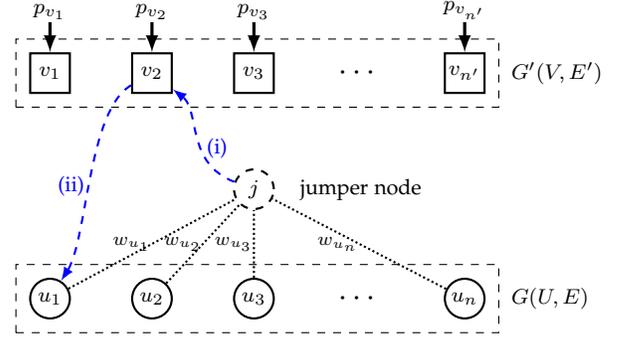

Similar to \VS{A}, we assume that a node $v$ in $G'$ can be sampled with
probability $p_v\propto a_v >0$.
Similar to the discussion of RWwJ in Section~\ref{sec:preliminaries}, in
\RW{T}\VS{A}, we virtually connect each node $u\in U$ to a jumper node $j$ with
edge $(u,j)$, and assign a weight $w_u$ for edge $(u,j)$.
The main challenge in designing \RW{T}\VS{A} is to determine the edge weights
$\{w_u\}_{u\in U}$.
With proper edge weights assignment, we can guarantee the {\em time
  reversibility}\footnote{A Markov chain is said to be time reversible with respect
  to $\pi$ if it satisfies condition $\pi_ip_{ij}=\pi_jp_{ji}, \forall i,j$.}
of random walks, which can facilitate us to determine the stationary probability of
a random walk visiting a node on target graph, and also simplify the estimator
design.
The following theorem states our main result on edge weights assignment.

\begin{theorem}\label{thm:rwtvsa-w}
  If we assign the edge weights $\{w_u\}_{u\in U}$ by
  \begin{equation}\label{eq:wu_rwtvsa}
    w_u=\alpha\sum_{v\in V_u}\frac{a_v}{d_v^{(b)}},\ u\in U
  \end{equation}
  for any $\alpha \geq 0$, then the random walk in \RW{T}\VS{A} is time
  reversible, and the stationary probability of the random walk visiting node
  $u\in U$ satisfies $\pi_u\propto d_u+w_u$, where $d_u$ is the degree of $u$ in
  $G$.
\end{theorem}
\begin{proof}
  If the random walk is time reversible, the stationary probabilities of visiting
  $u$ and $j$ are
  \[
    \pi_u = \frac{d_u+w_u}{2|E|+2\sum_uw_u} \quad \text{and} \quad
    \pi_j = \frac{\sum_uw_u}{2|E|+2\sum_uw_u}.
  \]
  Because for any $w_u\geq 0$, it always holds that
  \[
  \pi_up_{uu'} = \pi_{u'}p_{u'u} = \frac{1}{2|E|+2\sum_uw_u}, \forall (u,u')\in E.
  \]
  That is, the random walk is always time reversible along the transitions in $E$.
  We only need to prove that with the $w_u$ given by Theorem~\ref{thm:rwtvsa-w},
  the random walk is also time reversible along the transitions $(u,j)$ and
  $(j,u)$, i.e., $\pi_up_{uj}=\pi_jp_{ju}$.

  The walker residing at node $u$ moves to $j$ to perform an indirect jump with
  probability $p_{uj} = w_u/(d_u+w_u)$.
  Because an indirect jump is performed by first sampling a node $v$ in $G'$, and
  then choosing a neighbor $u$ of $v$ uniformly at random.
  Thus, the walker jumps from $j$ to $u$ with probability
  \begin{equation}\label{eq:qu}
    p_{ju}=c\sum_{v\in V_u}\frac{a_v}{d_v^{(b)}}\triangleq cb_u
  \end{equation}
  where $c$ is a constant.
  When $w_u=\alpha b_u$, so $\sum_uw_u=\alpha/c$, it indeed holds that
  \[
    \pi_up_{uj} = \pi_jp_{ju} = \frac{w_u}{2|E|+2\alpha/c},\ \forall u\in U.
  \]
  This demonstrates that when $w_u=\alpha b_u$, the random walk is time
  reversible, and the stationary probability of visiting $u$ satisfies
  $\pi_u\propto d_u+w_u$.
\end{proof}

Note that if $d_u^{(b)}=0$, then $p_{uj}=p_{ju}=0$, i.e., the walker does not jump
from/to $u$; the worker just moves from/to $u$ to/from a neighbor of $u$.
Hence, $u$ can still be sampled by the random walk.
$\alpha$ controls the probability of conducting a jump on a node.
If $\alpha=0$, \RW{T}\VS{A} does not perform jumps, and it actually becomes a
simple random walk on the target graph; if $\alpha\rightarrow\infty$, \RW{T}\VS{A}
is equivalent to \VS{A}-I.
Thus, \RW{T}\VS{A} behaves similarly as RWwJ.

\header{Sampling design.}
Suppose the random walk starts at node $x_1\in U$, and at step $i$ the random walk
is at node $x_i$.
We calculate the probability of jumping $w_{x_i}$ by Eq.~\eqref{eq:wu_rwtvsa}.
At step $i$, the walker jumps with probability $w_{x_i}/(d_{x_i} + w_{x_i})$;
otherwise, the walker moves to a neighbor $u$ of $x_i$ chosen uniformly at random
and let $x_{i+1}=u$.
An indirect jump is performed as follows:
\begin{itemize}
\item sample a node $v\in V$ in the auxiliary graph;
\item sample a neighbor $u$ of $v$ in $G_b$ uniformly at random, and let
  $x_{i+1}=u$.
\end{itemize}

\header{Estimator.}
Using the collected samples, denoted by $S=(x_i,\ldots,x_B)$, we propose to
estimate $\theta$ by
\begin{equation}\label{eq:est_rwtvsa}
  \hat\theta^{\RW{T}\VS{A}}
  = \frac{1}{Z^{\RW{T}\VS{A}}}\sum_{u\in S}\frac{f(u)}{d_u+w_u},
\end{equation}
where $Z^{\RW{T}\VS{A}}\triangleq\sum_{u\in S} 1/(d_u+w_u)$.

\begin{theorem}\label{th:rwtvsa}
  Estimator~\eqref{eq:est_rwtvsa} is asymptotically unbiased.
\end{theorem}
\begin{proof}
  Since $\pi_u\propto d_u+w_u$, then
  \[
    \E[\pi]{\frac{f(u)}{d_u+w_u}}
    = \sum_{u\in U}\pi_u\frac{f(u)}{d_u+w_u}
    = cn\theta.
  \]
  Similarly,
  \[
    \E[\pi]{\frac{1}{d_u+w_u}}
    = \sum_{u\in U}\pi_u\frac{1}{d_u+w_u}
    = cn.
  \]
  By Lemma~\eqref{lemma:lln}, we obtain
  \[
    \lim_{B\rightarrow\infty}\hat\theta^{\RW{T}\VS{A}}
    =\frac{\E[\pi]{f(u)/(d_u+w_u)}}{\E[\pi]{1/(d_u+w_u)}}
    =\theta \quad a.s.
  \]
\end{proof}

\header{Remark.}
Note that \RW{T}\VS{A} requires vertex sampling (e.g., UNI) on the auxiliary graph
$G'$.
If vertex sampling is also not allowed on $G'$, \RW{T}\VS{A} is inapplicable.
However, one can replace the vertex sampling on $G'$ by a random walk on $G'$.
Unfortunately, this naive approach can perform very poorly when the auxiliary
graph $G'$ is not well connected, because a poorly connected graph can easily trap
a simple random walk in a community.
In what follows, we design a third method to address this challenge.

\subsection{Random Walk on Target Graph Incorporating with Random
  Walk on Auxiliary Graph (\RW{T}\RW{A})}

When both the target and auxiliary graphs do not support vertex sampling, neither
\VS{A} nor \RW{T}\VS{A} is applicable.
Therefore, we design the \RW{T}\RW{A} method to address this challenge.
\RW{T}\RW{A} consists of two parallel random walks on $G$ and $G'$ respectively.
The two random walks cooperate with each other, and can be viewed as two RWwJs, as
illustrated in Fig.~\ref{fig:rwtrwa}.
Unlike \RW{T}\VS{A} where only nodes in $G$ are virtually connected to a jumper
node, in \RW{T}\RW{A}, nodes in both $G$ and $G'$ are virtually connected to two
jumper nodes $j$ and $j'$ with weights $\{w_u\}_{u\in U}$ and $\{w_v\}_{v\in V}$
to perform indirect jumps on $G$ and $G'$ respectively.

\begin{figure}[htp]
\centering
\begin{tikzpicture}[
  every node/.style={inner sep=1pt, minimum size=15pt, font=\footnotesize},
  und/.style={draw, thick, circle},
  vnd/.style={draw, thick, rectangle},
  att/.style={left, minimum size=0},
  arr/.style={thick, -latex, blue, dashed},
  jumper/.style={und, dashed},
  box/.style={draw, dashed, inner sep=5pt},
  ]

  \node(u1) [und] at (0,0) {$u_1$};
  \node(u2) [und, right=.8 of u1] {$u_2$};
  \node(u3) [und, right=.8 of u2] {$u_3$};
  \node(p1) [att, right=.8 of u3] {\large $\cdots$};
  \node(un) [und, right=.8 of p1] {$u_n$};

  \node(j1) [jumper,above left=1.08 and 0.1 of u3] {$j$};

  \node[box, fit={(u1) (un)}] {};

  \node[att, right=3mm of un] {$G(U,E)$};

  \foreach \v in {1,2,3,n}{
    \draw[thick, densely dotted] (j1) to
      node[att] {\scriptsize $w_{u_{\v}}$} (u\v);
  }

  \node(v1) [vnd] at (0,3) {$v_1$};
  \node(v2) [vnd, right=.8 of v1] {$v_2$};
  \node(v3) [vnd, right=.8 of v2] {$v_3$};
  \node(p2) [att, right=.8 of v3] {\large $\cdots$};
  \node(vnp) [vnd, right=.8 of p2] {$v_{n'}$};

  \node(j2) [jumper, right=0.5 of j1] {$j'$};

  \node[box, fit={(v1) (vnp)}] {};

  \node[att, right=3mm of vnp] {$G'(V,E')$};
  \node[att, right=2mm of j2] {jumper nodes};

  \foreach \v/\a in {1/1,2/2,3/3,np/n'} {
    \draw[thick, densely dotted] (j2) to
      node[att] {\scriptsize $w_{v_{\a}}$} (v\v);
  }

  \draw[arr] (j1) .. controls (1.5,2) and (0.5,2) .. (v1);
  \draw[arr] (v1) .. controls (-.3,2) and (.3,1) .. (u1);
  \draw[arr] (j2) .. controls (3.8,0.8) and (5,0.8) .. (un);
  \draw[arr] (un) .. controls (6,1.5) .. (vnp);
\end{tikzpicture}

\caption{Illustration of \RW{T}\RW{A} and indirect jumps}
\label{fig:rwtrwa}
\end{figure}
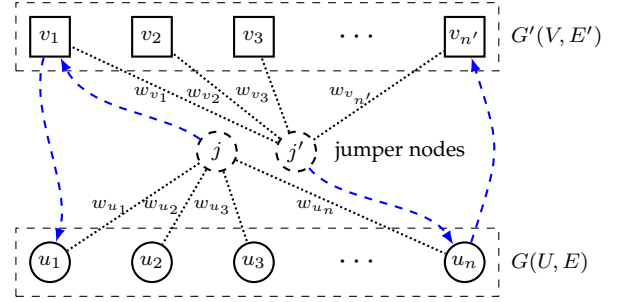

The basic idea behind \RW{T}\RW{A} is as follows.
Suppose the two random walks are $RW$ on $G$ and $RW'$ on $G'$, and at step $i$,
they reside at $x_i\in U$ and $y_i\in V$, respectively.
If one random walk needs to jump at step $i$, say $RW$ on $G$, then it jumps to a
uniformly at random chosen neighbor of $y_i$ in the bipartite graph, which is
assigned to $x_{i+1}$.
Similar jumping procedure also applies to $RW'$ on $G'$.
Hence, they are analogous to two RWwJs, and both can avoid being trapped on $G$
and $G'$.

Similar to \RW{T}\VS{A}, the main challenge in designing \RW{T}\RW{A} is to
determine edge weights $\{w_u\}_{u\in U}$ and $\{w_v\}_{v\in V}$, which control
the probability of jumping of the two random walks.
Obviously, the stationary distributions $\{\pi_u\}_{u\in U}$ and $\{\pi_v\}_{v\in
  V}$ of the two random walks are also related to these weights.
Here we leverage our previous analysis of \RW{T}\VS{A}, and derive that, when
parameters $w_u$ and $w_v$ satisfy the following conditions
\begin{equation}\label{eq:w}
  w_u = \alpha\sum_{v\in V_u}\frac{\pi_v}{d_v^{(b)}}, u\in U,\quad
  w_v = \beta\sum_{u\in U_v}\frac{\pi_u}{d_u^{(b)}}, v\in V,
\end{equation}
for any $\alpha, \beta>0$, the stationary distributions of the two random walks on
$G$ and $G'$ (discarding states $j$ and $j'$) are
\begin{equation}\label{eq:pi}
  \pi_u = \frac{d_u+w_u}{2|E|+\alpha}, u\in U, \quad
  \pi_v = \frac{d_v+w_v}{2|E'|+\beta}, v\in V.
\end{equation}
The matrix forms of Eqs.~\eqref{eq:w}--\eqref{eq:pi} yield
\begin{align}
  w_U &=\alpha AD_V^{-1}\pi_V,
  &w_V &= \beta A^TD_U^{-1}\pi_U, \label{eqs:W} \\
  \pi_U &= \frac{d_U+w_U}{2|E|+\alpha},
  &\pi_V &= \frac{d_V+w_V}{2|E'|+\beta}, \label{eqs:PI}
\end{align}
where $A_{n\times n'}$ is the adjacency matrix of $G_b$, $w_U=[w_u]_{u\in U}^T$,
$w_V=[w_v]_{v\in V}^T$, $\pi_U=[\pi_u]_{u\in U}^T$, $\pi_V=[\pi_v]_{v\in
  V}^T$, $d_U=[d_u]_{u\in U}^T$ and $d_V=[d_v]_{v\in V}^T$ are vectors,
$D_U=diag(d_{u_1}^{(b)}, \ldots,d_{u_n}^{(b)})$ and
$D_V=diag(d_{v_1}^{(b)},\ldots,d_{v_{n'}}^{(b)})$ are diagonal matrices.

Equations~\eqref{eqs:W}--\eqref{eqs:PI} uniquely determine $w_U$ and $w_V$,
i.e.,
\begin{align*}
  w_U^* &= c(I-cc'AD_V^{-1}A^TD_U^{-1})^{-1} AD_V^{-1}(d_V+c'A^TD_U^{-1}d_U) \\
  w_V^* &= c'(I-cc'A^TD_U^{-1}AD_V^{-1})^{-1} A^TD_U^{-1}(d_U+cAD_V^{-1}d_V)
\end{align*}
where $c=\alpha/(2|E'|+\beta)$ and $c'=\beta/(2|E|+\alpha)$ are constants.

The above results illustrate that, when $\alpha$ and $\beta$ are given, $w_U$ and
$w_V$ are uniquely determined.
However, one needs complete knowledge of $G$, $G'$ and $G_b$ to determine their
values.
In graph sampling, we are interested in methods without having to know the
complete graph structure in advance.
In what follows, we design \RW{T}\RW{A} in a way that only makes use of {\em local
  knowledge} of these graphs.

In general, if $w_U\neq w_U^*$ (or $w_V\neq w_V^*$), the random walks on the two
modified graphs are no longer timer reversible, and
Eqs.~\eqref{eqs:W}--\eqref{eqs:PI} do not hold.
There is another way to understand why they do not hold, and this understanding
could motivate us to propose a solution.
Variables in Eqs.~\eqref{eqs:W}--\eqref{eqs:PI} form dependent relations, as
illustrated in Fig.~\ref{fig:dependence}.
Given $w_U$, we can obtain $\pi_U$ (from the first equation of~\eqref{eqs:PI}),
and then obtain $w_V$ (from the second equation of~\eqref{eqs:W}), and finally
obtain $w_U'$ (from the first equation of~\eqref{eqs:W}).
If $w_U=w_U^*$, then $w_U'=w_U^*$; otherwise, $w_U'\neq w_U\neq w_U^*$, and this
forms a contradiction.

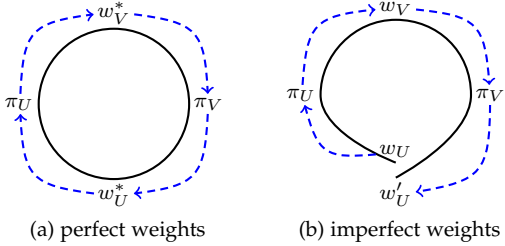
\begin{figure}[htp]
\centering
\subfloat[perfect weights]{\begin{tikzpicture}[
  thick,smooth cycle,
  arr/.style={densely dashed, ->, blue},
  att/.style={minimum size=0, inner sep=1pt, font=\footnotesize},
  ]

  \node[circle,draw,thick, minimum size=2cm] (cir) at (0,0) {};
  \node[att,below=0 of cir]  (wu) {$w_U^*$};
  \node[att,left=0 of cir] (pu) {$\pi_U$};
  \node[att,above=0 of cir] (wv) {$w_V^*$};
  \node[att,right=0 of cir] (pv) {$\pi_V$};
  \draw [arr] (wu) edge[bend left=40,looseness=1.5] (pu);
  \draw [arr] (pu) edge[bend left=40,looseness=1.5] (wv);
  \draw [arr] (wv) edge[bend left=40,looseness=1.5] (pv);
  \draw [arr] (pv) edge[bend left=40,looseness=1.5] (wu);
\end{tikzpicture}

}\qquad
\subfloat[imperfect weights]{\begin{tikzpicture}[
  thick, smooth,
  arr/.style={densely dashed, ->, blue},
  att/.style={minimum size=0, inner sep=1pt, font=\footnotesize},
  ]

  \node[circle,minimum size=2cm] (cir) at (0,0) {};
  \draw plot[tension=1] coordinates {(0,-0.9) (-1,0) (0,1) (1,0) (0,-1.1)};
  \node[att,below=-.4 of cir] (wu) {$w_U$};
  \node[att,below=.1 of cir] (wup) {$w_U'$};
  \node[att,left=0 of cir] (pu) {$\pi_U$};
  \node[att,above=0 of cir] (wv) {$w_V$};
  \node[att,right=0 of cir] (pv) {$\pi_V$};
  \draw[arr] (wu) edge[bend left=40,looseness=1.5] (pu);
  \draw[arr] (pu) edge[bend left=40,looseness=1.5] (wv);
  \draw[arr] (wv) edge[bend left=40,looseness=1.5] (pv);
  \draw[arr] (pv) edge[bend left=40,looseness=1.5] (wup);
\end{tikzpicture}

}
\caption{Dependent relations among variables.
  The variable at the head of an arrow depends on the variable at the tail of the
  arrow.}
\label{fig:dependence}
\end{figure}

We find that this contradiction has a physical meaning, and it is fixable.
The normalized weights $w_U$ can be viewed as a distribution that describes the
probability a walker jumping to a node in $G$.
When we specify some particular weights $w_U$, it means that we expect the walker
to jump to a node in $G$ following a distribution specified by $w_U$.
If $w_U\neq w_U^*$, we will derive a different $w_U'$ using
Eqs.~\eqref{eqs:W}--\eqref{eqs:PI}.
It means that the walker actually jumps to a node in $G$ following a different
distribution specified by $w_U'$.
This is the reason why the random walk is not time reversible.
Fortunately, with this understanding, the contradiction becomes fixable by
applying the famous Metropolis-Hastings (MH) sampler~\cite{Robert2004}.
We can treat (normalized) $w_U$ as the {\em desired distribution}, and
(normalized) $w_U'$ as the {\em proposal distribution}, and we use a MH sampler to
build a Markov chain (referred as the MH chain) that generates samples with the
desired distribution.
Each time when the walker requires jumping, it jumps to a node generated by the MH
chain.
This guarantees that the walker jumps to nodes in $G$ following the desired
distribution, and ensures that $\pi_U$ and $\pi_V$ are still the stationary
distributions of the random walks.

\header{Sampling design.}
The complete design of \RW{T}\RW{A} comprises three parallel Markov chains as
illustrated in Fig.~\ref{fig:mcs}, and we need to specify desired weights
$w_U$ in advance, e.g., from a uniform distribution.

\begin{figure}[htp]
\centering
\footnotesize
\begin{tikzpicture}[
  every node/.style={inner sep=0, minimum size=18pt, font=\footnotesize},
  und/.style={draw,thick,circle},
  vnd/.style={draw,thick,rectangle, minimum size=16pt},
  att/.style={left,minimum size=0},
  arr/.style={thick, -latex},
  ]

  \node[und] (u1) at (0,0) {$x_1$};
  \node[att,right=.5 of u1] (p1)  {\large$\cdots$};
  \node[und,right=.5 of p1] (ui) {$x_i$};
  \node[und,right=.5 of ui] (ui1) {$x_{i+\!1}$};
  \node[att,right=.5 of ui1] (p2) {\large$\cdots$};
  \draw[arr] (u1)--(p1);
  \draw[arr] (p1)--(ui);
  \draw[arr] (ui)--(ui1);
  \draw[arr] (ui1)--(p2);
  \node[att,left=.3 of u1] {RW on $G$:};

  \node[und,above=.2 of u1] (up1) {$x'_1$};
  \node[att,right=.5 of up1] (pp1) {\large$\cdots$};
  \node[und,right=.5 of pp1] (upi) {$x'_i$};
  \node[und,right=.5 of upi] (upi1) {$x'_{i+\!1}$};
  \node[att,right=.5 of upi1] (pp2) {\large$\cdots$};
  \draw[arr] (up1)--(pp1);
  \draw[arr] (pp1)--(upi);
  \draw[arr] (upi)--(upi1);
  \draw[arr] (upi1)--(pp2);
  \node[att,left=.3 of up1] {MH chain:};

  \node[vnd,above=.2 of up1] (v1) {$y_1$};
  \node[att,right=.55 of v1] (pv1) {\large$\cdots$};
  \node[vnd,right=.55 of pv1] (vi) {$y_i$};
  \node[vnd,right=.55 of vi] (vi1) {$y_{i+\!1}$};
  \node[att,right=.55 of vi1] (pv2) {\large$\cdots$};
  \draw[arr] (v1)--(pv1);
  \draw[arr] (pv1)--(vi);
  \draw[arr] (vi)--(vi1);
  \draw[arr] (vi1)--(pv2);
  \node[att,left=.3 of v1] {RW on $G'$:};
\end{tikzpicture}

\caption{Three parallel Markov chains in \RW{T}\RW{A}.}
\label{fig:mcs}
\end{figure}
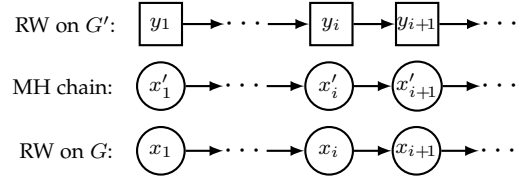

\bullethdr{Random walk on auxiliary graph $G'$:}
Suppose the random walk resides at node $y_i\in V$ at step $i$.
Then we can calculate $w_{y_i}$ according to Eq.~\eqref{eq:w}.
At step $i+1$, the random walk executes one of the following two steps.
\begin{description}[\emph{Walk}]
\item[\emph{Jump}:] With probability $w_{y_i}/(d_{y_i}+w_{y_i})$, the walker jumps
  to a random neighbor $v\in V$ of node $x_i$ in $G_b$, and $y_{i+1}=v$;
\item[\emph{Walk}:] Otherwise, the walker moves to a random neighbor $v\in V$ of
  $y_i$ in $G'$, and $y_{i+1}=v$.
\end{description}

\bullethdr{MH chain:}
Suppose the MH chain resides at node $x'_i$ at step $i$.
At step $i+1$, we randomly choose a neighbor $u\in U$ of $y_i$ in $G_b$. This is
equivalent to sample a node $u\in U$ with probability proportional to $w_u'$.
\begin{description}[\emph{Acceptance}]
\item[\emph{Acceptance}:] With probability $r_i$, we accept $u$ and $x'_{i+1}=u$,
  where $r_i=\min\{1,(w_uw_{x'_i}')/(w_{x'_i}w_u')\}$;
\item[\emph{Rejection}:] Otherwise, we reject $u$ and $x'_{i+1}=x'_i$.
\end{description}

\bullethdr{Random walk on target graph $G$:}
Suppose the random walk resides at node $x_i\in U$ at step $i$.
At step $i+1$, the walker executes one of the following two steps.
\begin{description}[\emph{Jump}]
\item[\emph{Jump}:] With probability $w_{x_i}/(d_{x_i}+w_{x_i})$, the walker jumps
  to $x'_{i+1}$, and $x_{i+1}=x'_{i+1}$;
\item[\emph{Walk}:] Otherwise, the walker moves to a random neighbor $u\in U$ of
  $x_i$ in $G$, and $x_{i+1}=u$.
\end{description}

This sampling design ensures that we use only local knowledge of the three graphs
to obtain a sample path $S=(x_1,\ldots,x_B)$, which can yield unbiased estimate of
$\theta$.

\header{Estimator.}
Given the sample path $S=(x_1,\ldots,x_B)$, we propose to use the following
estimator to estimate $\theta$.
\begin{equation}\label{eq:est_rwtrwa}
  \hat\theta^{\RW{T}\RW{A}}
  =\frac{1}{Z^{\RW{T}\RW{A}}}\sum_{u\in S}\frac{f(u)}{d_u+w_u},
\end{equation}
where $Z^{\RW{T}\RW{A}}\triangleq\sum_{u\in S}1/(d_u+w_u)$.

\begin{theorem}
  Estimator~\eqref{eq:est_rwtrwa} is asymptotically unbiased.
\end{theorem}
\begin{proof}
  Since we have constructed the Markov chain on $G$ with stationary distribution
  $\pi_u\propto d_u + w_u$, the proof is exactly the same as
  Theorem~\ref{th:rwtvsa}.
\end{proof}

\section{Experiments}
\label{sec:experiment}

In this section, we conduct experiments on both synthetic and real datasets to
validate our sampling designs.
Our goal is to demonstrate the unbiasedness of proposed estimators
(\eqref{eq:est_vsa1}, \eqref{eq:est_vsa2}, \eqref{eq:est_rwtvsa},
and~\eqref{eq:est_rwtrwa}) and study their estimation errors with respect to
different factors such as sampling budget $B$ and parameter settings $\alpha$ and
$\beta$.

We consider to estimate the PDF and CCDF of degree distribution of a graph.
For PDF, the characteristic function is defined as $f_d(u)\triangleq\indr{d_u=d}$,
where $\indr{\cdot}$ is the indicator function, and the graph characteristic is
the distribution $\{\theta_d\}_{d\geq 0}$ where $\theta_d=\sum_u f_d(u)/n$ is the
fraction of nodes with degree $d$ in graph $G$.
For CCDF, the characteristic function is defined as
$f_d(u)\triangleq\indr{d_u>d}$, and the graph characteristic is the distribution
$\{\theta_d\}_{d\geq 0}$ where $\theta_d=\sum_u f_d(u)/n$ is the fraction of nodes
with degree larger than $d$ in graph $G$.
In some experiments, we will only show the results of estimating CCDF due to
space limitation.

\subsection{Experiments on Synthetic Data}

In the first experiment, we validate the sampling methods using synthetic data.

\header{Synthetic data.}
We generate a two-layered network structure by connecting three
Barab\'{a}si-Albert (BA) graphs~\cite{Barabasi1999} $G_1, G_2$ and $G_3$.
Each BA graph contains 100,000 nodes, and the three BA graphs have average degree
$4$, $10$ and $20$, respectively.
$G_1$ and $G_3$ are connected by one edge to form the target graph $G$, which thus
has a barbell structure.
$G_2$ is the auxiliary graph $G'$, and the bipartite graph $G_b$ is formed by
connecting nodes in $G$ and $G'$ according to the following two steps:
\begin{itemize}
\item connect every node in $G$ to a randomly selected node in $G'$;
\item randomly connect $200,000$ pairs of nodes, and each pair has one node in $G$
  and the other node in $G'$.
\end{itemize}

The first step ensures that every node in $U$ satisfies $d_u^{(b)}>0$ so that we
can apply \VS{A} on this dataset.

\header{Results and analysis.}
First we demonstrate that the proposed estimators $\hat\theta_d^{\VS{A}\text{-I}},
\hat\theta_d^{\VS{A}\text{-II}}, \hat\theta_d^{\RW{T}\VS{A}}$, and
$\hat\theta_d^{\RW{T}\RW{A}}$ are asymptotically unbiased.
To show this, we apply these sampling methods to estimate the fraction of nodes
with degree $2$ and $12$ in the target graph, denoted by $\theta_2$ and
$\theta_{12}$.
We compare their estimates to the ground truth for different sampling budgets $B$.
We also show the estimates using a simple random walk on the target graph.
Because the target graph has a barbell structure, the random walk is easily to be
trapped into one component and fail to explore the other component.
We expect to see that the random walk estimator does not perform well.
The results are depicted in Fig.~\ref{fig:unbias}.
Indeed, the random walk incurs large biases, and always overestimates $\theta_2$
and $\theta_{12}$.
In comparison, our proposed estimators can obtain more accurate estimates, and it
is clear to see that when sampling budget $B$ increases, all our proposed
estimators can converge to the ground truth.
Hence, these results demonstrate that our proposed estimators are asymptotically
unbiased.

\begin{figure}[tp]
\centering
\subfloat[$\hat\theta^{\VS{A}\text{-I}}_d$]{%
  \includegraphics[width=.5\linewidth]{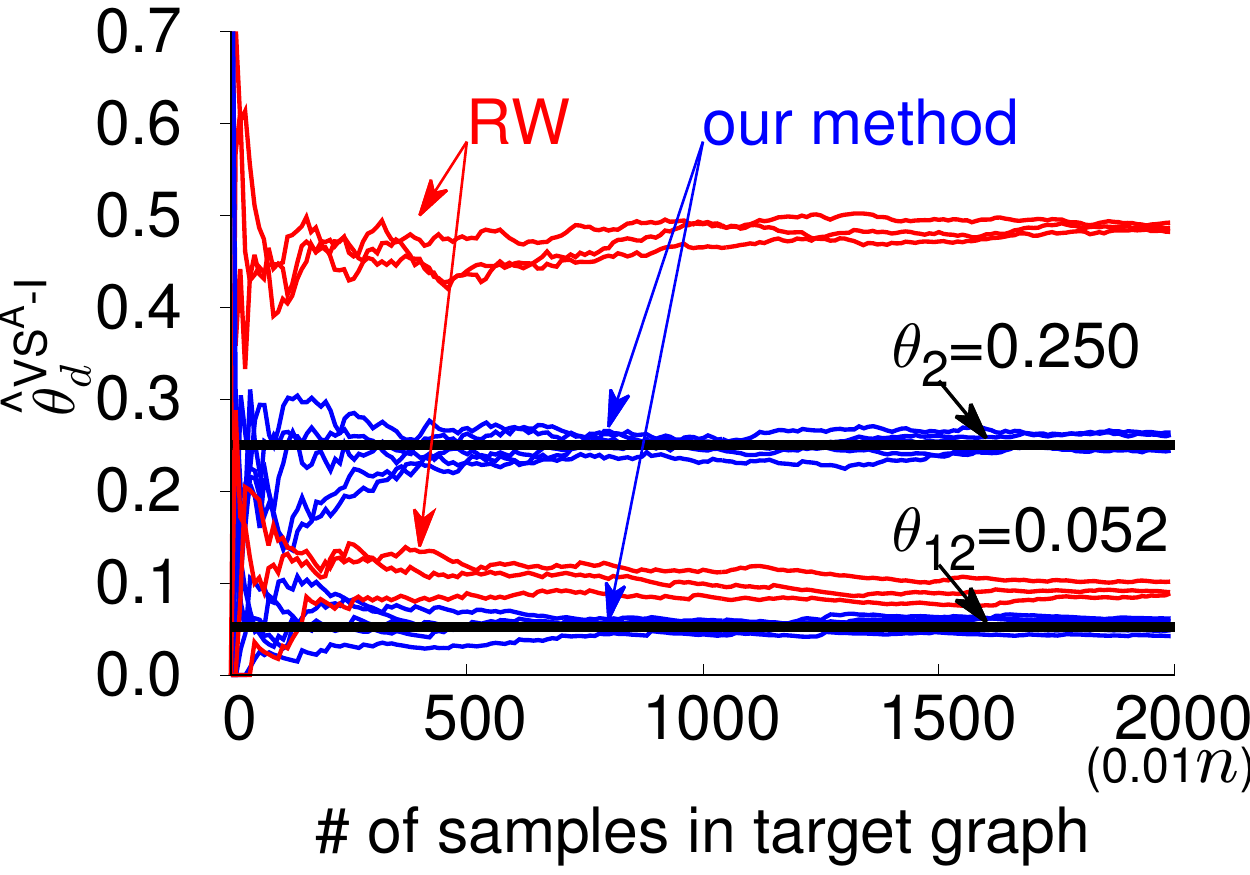}}
\subfloat[$\hat\theta^{\VS{A}\text{-II}}_d$]{%
  \includegraphics[width=.5\linewidth]{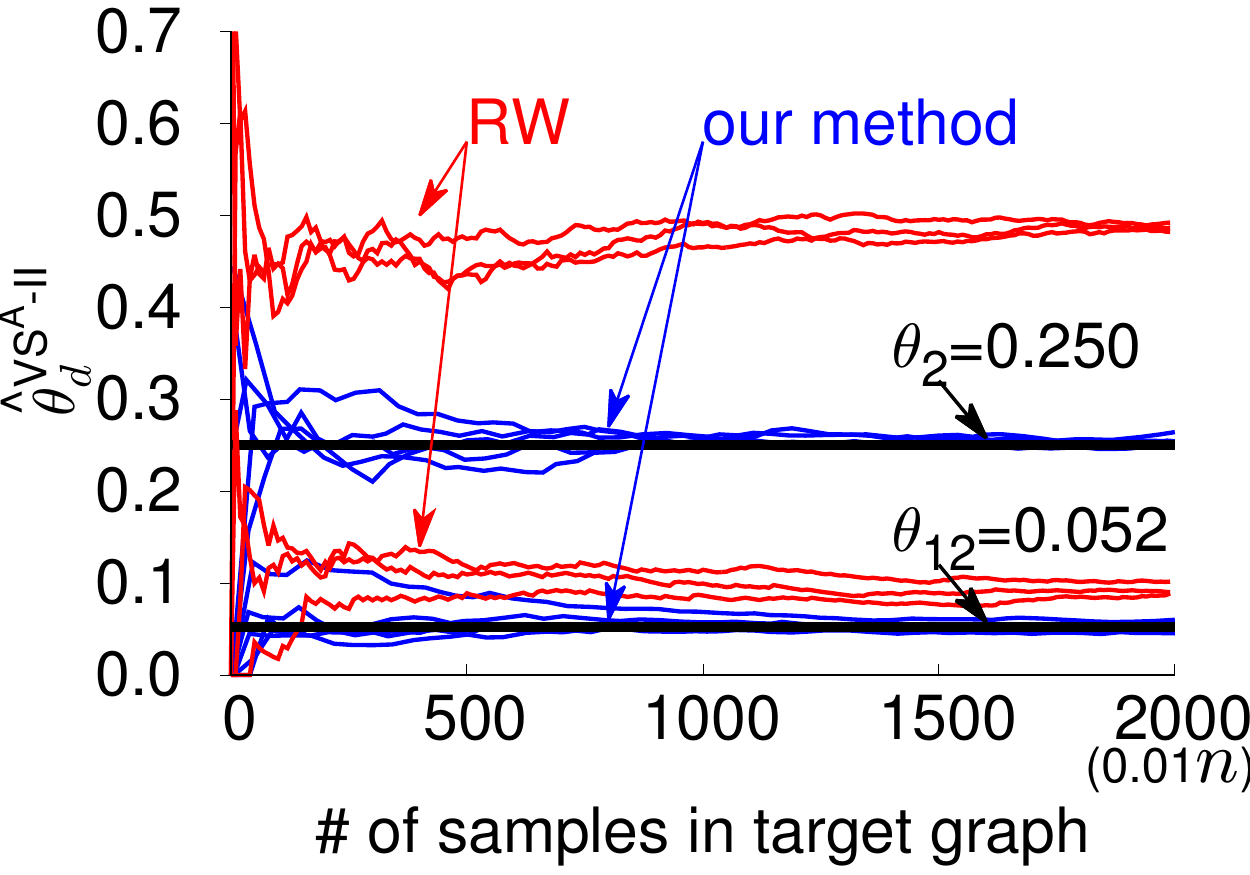}} \\
\subfloat[$\hat\theta^{\RW{T}\VS{A}}_d$ ($\alpha=10$)]{%
  \includegraphics[width=.5\linewidth]{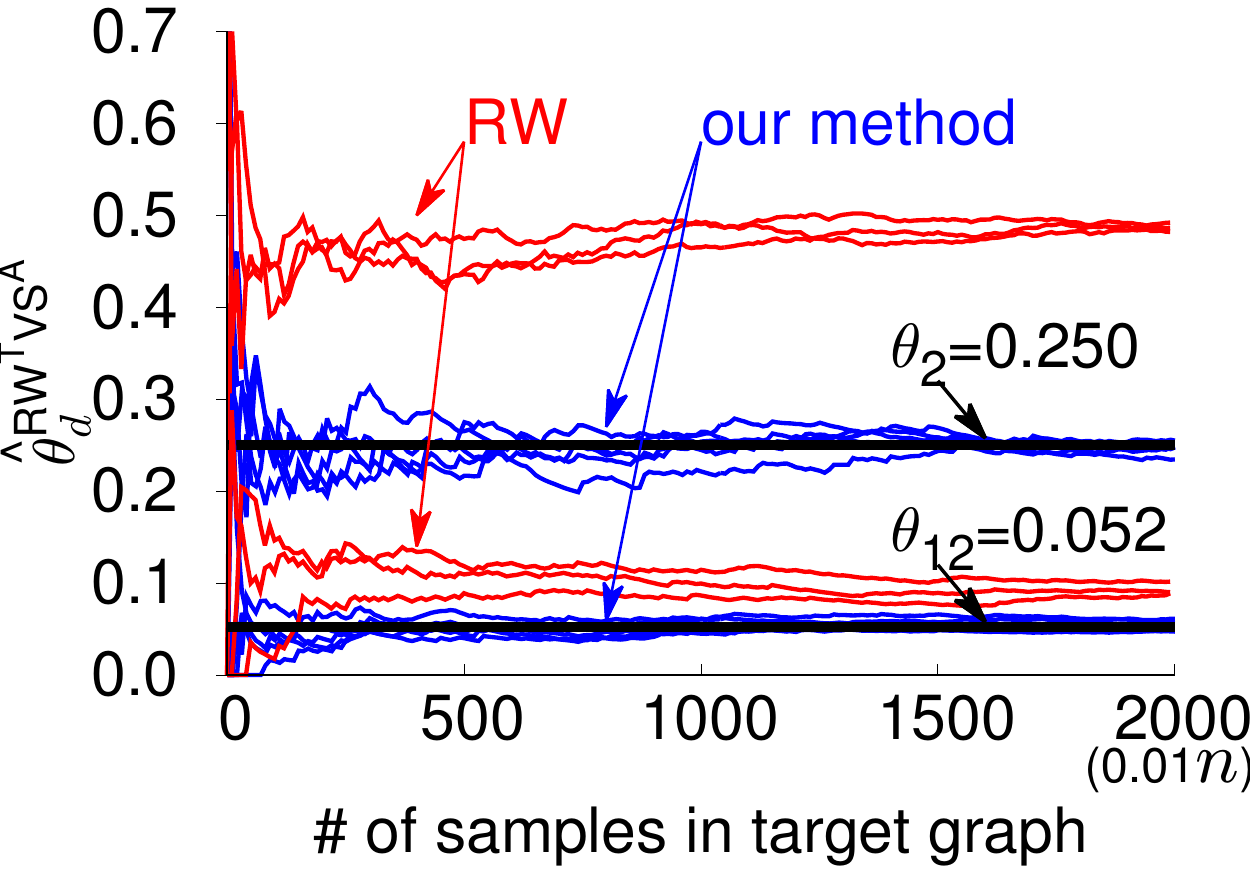}}
\subfloat[$\hat\theta^{\RW{T}\RW{A}}_d$ ($\alpha=\beta=10$)]{%
  \includegraphics[width=.5\linewidth]{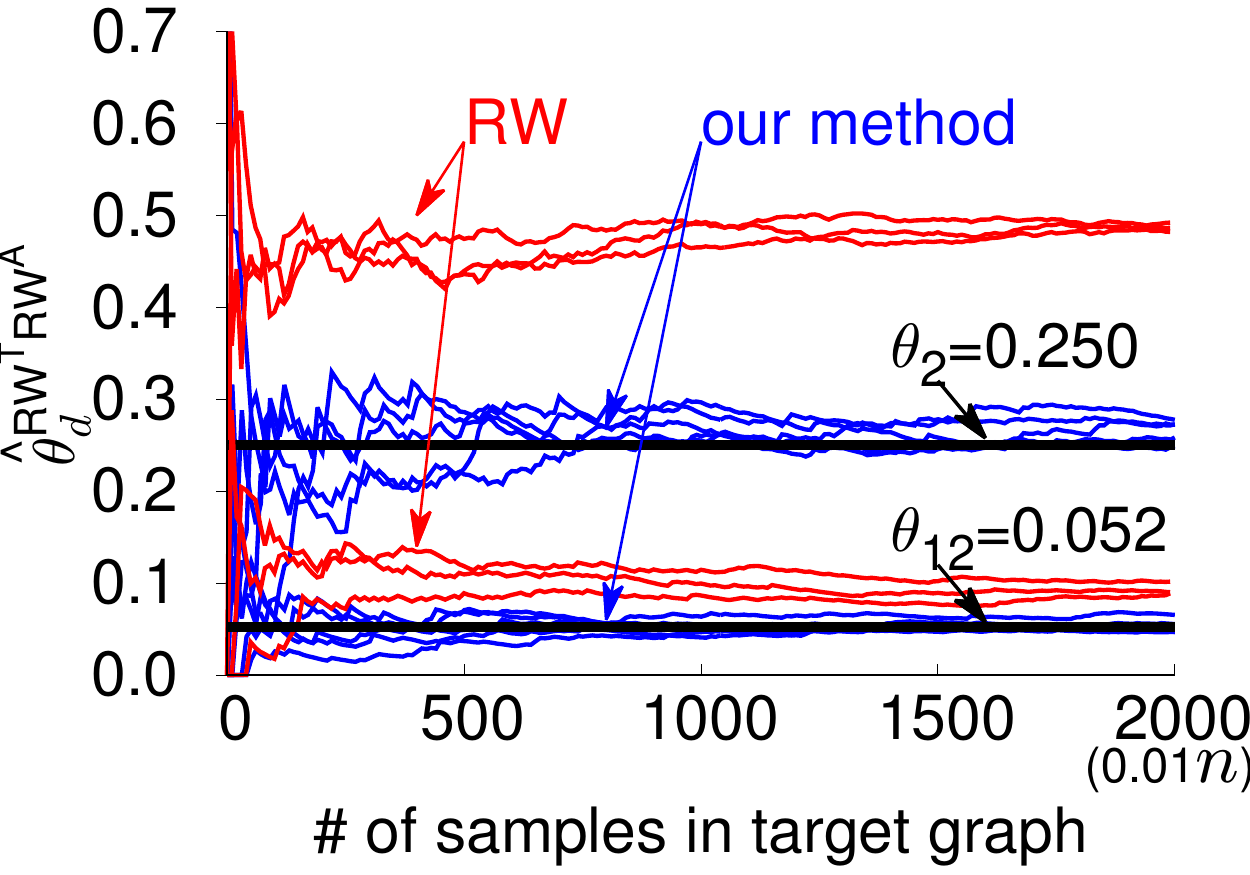}}
\caption{Asymptotic unbiasedness of estimators.
  \label{fig:unbias}}
\end{figure}

Next, we study the estimation error of each estimator for estimating the PDF and
CCDF of degree distribution.
We choose the \emph{normalized rooted mean squared error} (NRMSE) as a metric to
evaluate the estimation error of an estimator, which is defined as follows
\[
  \text{NRMSE}(\hat\theta)
  =\frac{\sqrt{\E{(\hat\theta-\theta)^2}}}{\theta}.
\]
NRMSE measures the relative difference between an estimated value $\hat\theta$ and
a real value $\theta$.
The smaller the NRMSE, the more accurate the estimator $\hat\theta$ is.
To compare the NRMSE of different estimators, we fix the sampling budget $B$ to be
$1\%$ of the target graph size, and calculate the averaged empirical NRMSE over
$1,000$ runs.
The results are depicted in Figs.~\ref{fig:syn_nrmse_pdf}
and~\ref{fig:syn_nrmse_ccdf}.

\begin{figure}[tp]
  \centering
  \subfloat[\VS{A}-I]{\includegraphics[width=.5\linewidth]{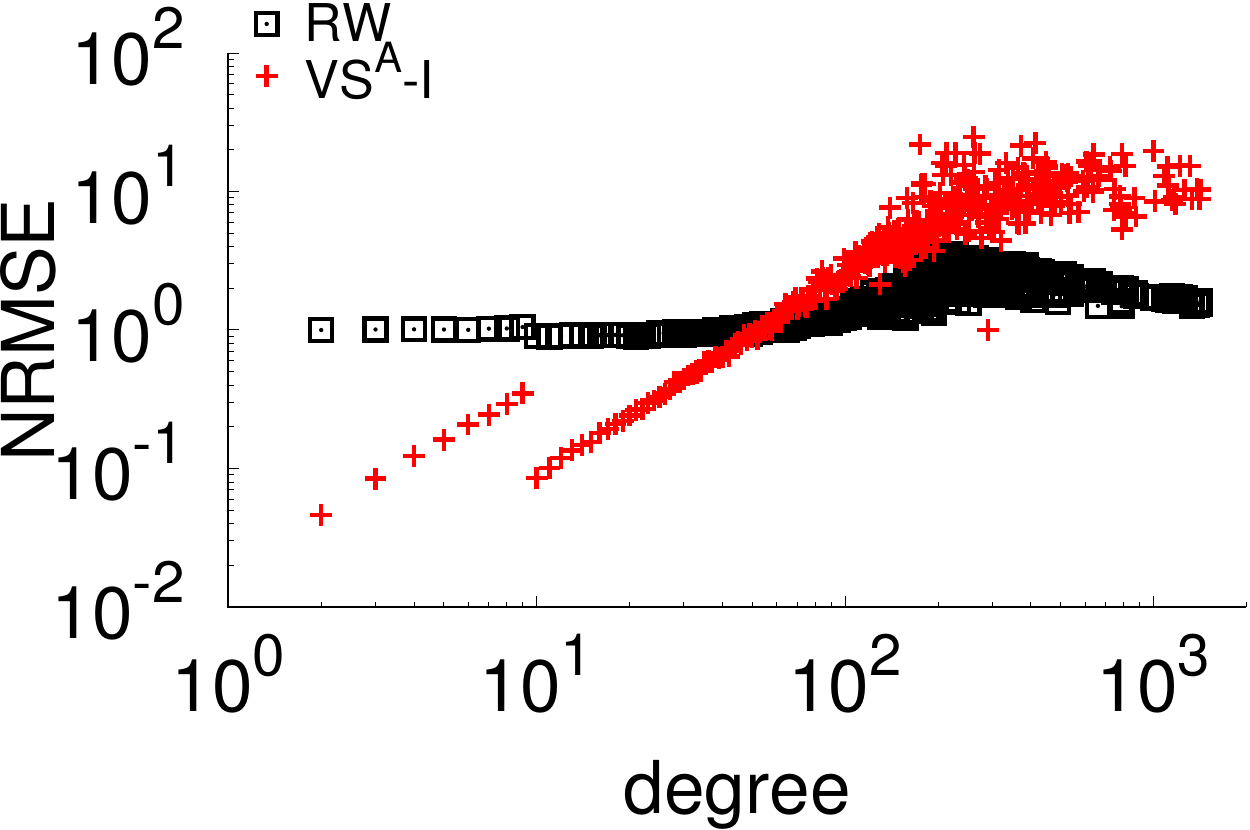}}
  \subfloat[\VS{A}-II]{\includegraphics[width=.5\linewidth]{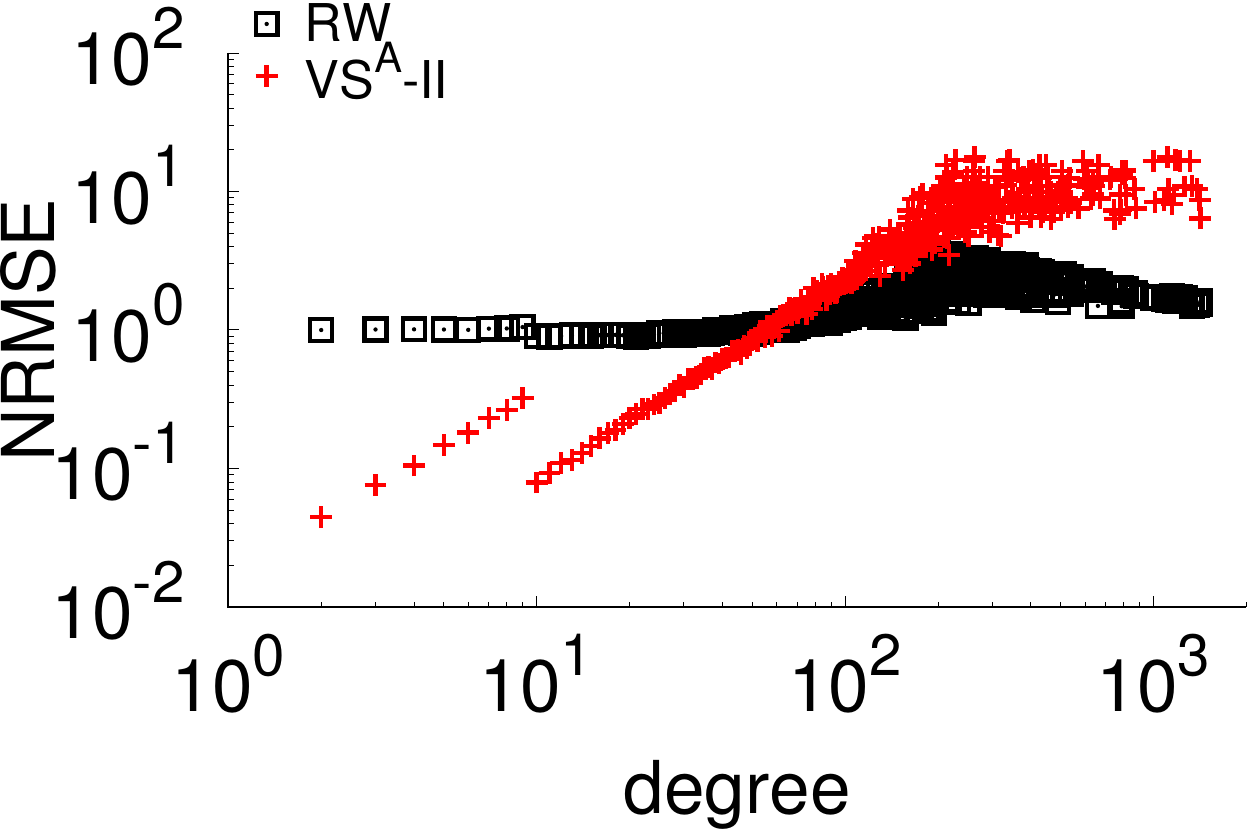}} \\
  \subfloat[\RW{T}\VS{A}]{\includegraphics[width=.5\linewidth]{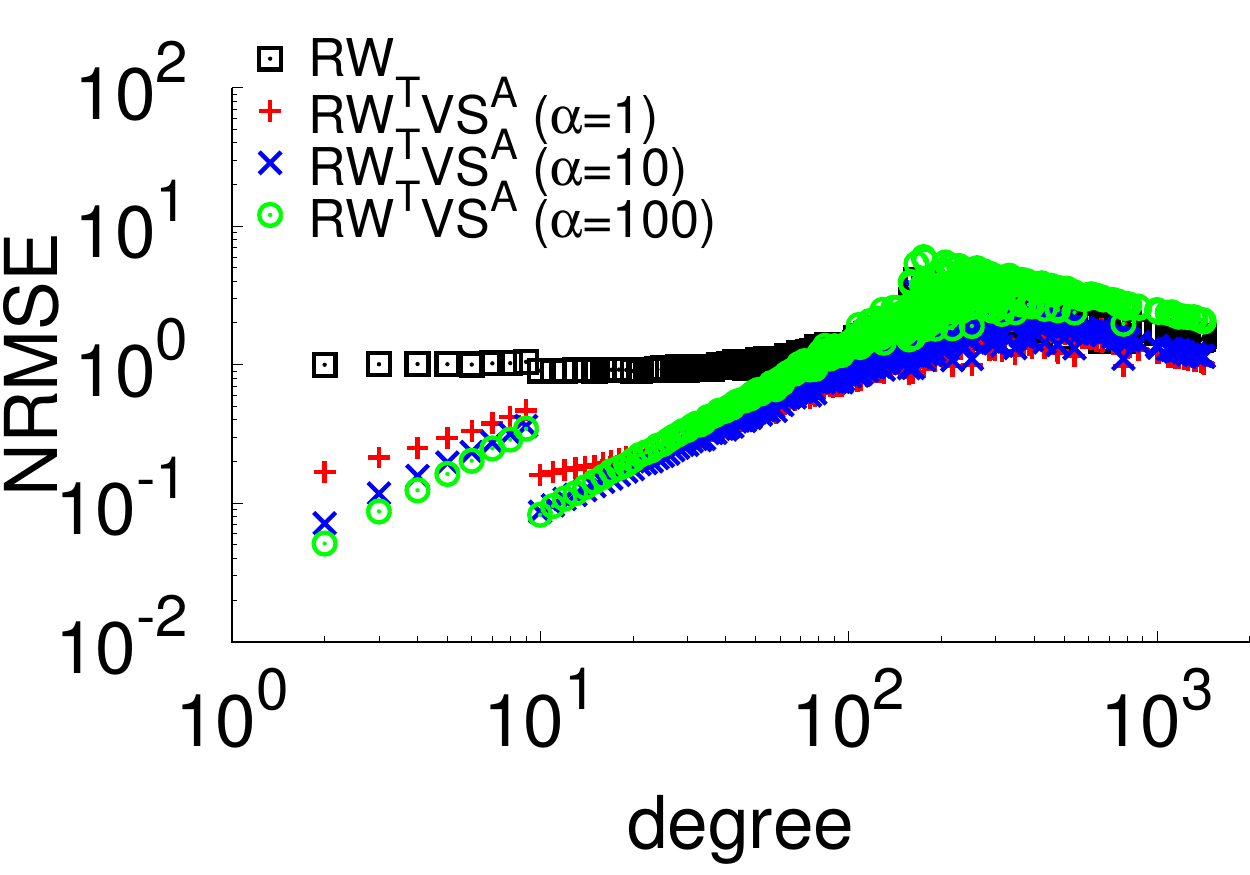}}
  \subfloat[\RW{T}\RW{A}]{\includegraphics[width=.5\linewidth]{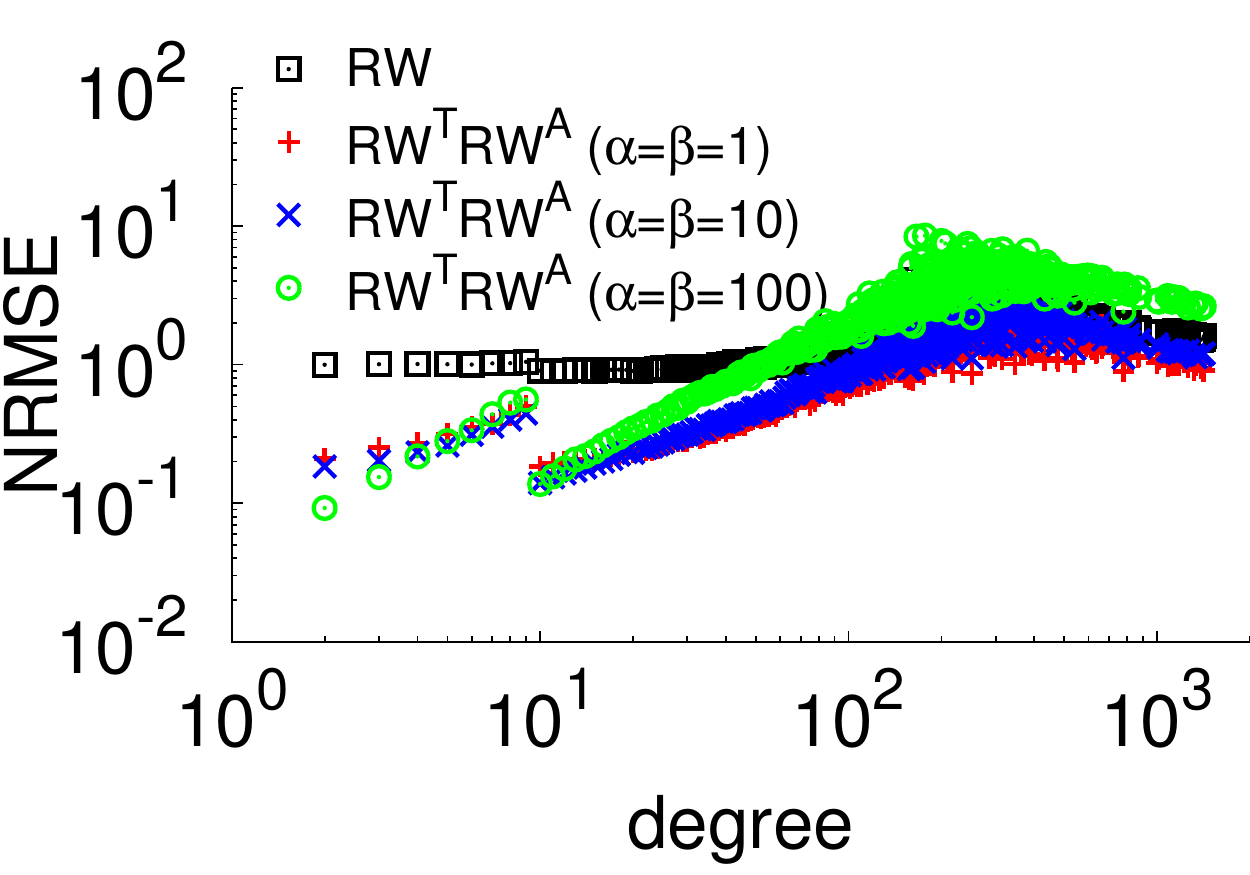}}
  \caption{PDF NRMSE of different estimators. \label{fig:syn_nrmse_pdf}}
\end{figure}

\begin{figure}[tp]
  \centering
  \subfloat[\VS{A}-I]{\includegraphics[width=.5\linewidth]{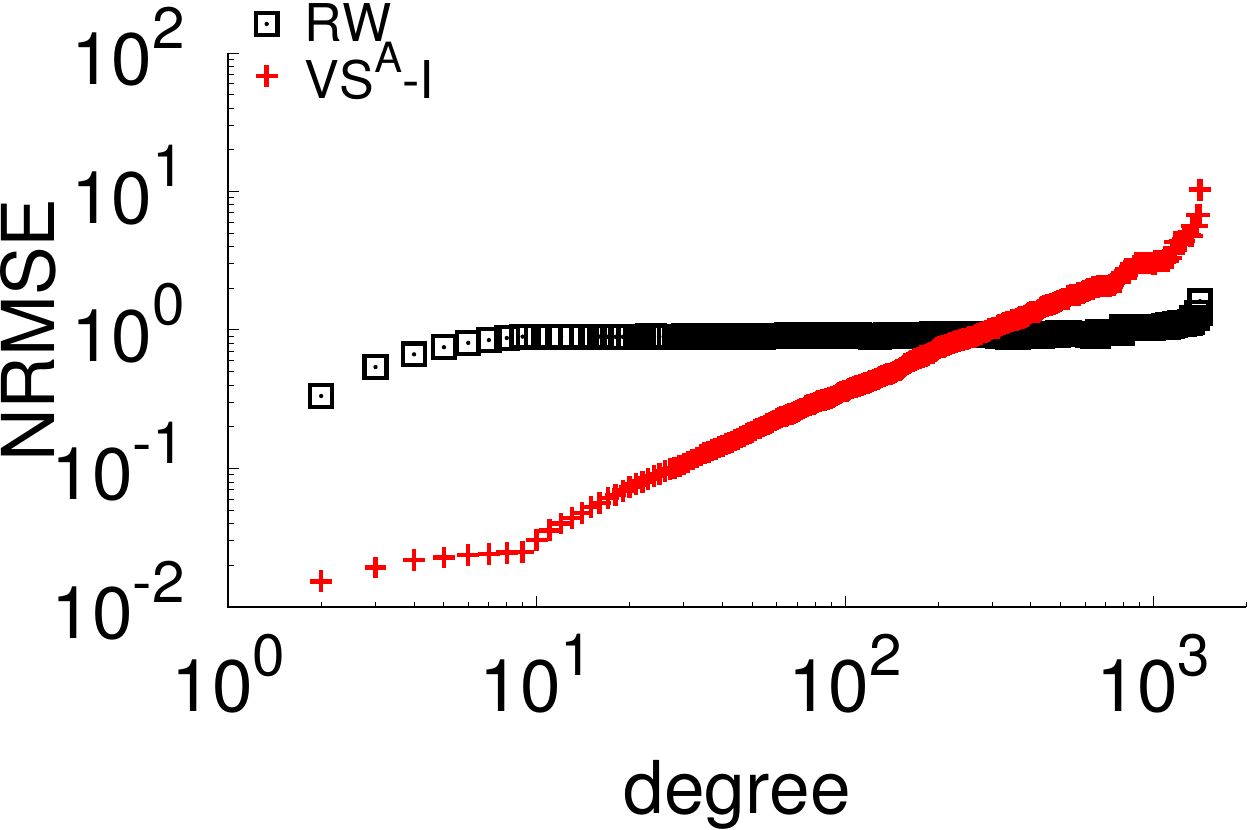}}
  \subfloat[\VS{A}-II]{\includegraphics[width=.5\linewidth]{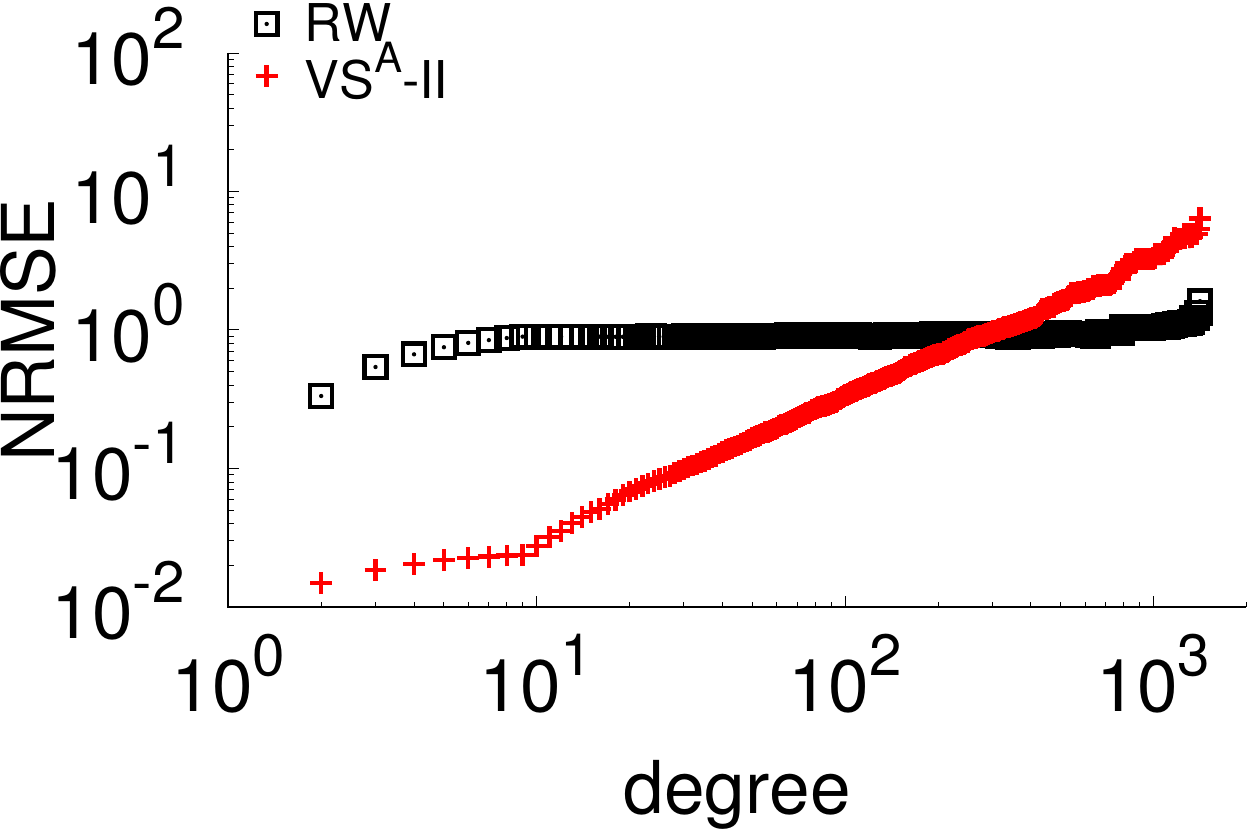}} \\
  \subfloat[\RW{T}\VS{A}]{\includegraphics[width=.5\linewidth]{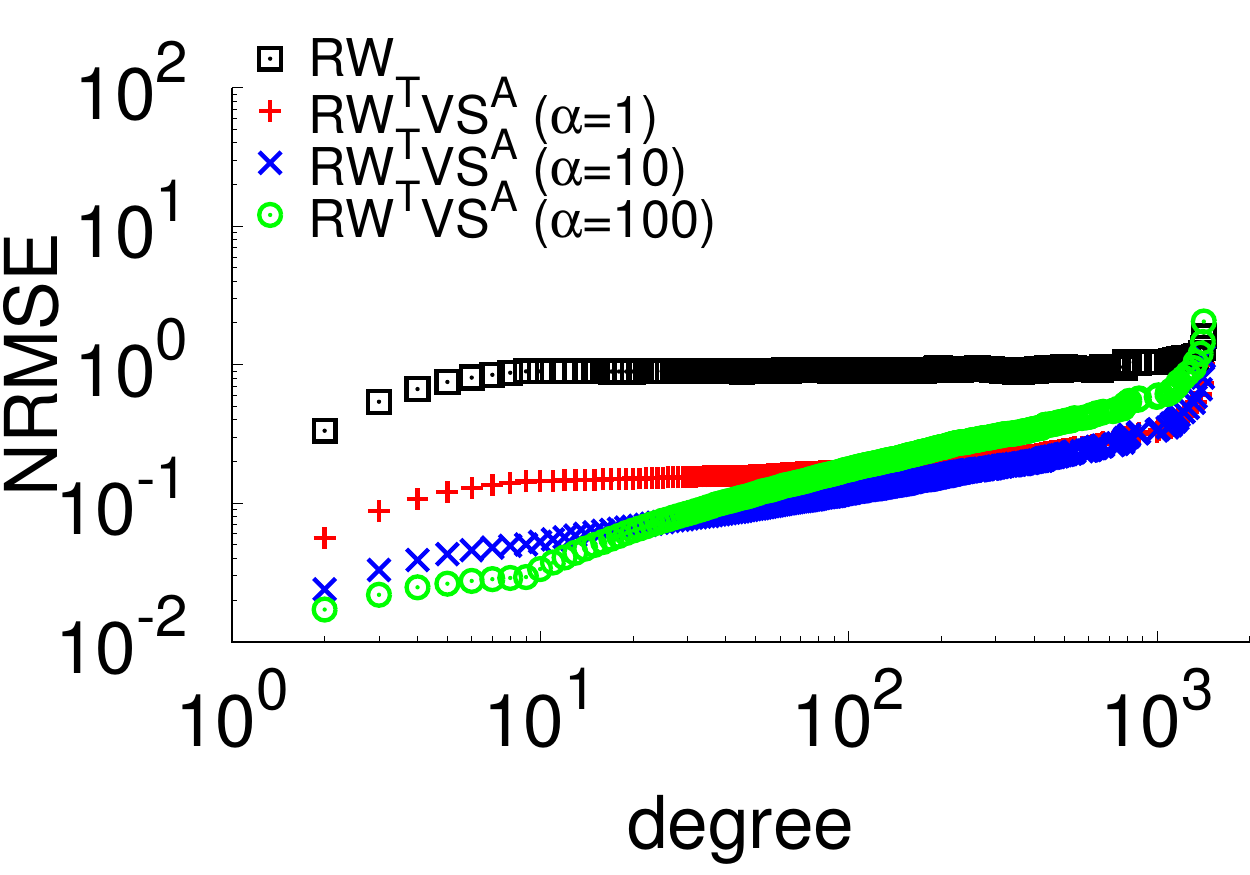}}
  \subfloat[\RW{T}\RW{A}]{\includegraphics[width=.5\linewidth]{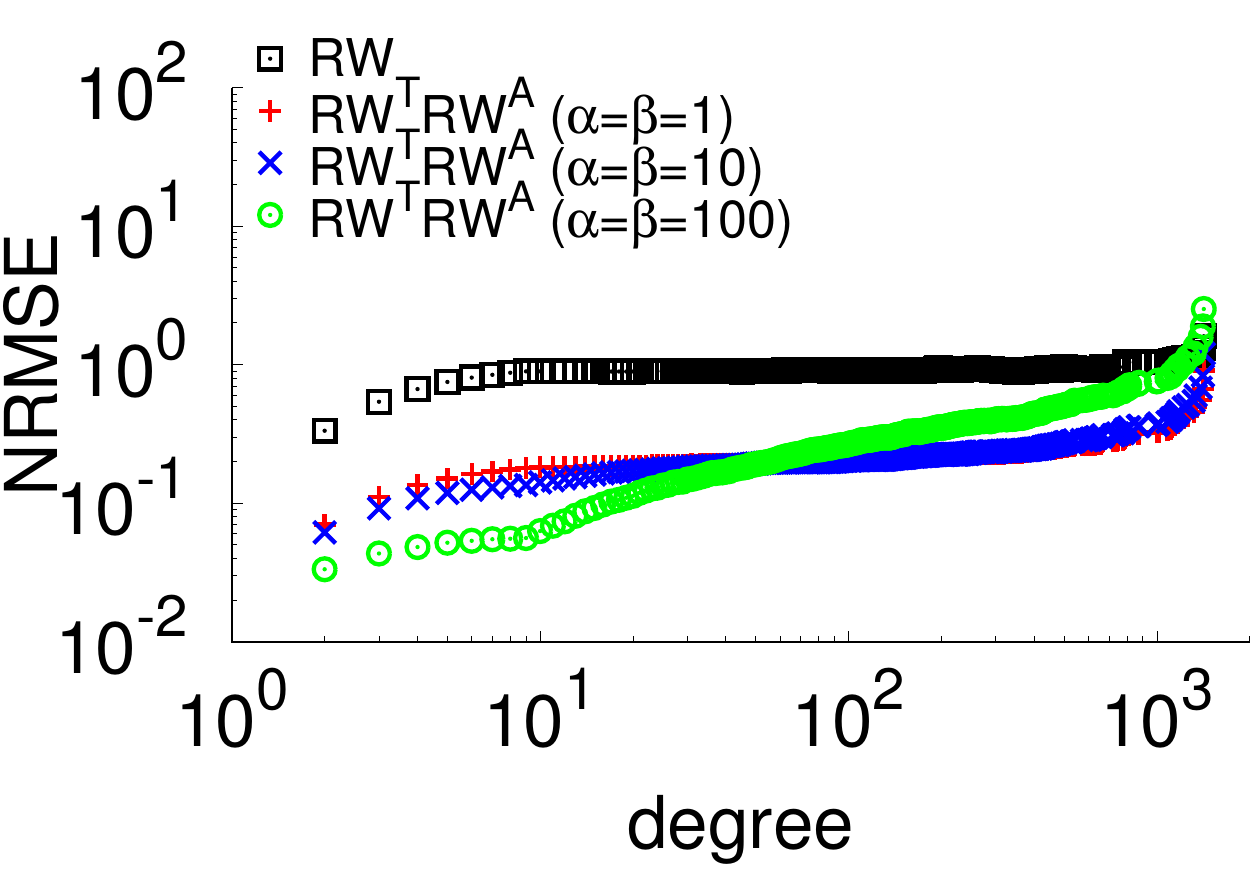}}
  \caption{CCDF NRMSE of different estimators. \label{fig:syn_nrmse_ccdf}}
\end{figure}

To clearly see the performance difference, we also show the NRMSE of the random
walk (RW) estimator as a baseline.
Because a RW can hardly converge over a barbell graph within $B=0.01n$ steps, we
observe that NRMSE of RW is almost the largest among all estimators for low
degrees.
Comparing \VS{A}-I and \VS{A}-II with RW, we find that the two \VS{A} estimators
can provide smaller PDF/CCDF NRMSE for low degree nodes than RW.
However, \VS{A} estimators produce larger NRMSE for high degree nodes than RW.
Therefore, \VS{A} can better estimate low degree nodes than high degree nodes in a
graph.

The weakness of \VS{A} can be overcome by \RW{T}\VS{A} and \RW{T}\RW{A}.
From Fig.~\ref{fig:syn_nrmse_ccdf}, it is clearer to see that when indirect jumps
are incorporated into random walks in \RW{T}\VS{A} and \RW{T}\RW{A}, NRMSE for
high degree nodes decreases, and NRMSE for low degree nodes remains smaller than
RW.
If we increase the probability of jumping at each step of random walk by
increasing $\alpha$ and $\beta$, we observe that NRMSE for low degree nodes
decreases, but NRMSE for high degree nodes increases.
This behavior is similar to RWwJ~\cite{Avrachenkov2010,Ribeiro2012b} and
demonstrates that the indirect jumps in \RW{T}\VS{A} and \RW{T}\RW{A} indeed
behave similarly as the direct jumps in RWwJs.

\subsection{Experiments on LBSN Datasets}

In the second experiment, we apply the \VS{A}-II method on two real-world LBSN
datasets to solve the problem in Example~\ref{eg:weibo}, i.e., measure user
characteristics in an area of interest on the map.

\header{LBSN datasets.}
We obtain two public LBSN datasets from Brightkite and Gowalla~\cite{Cho2011}.
Brightkite and Gowalla are once two popular LBSNs where users shared their
locations by checking-in.
Users in the two LBSNs are also connected by undirected friendship relations,
which form two user social networks.
The statistics of these two datasets are summarized in Table~\ref{tab:lbsn}.

\begin{table}[t]
\centering
\caption{Summary of two LBSN datasets. \label{tab:lbsn}}
\begin{threeparttable}[b]
\begin{tabular}{|c|l@{}|r|r|}
\hline
\multicolumn{2}{|c|}{{\bf dataset}} & {\bf Brightkite} & {\bf Gowalla} \\
\hline
\hline
\multirow{5}{*}{$G$}
 & network type & undirected & undirected \\
 & users & $58,228$ & $196,591$ \\
 & friendship edges & $214,078$ & $950,327$ \\
 & users in LCC\tnote{1} & $56,739$ & $196,591$ \\
 & edges in LCC & $212,945$ & $950,327$ \\
\hline
\multirow{3}{*}{$G'$ and $G_b$}
 & venues & $772,966$ & $1,280,969$ \\
 & users having check-ins & $51,406$ & $107,092$ \\
 & check-ins & $4,491,143$ & $6,442,890$ \\
\hline
\multirow{3}{*}{\minitab{$G'$ and $G_b$ \\ for NYC}}
 & venues in NYC\tnote{2} & $23,484$ & $26,448$ \\
 & users checking in NYC & $4,257$ & $7,399$ \\
 & check-ins in NYC & $33,656$ & $113,423$ \\
\hline
\end{tabular}
\begin{tablenotes}
\item[1] The largest connected component.
\item[2] The New York City (Fig.~\ref{fig:nyc}).
\end{tablenotes}
\end{threeparttable}
\end{table}

Because we are only interested in users that have check-ins, i.e., each node in
the target graph connects to at least one node in the auxiliary graph, \VS{A} is
applicable on these two datasets.
Suppose that we want to measure characteristics of users located around New York
City (NYC), which is specified by a rectangle region on a map: latitude range
$40.4^\circ\sim 41.4^\circ$, longitude range $-74.3^\circ\sim -73.3^\circ$ (see
Fig.~\ref{fig:nyc}).
The goal is to estimate degree distribution of the users who checked in this
region.
As we explained in Introduction, directly sampling users is inefficient.
Here, we apply the \VS{A}-II along with a venue sampling method --- Random Region
Zoom-In (RRZI)~\cite{Wang2014} to illustrate how to sample users in NYC
more efficiently.

\begin{figure}[t]
\centering
\includegraphics[width=.65\linewidth]{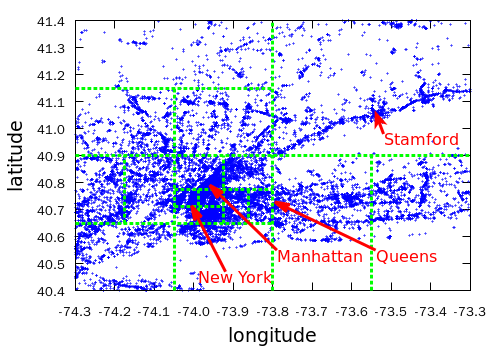}
\caption{Venue distribution in New York City and illustration of accessible
  subregions used by RRZI.
  Each subregion contains less than $K$ venues.}
\label{fig:nyc}
\end{figure}

\header{Venue sampling.}
RRZI utilizes a venue query API provided by LBSNs to sample venues on a map.
The API requires a user to specify a rectangle region by providing the south-west
and north-east corners latitude-longitude coordinates, and then the API returns a
set of venues in this region.
Usually, the API can only return at most $K$ venues in a queried region.
RRZI regularly zooms in the region until the subregion is fully \emph{accessible},
i.e., the API returns less than $K$ venues in the subregion.
The zooming-in process is equivalent to dividing the region into many
non-overlapping accessible subregions, as illustrated in Fig.~\ref{fig:nyc}, and
each subregion is associated with a fixed probability related to the zooming-in
strategy.
This feature enables RRZI to sample venues within an area of interest.

\header{Results.}
Combining \VS{A}-II with RRZI, denoted by RRZI-\VS{A}, we conduct experiments on
Brightkite and Gowalla to indirectly sample users in NYC.
We totally sample $5\%$ of venues in NYC and calculate the degree distribution of
users in NYC.
The results are depicted in Fig.~\ref{fig:lbsn}.

\begin{figure}[t]
  \centering
  \subfloat[estimates (Brightkite)\label{f:bri_est}]{%
    \includegraphics[width=.5\linewidth]{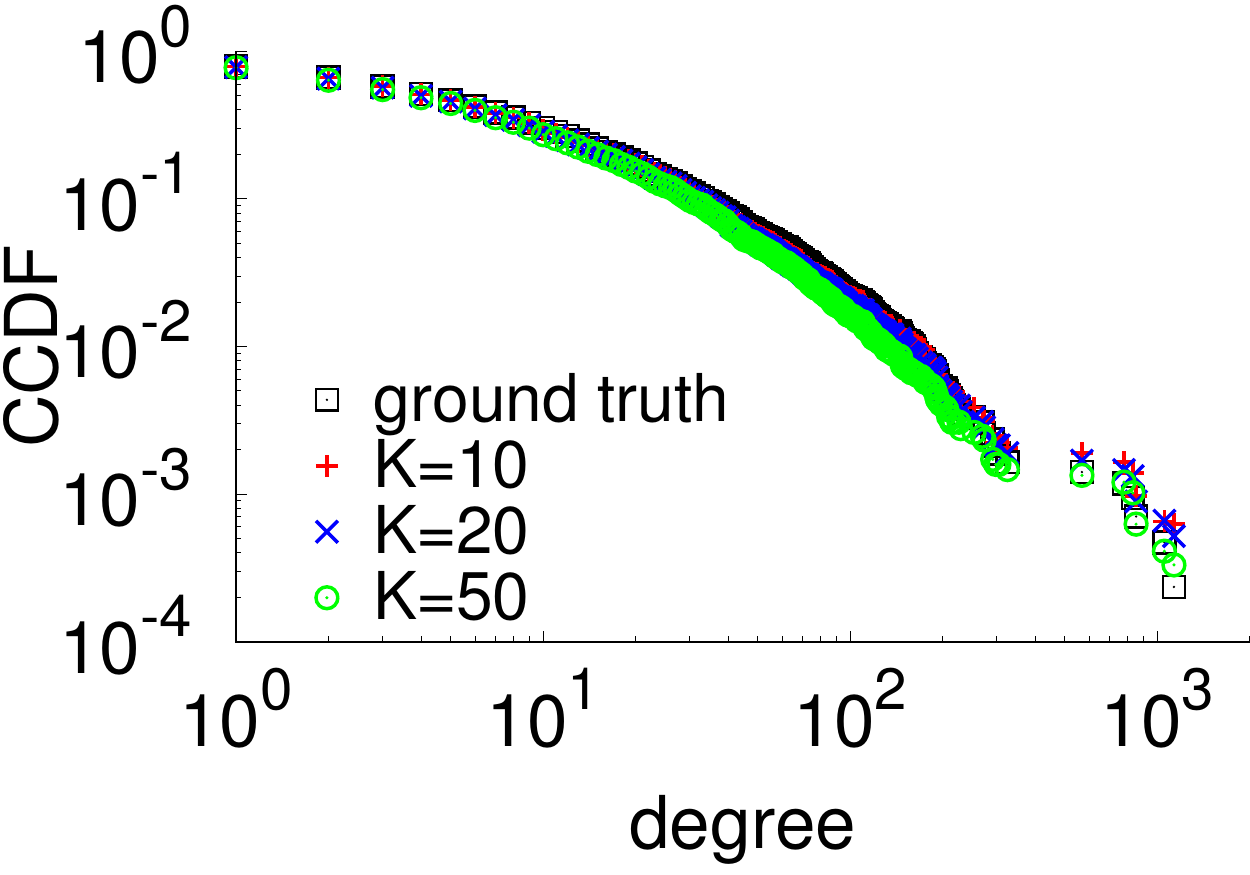}}
  \subfloat[estimates (Gowalla)\label{f:gow_est}]{%
    \includegraphics[width=.5\linewidth]{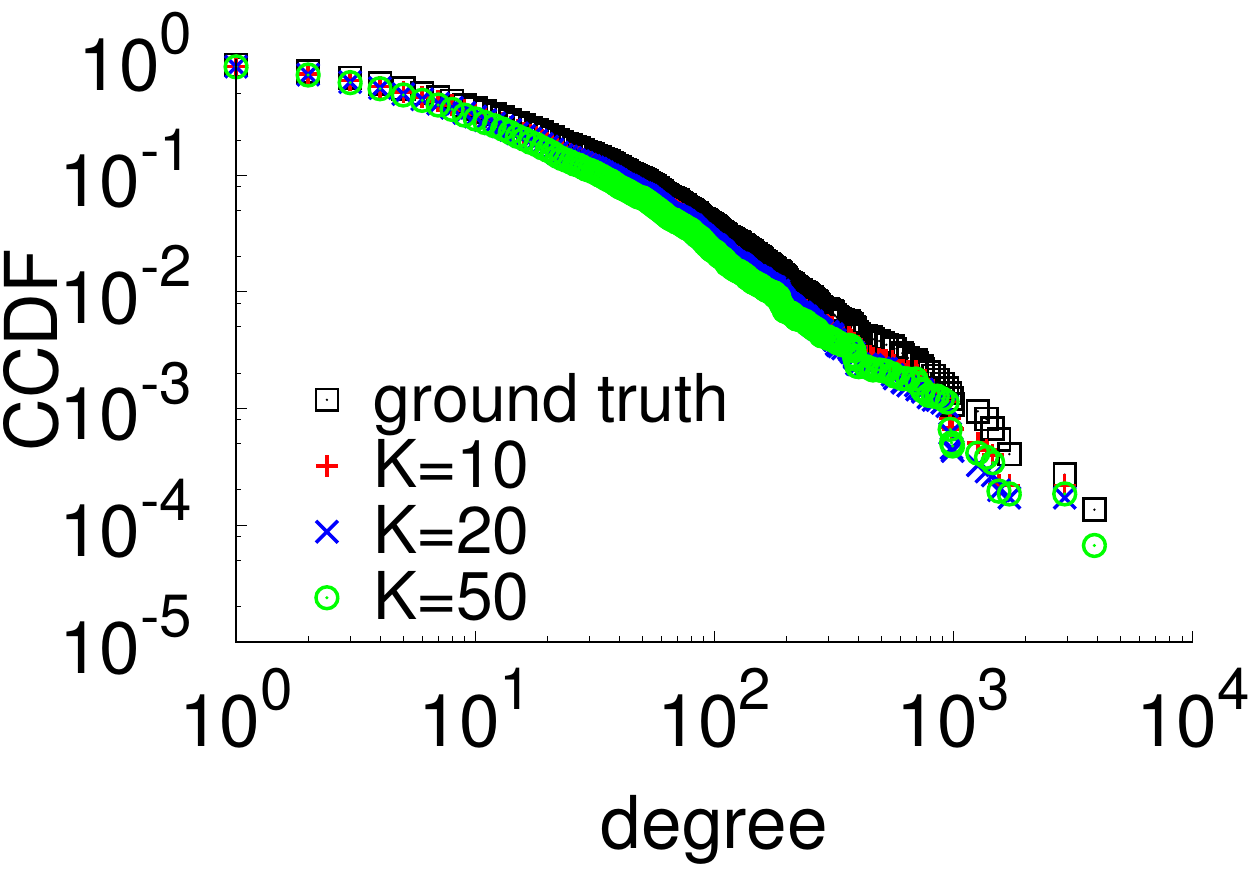}} \\
  \subfloat[PDF NRMSE (Brightkite)\label{f:bri_nrmse_pdf}]{%
    \includegraphics[width=.5\linewidth]{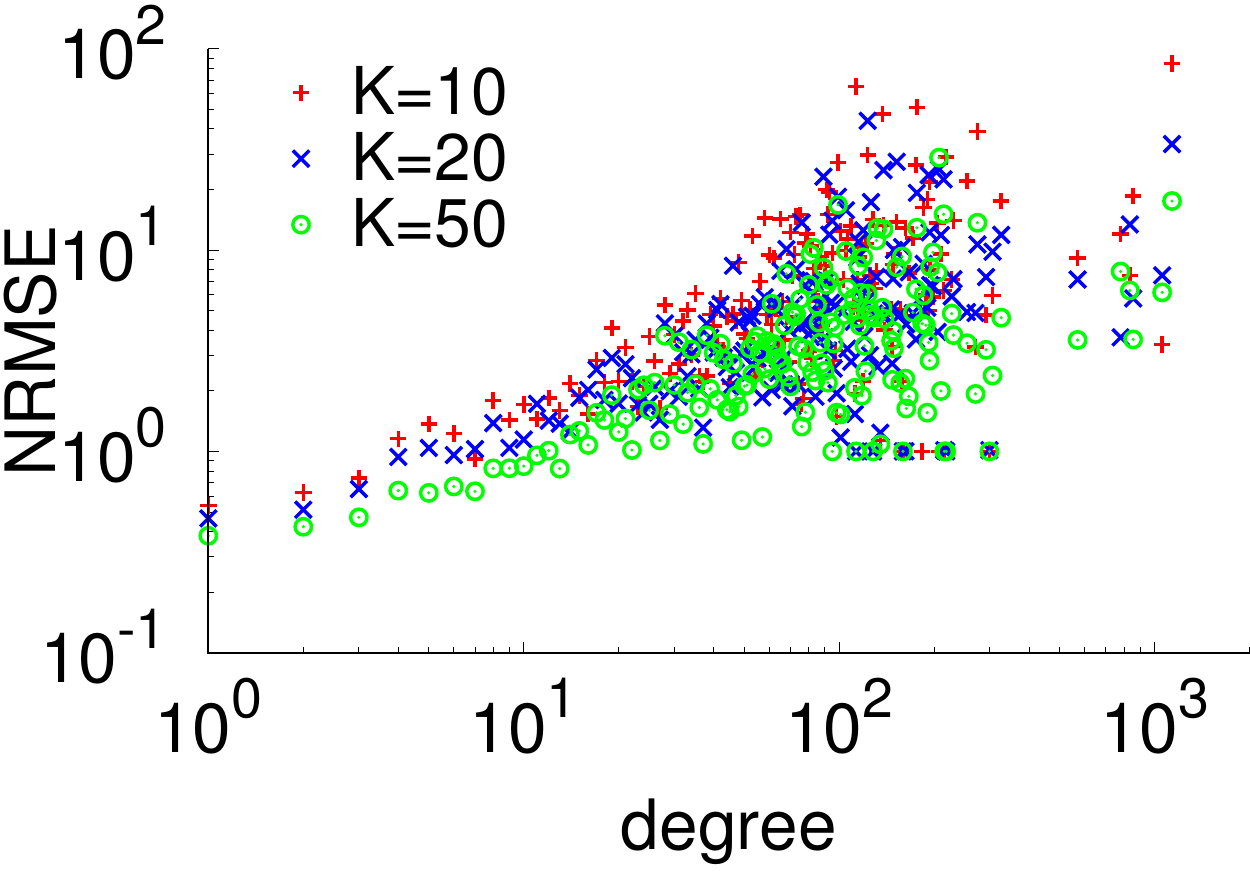}}
  \subfloat[PDF NRMSE (Gowalla)\label{f:gow_nrmse_pdf}]{%
    \includegraphics[width=.5\linewidth]{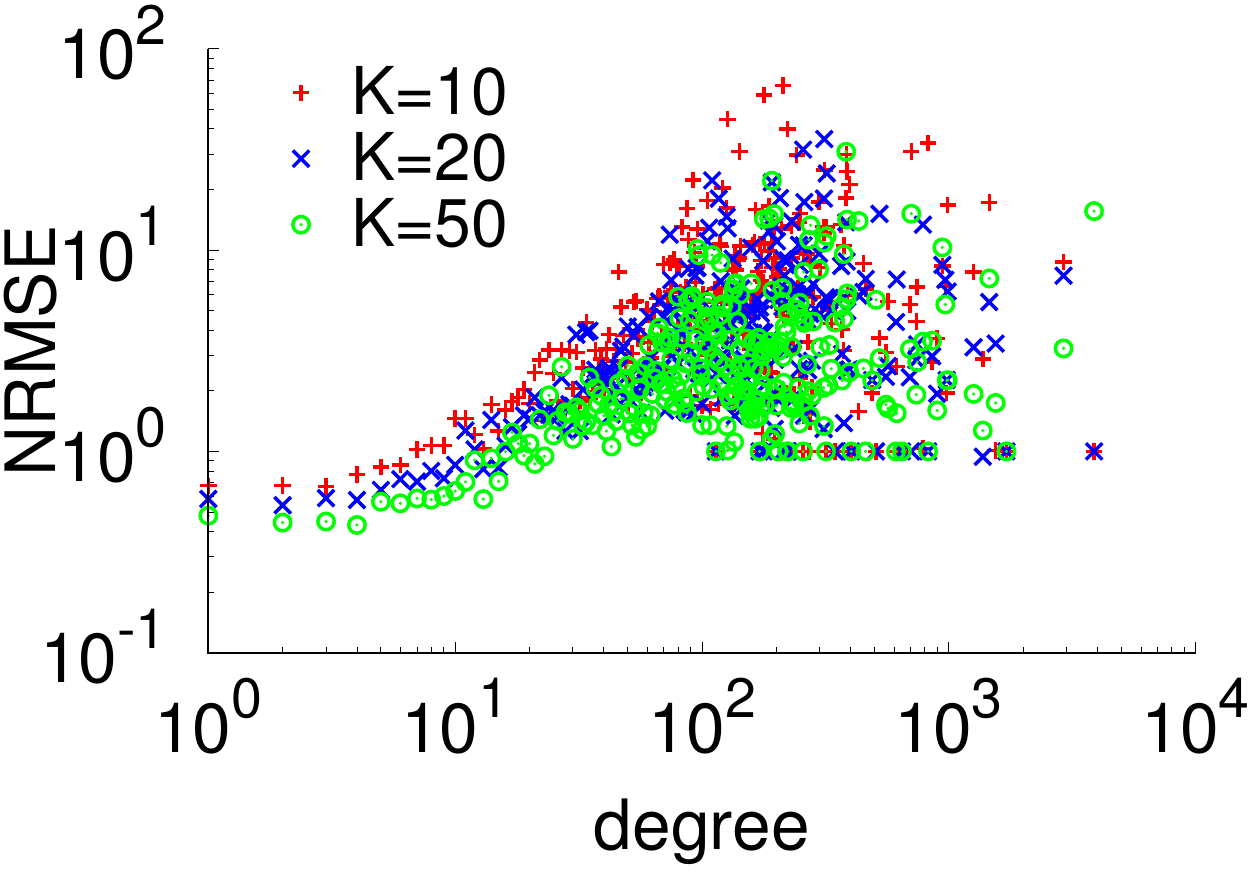}} \\
  \subfloat[CCDF NRMSE (Brightkite)\label{f:bri_nrmse_ccdf}]{%
    \includegraphics[width=.5\linewidth]{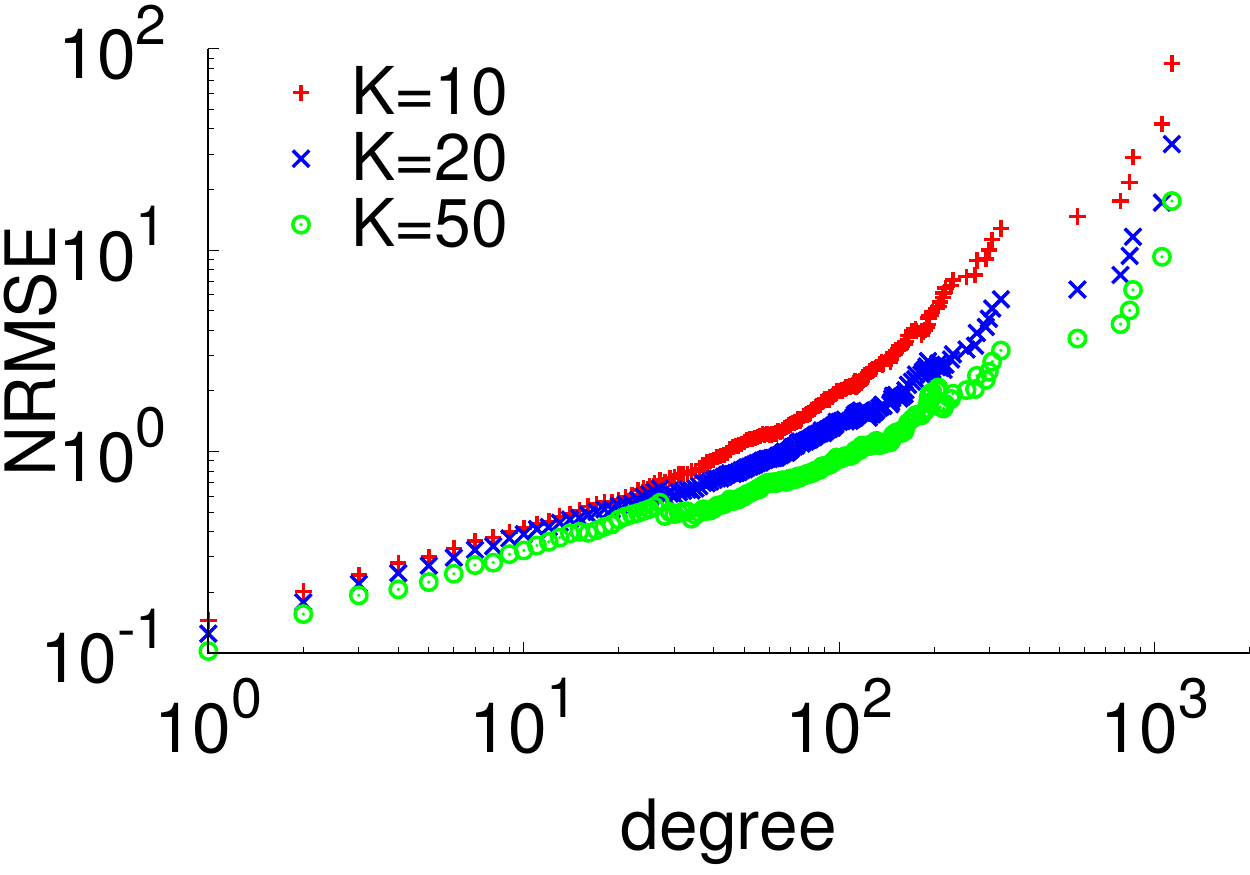}}
  \subfloat[CCDF NRMSE (Gowalla)\label{f:gow_nrmse_ccdf}]{%
    \includegraphics[width=.5\linewidth]{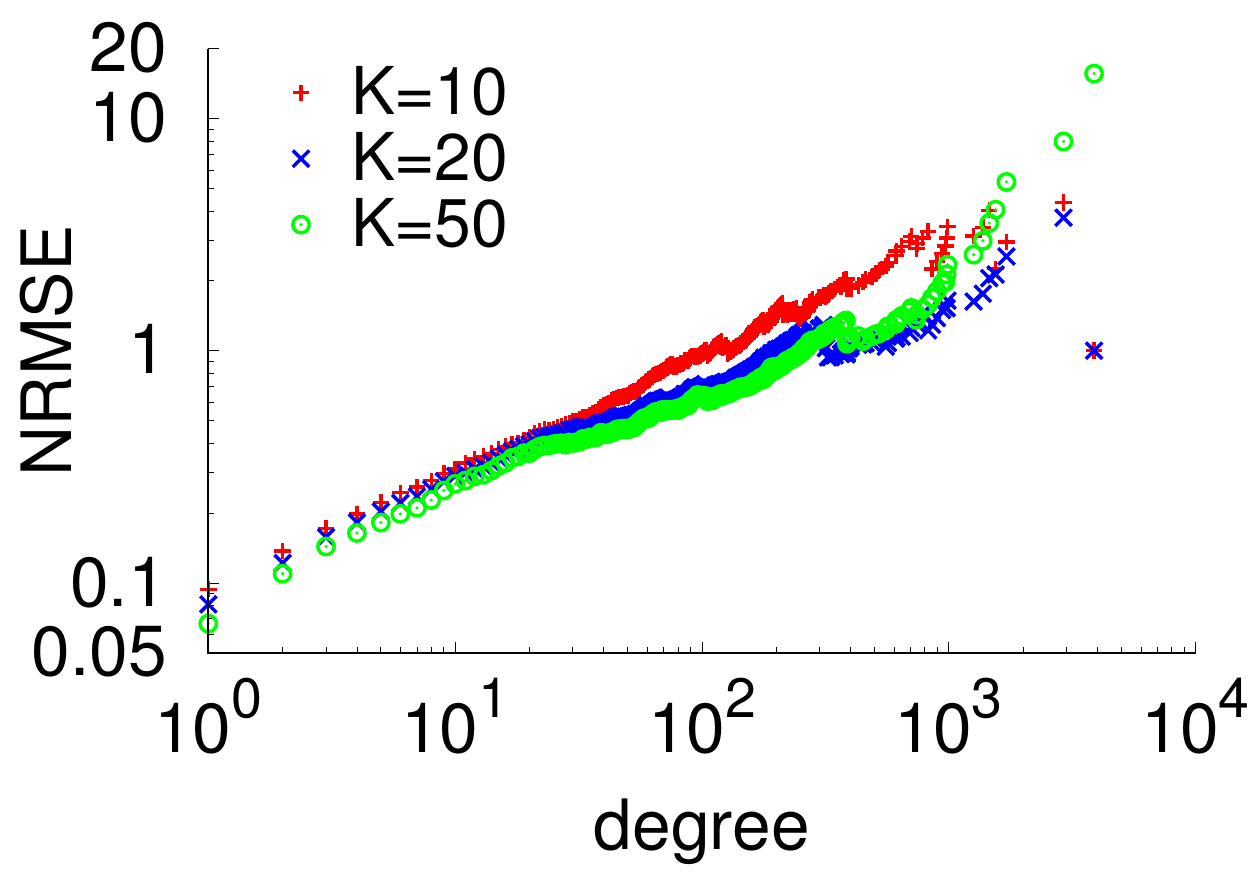}}
  \caption{Performance of RRZI-\VS{A} on Brightkite and Gowalla.\label{fig:lbsn}}
\end{figure}

Figures~\ref{f:bri_est} and~\ref{f:gow_est} depict the estimates of CCDF with
different query capacity $K$.
We observe that our RRZI-\VS{A} method can provide good estimates of user
characteristics in NYC on both datasets.
Specifically, the estimates for low degree users are better than high degree
users, and this is clear to see from the PDF/CCDF NRMSE plots.
This feature coincides with our previous analysis using synthetic data.
From the NRMSE plots, we can also find an approximate law that a larger query
capacity $K$, i.e., the maximum number of venues the API can return, reduces the
estimation error of RRZI-\VS{A}.
However, it is not true for estimating high degree users on Gowalla in
Fig.~\ref{f:gow_nrmse_ccdf}.
In fact, a better way to reduce estimation error is to combine \VS{A}-II with
other better venue sampling methods discussed in \cite{Li2012,Li2014,Wang2014}.
However, we have to omit this due to space limitation.

\subsection{Experiments on Amazon Product Co-purchasing Network}

In the third experiment, we compare the performance of \VS{A}-I and \RW{T}\VS{A}
sampling methods on the Amazon product co-purchasing network.

\header{Amazon product co-purchasing network.}
We build an Amazon product co-purchasing network from the Amazon dataset provided
by~\cite{McAuley2015}.
The network is created based on ``customers who bought this item also bought''
feature of the Amazon website.
That is, if a product $i$ is co-purchased with product $j$, the network contains
an undirected edge between $i$ and $j$.
In addition, each product belongs to at least one category on Amazon, and Amazon
provides a complete category list on its homepage to facilitate customers to
conveniently browse the products.
Thus, we can leverage this category list to perform indirect sampling of the
co-purchasing network.
The detailed statistics of the Amazon dataset are provided in
Table~\ref{tab:amazon}.

\begin{table}[htp]
  \centering
  \caption{Amazon product co-purchasing network statistics.}
  \label{tab:amazon}
  \begin{tabular}{|c|l|r|}
    \hline
    \multirow{3}{*}{$G$}
         & product co-purchasing network              & undirected   \\
         & \# of products                             & $4,015,942$  \\
         & \# of co-purchases                         & $78,792,050$ \\
    \hline
    $G'$ & \# of categories                           & $10,164$     \\
    \hline
    \multirow{3}{*}{$G_b$}
         & \# of product-category associations        & $15,829,046$ \\
         & avg. \# of categories a product belongs to & $4$          \\
         & avg. \# of products in a category          & $1,557$      \\
    \hline
  \end{tabular}
\end{table}

This dataset is suitable for us to study the performance of \VS{A}-I and
\RW{T}\VS{A}, where the availability of the complete category list allows us to
conduct uniform vertex sampling on the auxiliary graph.
Here we sample $1\%$ of the nodes from target graph, and compare the accuracy of
estimating PDF/CCDF degree distribution using different methods.
The results are averaged over $1,000$ runs and are depicted in
Fig.~\ref{fig:amazon}.

\begin{figure}[htp]
  \centering
  \subfloat[unbiasedness of CCDF estimates\label{f:amazon_unbias}]{%
    \includegraphics[width=.56\linewidth]{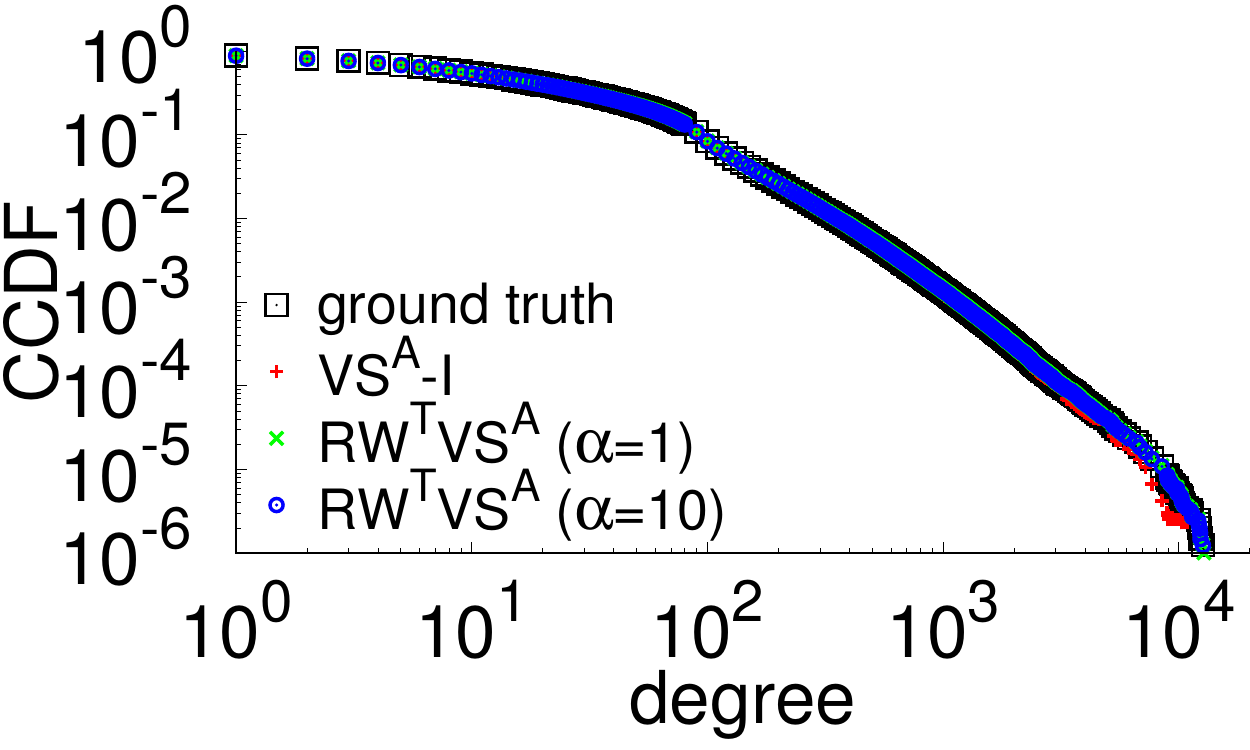}} \\
  \subfloat[PDF NRMSE\label{f:amazon_pnrmse}]{%
    \includegraphics[width=.5\linewidth]{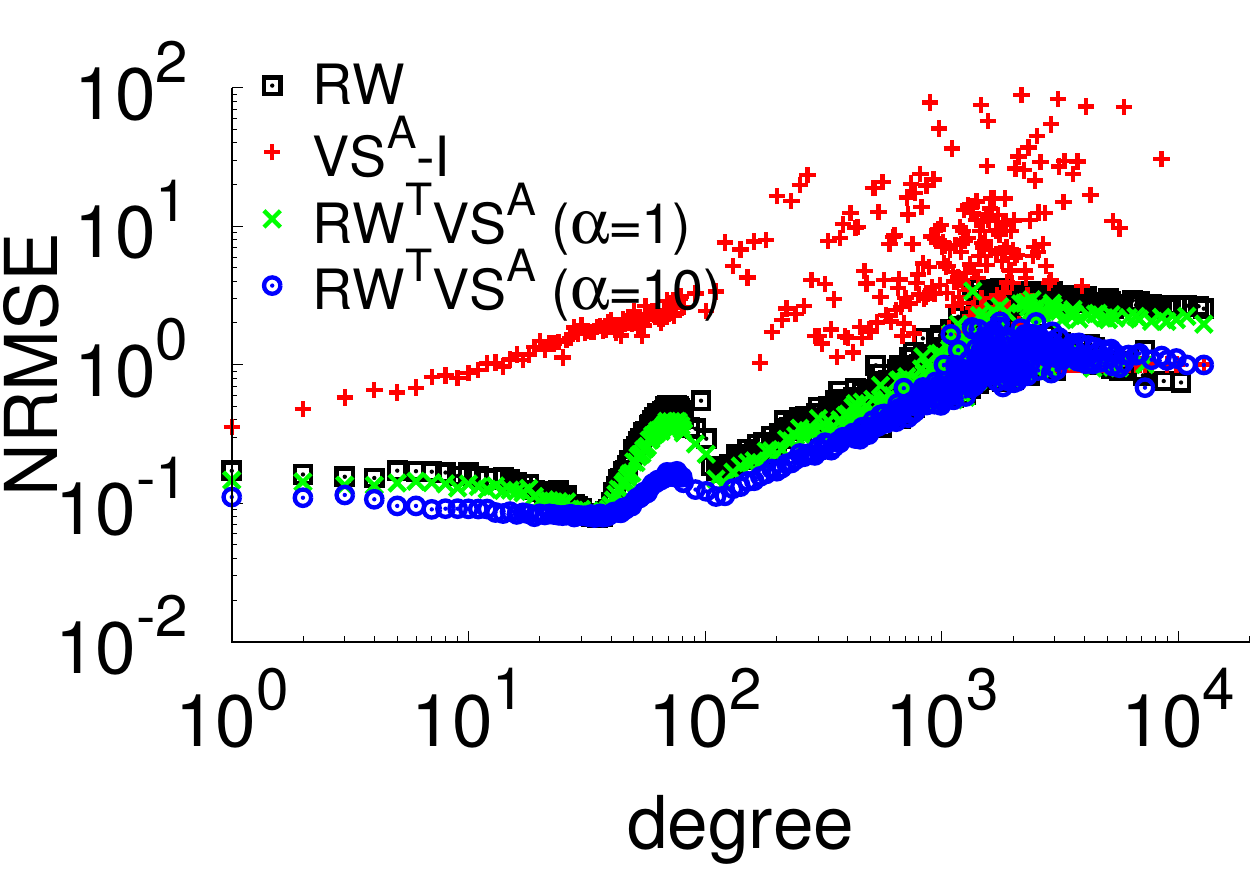}}
  \subfloat[CCDF NRMSE\label{f:amazon_cnrmse}]{%
    \includegraphics[width=.5\linewidth]{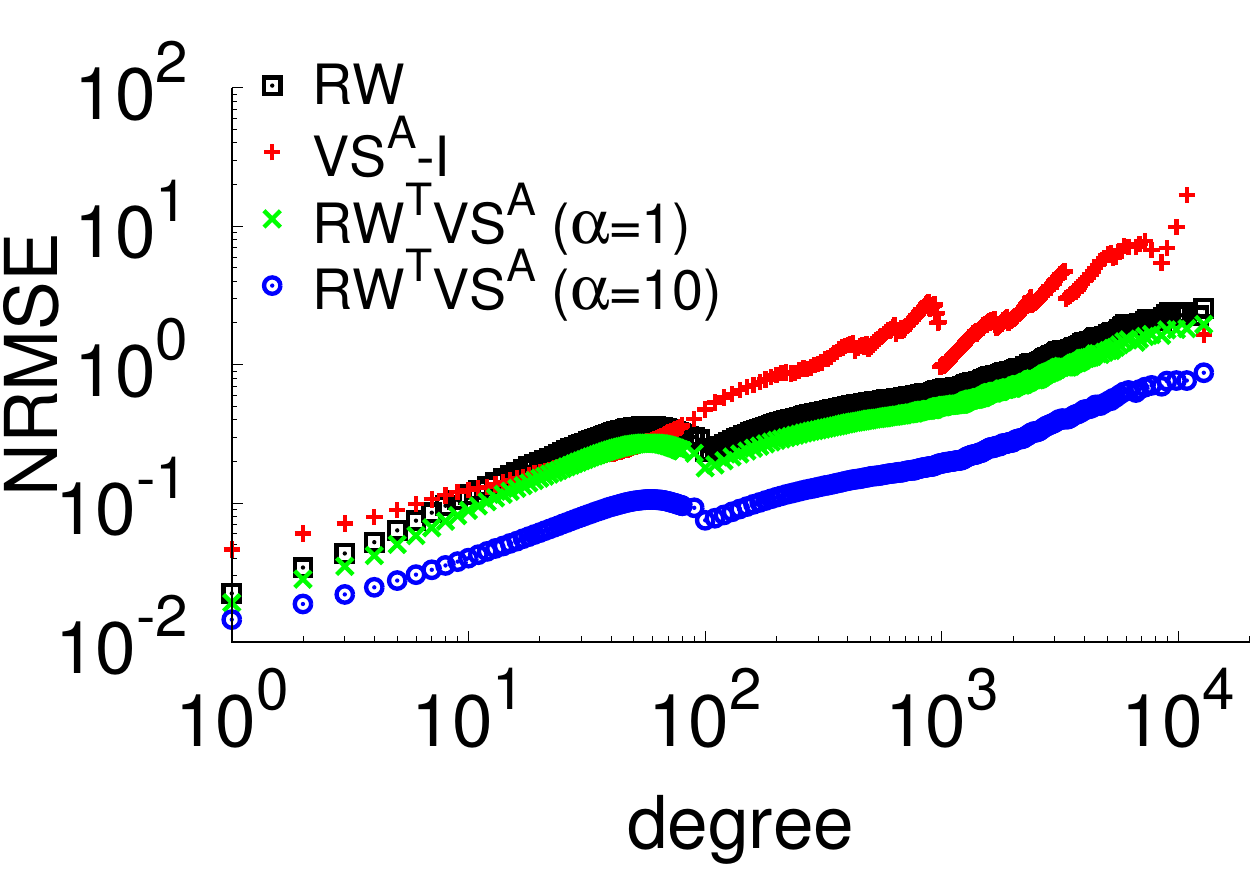}}
  \caption{Amazon product co-purchasing network characterizing.}
  \label{fig:amazon}
\end{figure}

\header{Results.}
From Fig.~\ref{f:amazon_unbias}, we observe that the two methods can indeed
provide unbiased estimates of the CCDF.
From Figs.~\ref{f:amazon_pnrmse} and~\ref{f:amazon_cnrmse}, we also observe that
different methods have different estimation accuracy.
In general, \VS{A}-I has relatively large estimation error, then comes the random
walk estimator, and \RW{T}\VS{A} has the lowest estimation error among these three
estimators.
\RW{T}\VS{A} leverages the category list to perform indirect jumps on the target
graph, and this approach can significantly improve the estimation accuracy.
If we slightly increase $\alpha$ to increase the jumping probability, we observe
that the estimation error further decreases.

\subsection{Experiments on Mtime Dataset}

In the fourth experiment, we apply \RW{T}\VS{A} and \RW{T}\RW{A} on Mtime to
measure the Mtime user characteristics.

\header{Mtime dataset.}
Mtime~\cite{Mtime} is a popular online movie database in China, which comprises
two types of accounts: Mtime users and movie actors.
Mtime users can follow each other to form a social network, and movie actors can
form connections with each other if they cooperated in the same movies.
A Mtime user can follow movie actors if she is a fan of the actor.
Suppose we want to measure Mtime user characteristics, then the relations between
Mtime users and movie actors naturally form a two-layered network structure, where
\begin{itemize}
\item the target graph consists of Mtime users and their following
  relations;
\item the auxiliary graph consists of movie actors and their cooperation
  relations;
\item and the bipartite graph consists of Mtime users, movie actors and the fan
  relations between them.
\end{itemize}

To build a groundtruth dataset, we have collected the complete Mtime network by
traversing Mtime user and movie actor ID spaces\footnote{The user ID space ranges
  from $100000$ to $10000000$, and actor ID space ranges from $892000$ to
  $2100000$.}.
For each Mtime user, we collect the set of users she follows and users who follow
her.
This builds up a directed follower network among Mtime users.
Each Mtime user maintains a list including a subset of movie actors she is
interested in.
This information is used to build up the fan-relations between Mtime users and
movie actors.
For each movie actor, we collect the movies she participated in, and if two actors
participated in a same movie, we connect them.
This builds up a cooperative network among actors.
The complete Mtime dataset is summarized in Table~\ref{tab:mt}.

\begin{table}[htp]
  \centering
  \caption{Summary of the Mtime dataset \label{tab:mt}}
  \begin{tabular}{|c|l@{}|r|}
    \hline
    \multirow{6}{*}{$G$}
    & user follower network type & directed \\
    & total users (isolated and non-isolated) & $1,878,127$ \\
    & non-isolated users in follower network & $1,035,164$ \\
    & following relations & $14,861,383$ \\
    & users in LCC & $987,055$ \\
    & following relations in LCC & $14,791,482$ \\
    \hline
    \multirow{6}{*}{$G'$}
    & actor cooperative network type & undirected \\
    & total actors (isolated and non-isolated) & $1,123,340$ \\
    & non-isolated actors in cooperative network & $1,122,166$ \\
    & cooperative relations & $10,344,364$ \\
    & actors in LCC & $1,114,065$ \\
    & cooperative relations in LCC & $10,328,904$ \\
    \hline
    \multirow{7}{*}{$G_b$}
    & fan relations & $225,558,343$ \\
    & users following actors & $1,419,339$ \\
    & isolated users following actors & $842,963$ \\
    & actors having fans & $441,413$ \\
    & isolated actors having fans & $1,174$ \\
    & isolated actors having only isolated fans & $225$ \\
    & isolated users following only isolated actors & $393$ \\
    \hline
\end{tabular}
\end{table}

\header{Analysis of the dataset.}
First we provide some analysis about the Mtime dataset.
In Table~\ref{tab:mt}, comparing the first block with second block, which are
related to target graph $G$ and auxiliary graph $G'$ respectively, we find that
about $19\%$ of the user IDs and $93\%$ of the actor IDs are valid.
This indicates that conducting UNI on the auxiliary graph is more efficient than
conducting UNI on the target graph.
Moreover, we find that more than $47\%$ of the Mtime users are not in LCC, but the
number for actors is less than $0.1\%$.
This indicates that the auxiliary graph is better connected than the target graph.
Although a large fraction of users are isolated nodes in the target graph, from
the last block, we find that almost all the isolated users are connected to
non-isolated actors (except a few hundreds of them).
So the majority of isolated users are indirectly connected to other users through
actors.
This is illustrated in Fig.~\ref{fig:mt_cpnt}.
The advantage of introducing the two-layered network structure is now clear for
Mtime dataset, i.e., we can study a larger user space than simply the LCC of
target graph.

\begin{figure}[tp]
  \centering
  \begin{tikzpicture}[
und/.style={draw,thick,circle,minimum size=6pt,inner sep=0},
vnd/.style={draw=blue,thick,rectangle,minimum size=6pt,inner sep=0},
att/.style={left,fill=white,minimum size=0,inner sep=0,align=center}]

\node[und] (u11) at (0,0) {};
\node[und,above right = 0.5 and 0.3 of u11] (u12) {};
\node[und,above right = 0.4 and 1.0 of u11] (u13) {};
\node[und,right = 1 of u11] (u14) {};
\node[und,above right = 0.1 and 0.5 of u11] (u15) {};
\node[und,yshift=0.5cm] (u16) at (u11) {};
\draw[thick] (u16)--(u11)--(u12)--(u13)--(u14)--(u11)--(u15)--(u12) (u15)--(u14);
\node[draw,dashed,fit={(u11) (u12) (u13) (u14) (u15)}] (fu1) {};
\node[att,below=0.1 of fu1] {LCC of $G$\\ ($52\%$ of all users)};

\node[und,right =3.5 of u11] (u21) {};
\node[und,xshift=0.1cm,yshift=0.5cm] (u22) at (u21) {};
\node[und,above right = 0.5 and 1.0 of u21] (u23) {};
\node[und,right = 1 of u21] (u24) {};
\node[und,xshift=0.5cm,yshift=0.65cm] (u25) at (u21) {};
\node[und,xshift=0.5cm,yshift=0.1cm] (u26) at (u21) {};
\node[und,xshift=0.3cm,yshift=0.3cm] (u27) at (u26) {};
\node[draw,dashed,fit={(u21) (u22) (u23) (u24) (u25)}] (fu2) {};
\draw[thick] (u22)--(u25) (u26)--(u27);
\node[att,below=0.1 of fu2] {Isolated parts of $G$\\ ($48\%$ of all users)};

\node[vnd,yshift=1.9cm] (v11) at (u11) {};
\node[vnd,above right = 0.5 and 0.3 of v11] (v12) {};
\node[vnd,above right = 0.4 and 1.0 of v11] (v13) {};
\node[vnd,right = 1 of v11] (v14) {};
\node[vnd,above right = 0.1 and 0.5 of v11] (v15) {};
\node[vnd,yshift=0.55cm] (v16) at (v11) {};
\draw[thick,blue] (v11)--(v12)--(v13)--(v14)--(v15)--(v13) (v11)--(v14) (v15)--(v12)--(v16);
\node[draw,dashed,fit={(v11) (v12) (v13) (v14) (v15)}] (fa1) {};
\node[att,above=0.1 of fa1] {LCC of $G'$\\ ($99\%$ of all actors)};

\node[vnd,right =3.5 of v11] (v21) {};
\node[vnd,xshift=0.1cm,yshift=0.5cm] (v22) at (v21) {};
\node[vnd,above right = 0.5 and 1 of v21] (v23) {};
\node[vnd,right = 1 of v21] (v24) {};
\node[vnd,xshift=0.5cm,yshift=0.65cm] (v25) at (v21) {};
\node[vnd,xshift=0.5cm,yshift=0.1cm] (v26) at (v21) {};
\node[vnd,xshift=0.4cm,yshift=0.3cm] (v27) at (v26) {};
\node[draw,dashed,fit={(v21) (v22) (v23) (v24) (v25)}] (fa2) {};
\draw[thick,blue] (v27)--(v25) (v21)--(v26);
\node[att,above=0.1 of fa2] {Isolated parts of $G'$\\ ($1\%$ of all actors)};

\draw[line width=3pt,red,dashed] (fu1)--(fa1)--(fu2) (fu1)--(fa2)--(fu2);
\end{tikzpicture}

  \caption{The Mtime network components.
    Dashed red lines denote fan relations between actors and users.}
  \label{fig:mt_cpnt}
\end{figure}
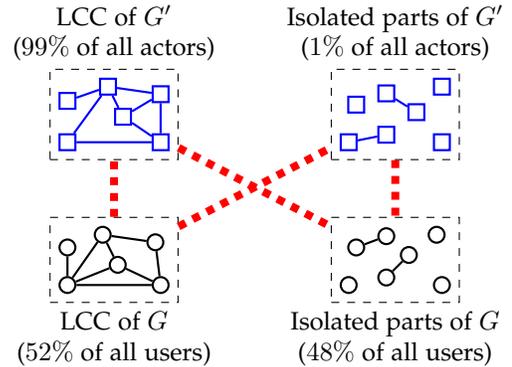

\header{Results.}
Using the Mtime dataset as a testbed, we demonstrate that \RW{T}\VS{A} and
\RW{T}\RW{A} methods can provide good estimates of user characteristics.
Although the user follower network is directed, we can build an undirected version
of the target graph on-the-fly while sampling because a user's in-coming and
out-going neighbors are known once the user is
queried~\cite{Ribeiro2010,Ribeiro2012b}.
Slightly different from previous experiments, here we will estimate both the in-
and out-degree distributions.

Figure~\ref{fig:mt_rwtvsa} depicts the results of \RW{T}\VS{A}.
In Figs.~\ref{f:mt_est_di} and~\ref{f:mt_est_do}, we show the in-degree and
out-degree CCDF estimates.
We can see that \RW{T}\VS{A} can provide unbiased estimates.
From Figs.~\ref{f:mt_nrmse_di_B} and~\ref{f:mt_nrmse_do_B}, we observe that when
sampling budget increases, the NRMSE decreases for both in-degree and out-degree
estimations.
From Figs.~\ref{f:mt_nrmse_di_a} and~\ref{f:mt_nrmse_do_a} we observe that when
more jumps are allowed by increasing $\alpha$ from $1$ to $100$, estimation
accuracy also increases.

\begin{figure*}
  \centering
  \subfloat[in-degree estimates ($\alpha=1$)\label{f:mt_est_di}]{%
    \includegraphics[width=.25\linewidth]{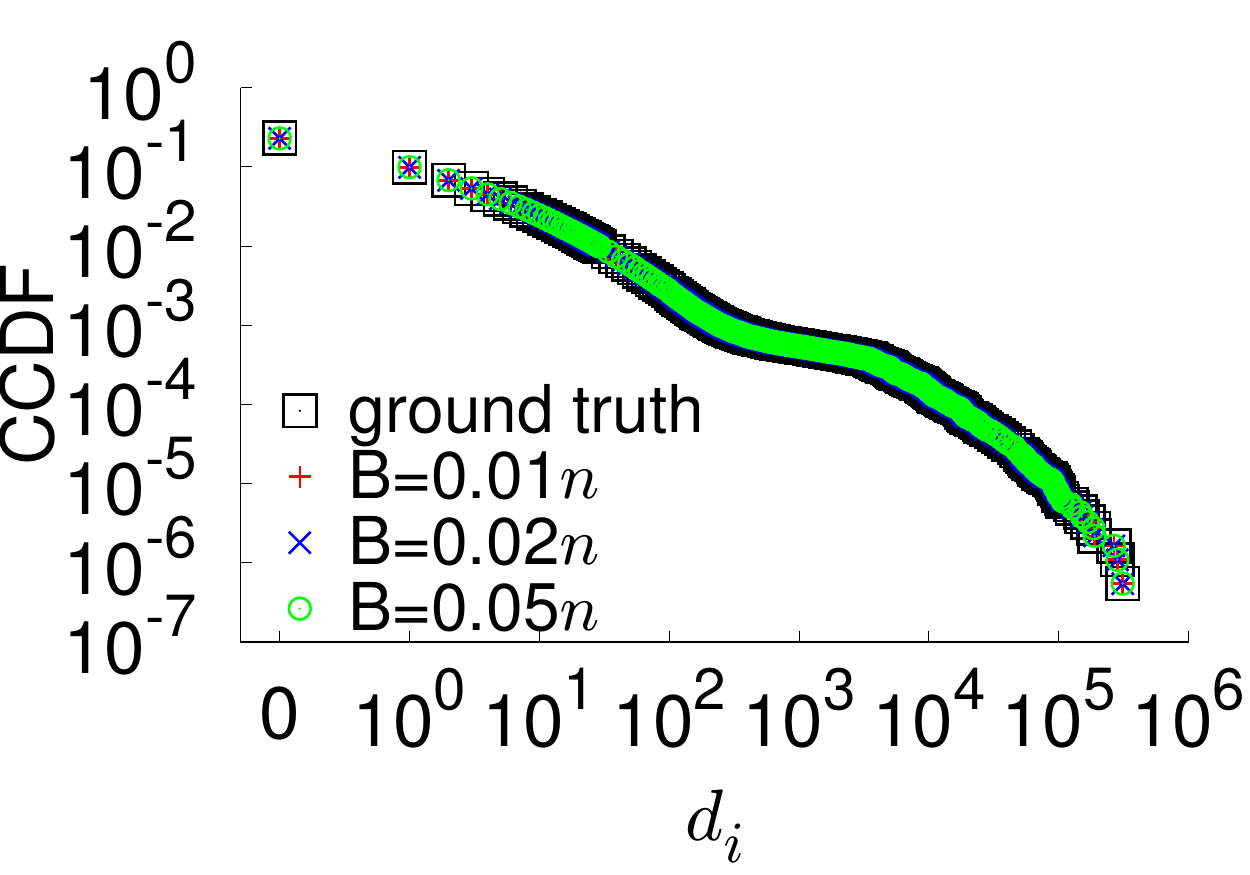}}
  \subfloat[CCDF NRMSE ($\alpha=1$)\label{f:mt_nrmse_di_B}]{%
    \includegraphics[width=.25\linewidth]{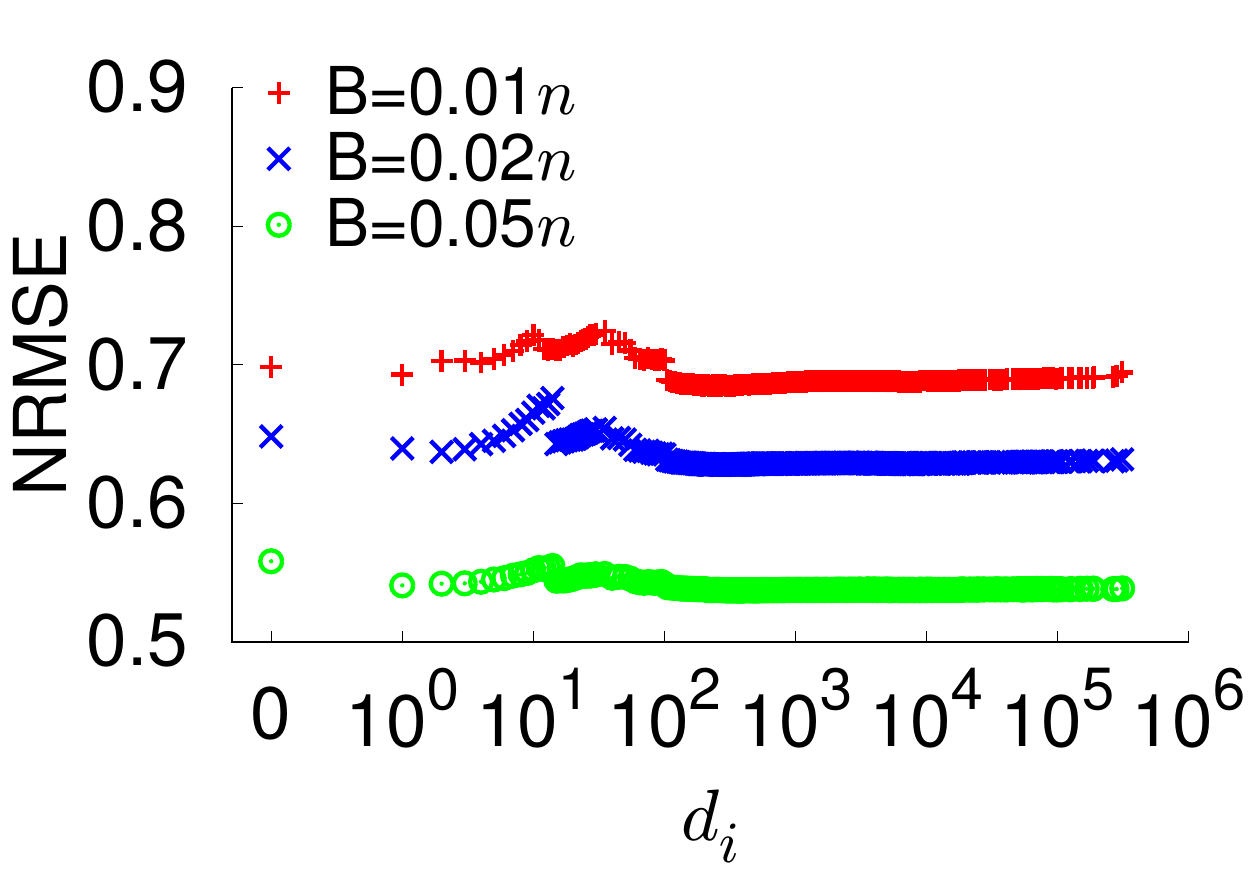}}
  \subfloat[CCDF NRMSE ($B=0.01n$)\label{f:mt_nrmse_di_a}]{%
    \includegraphics[width=.25\linewidth]{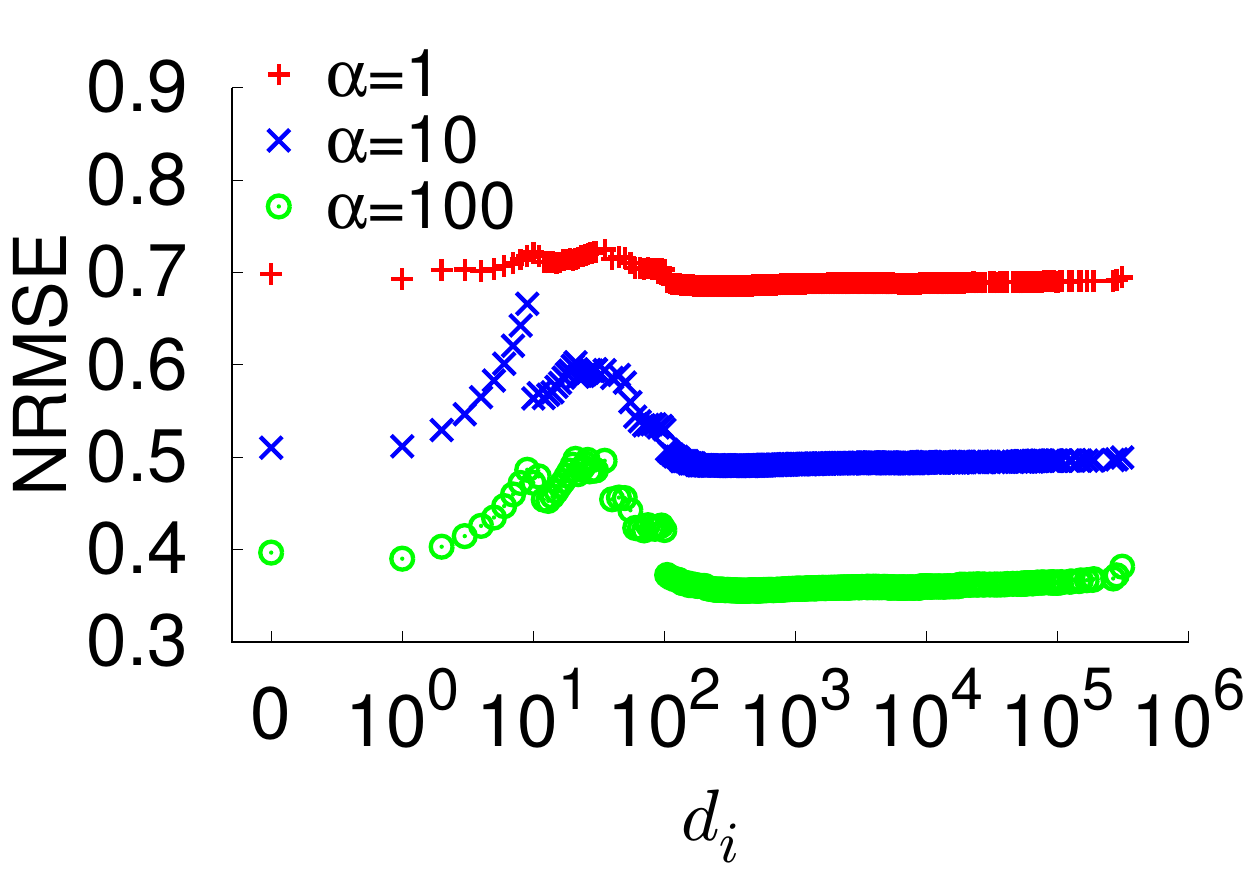}}
  \subfloat[too frequent jumping ($B\!=\!0.01n$)]{%
    \includegraphics[width=.25\linewidth]{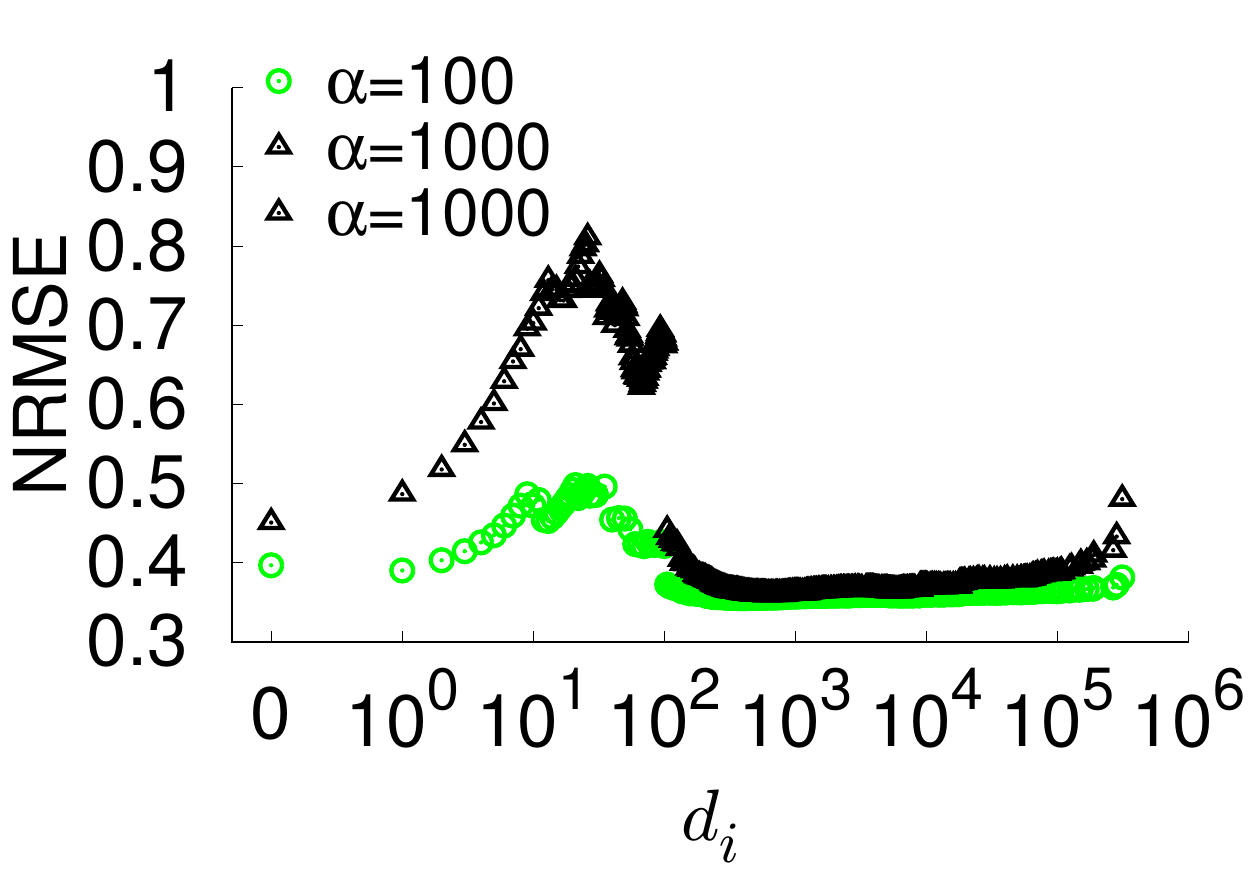}} \\
  \subfloat[out-degree estimates ($\alpha=1$)\label{f:mt_est_do}]{%
    \includegraphics[width=.25\linewidth]{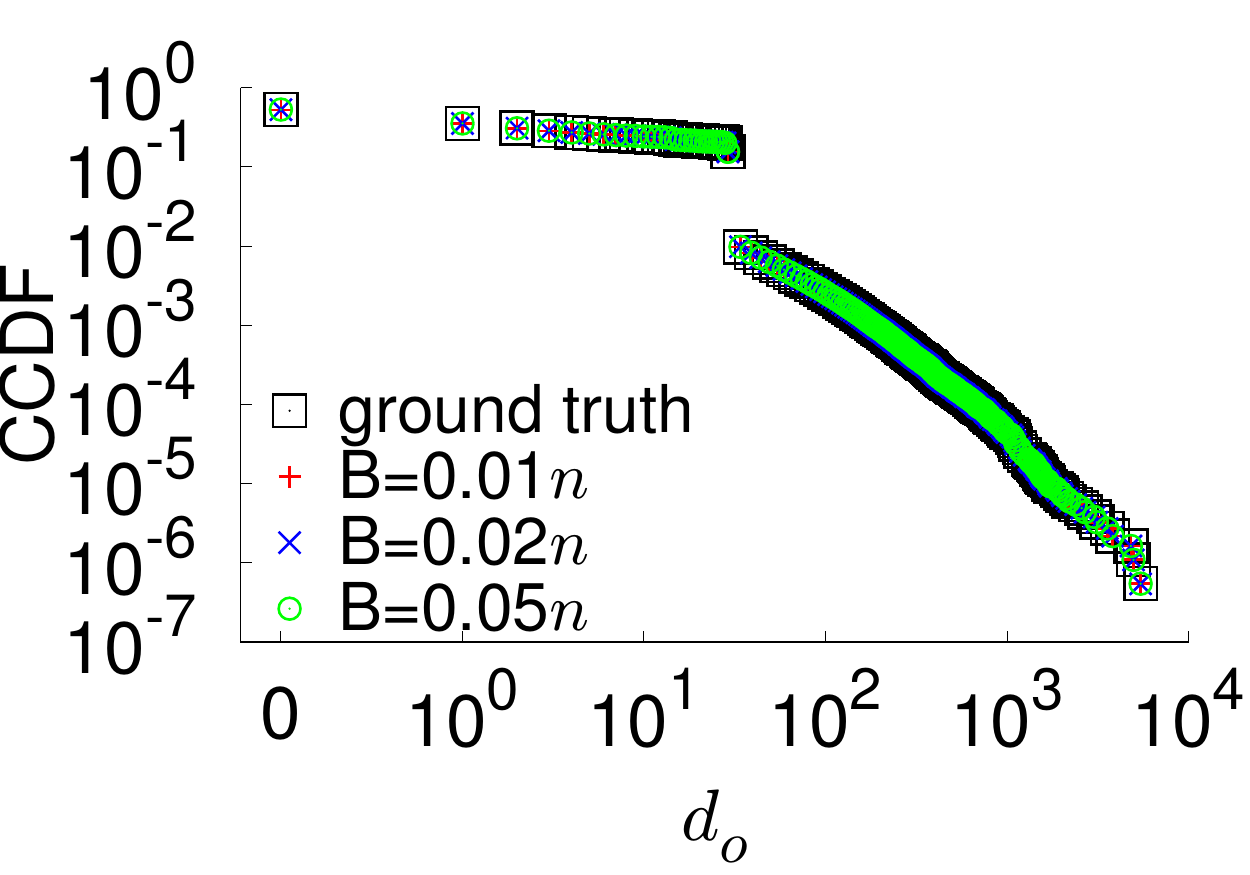}}
  \subfloat[CCDF NRMSE ($\alpha=1$)\label{f:mt_nrmse_do_B}]{%
    \includegraphics[width=.25\linewidth]{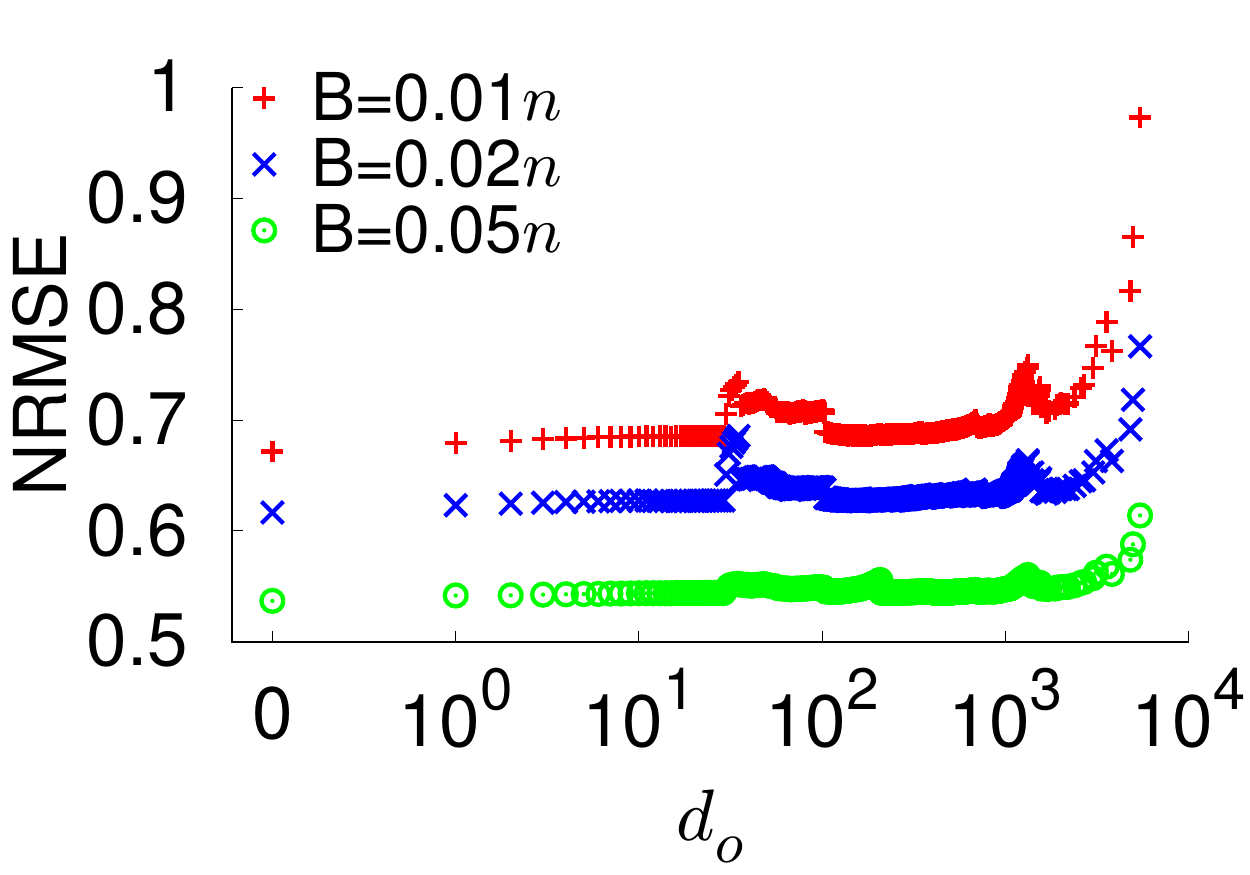}}
  \subfloat[CCDF NRMSE ($B=0.01n$)\label{f:mt_nrmse_do_a}]{%
    \includegraphics[width=.25\linewidth]{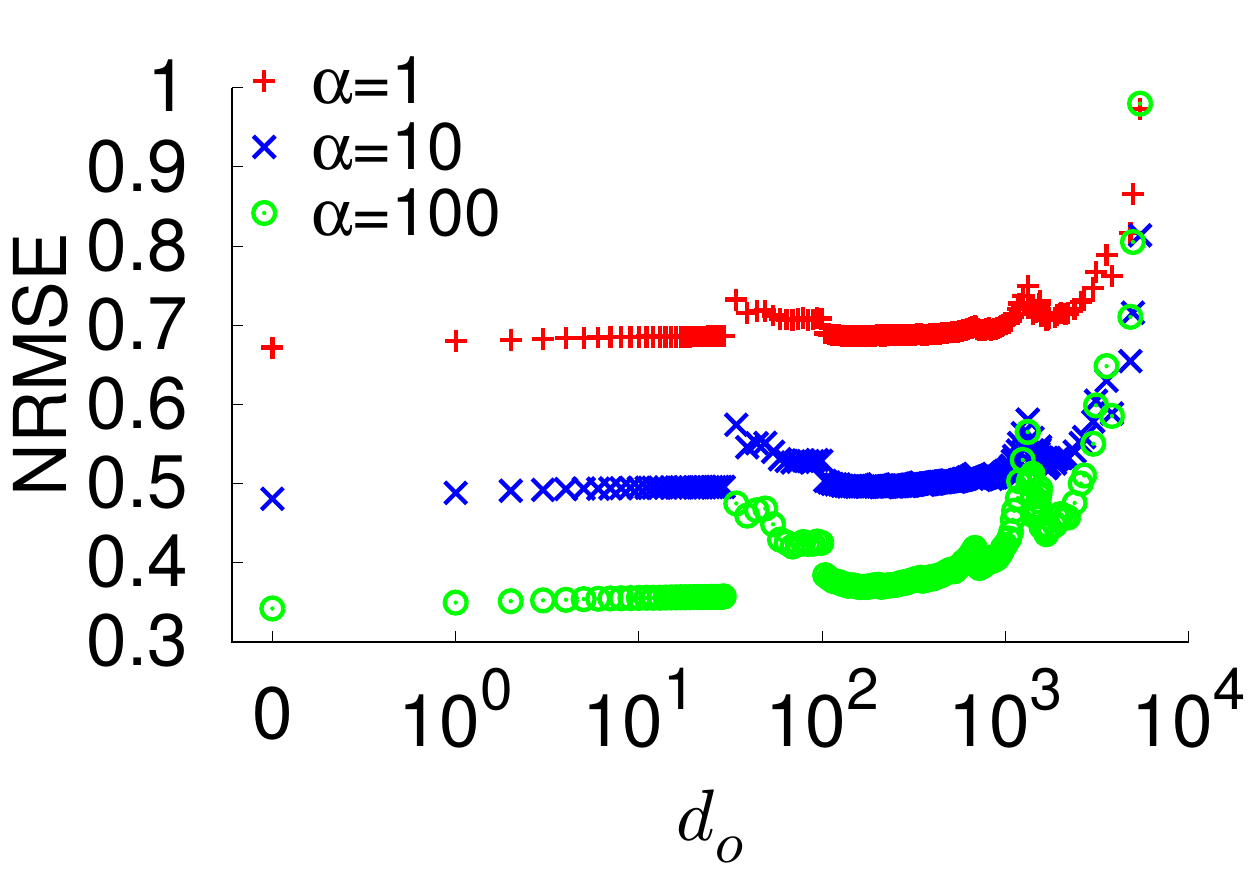}}
  \subfloat[too frequent jumping ($B\!=\!0.01n$)]{%
    \includegraphics[width=.25\linewidth]{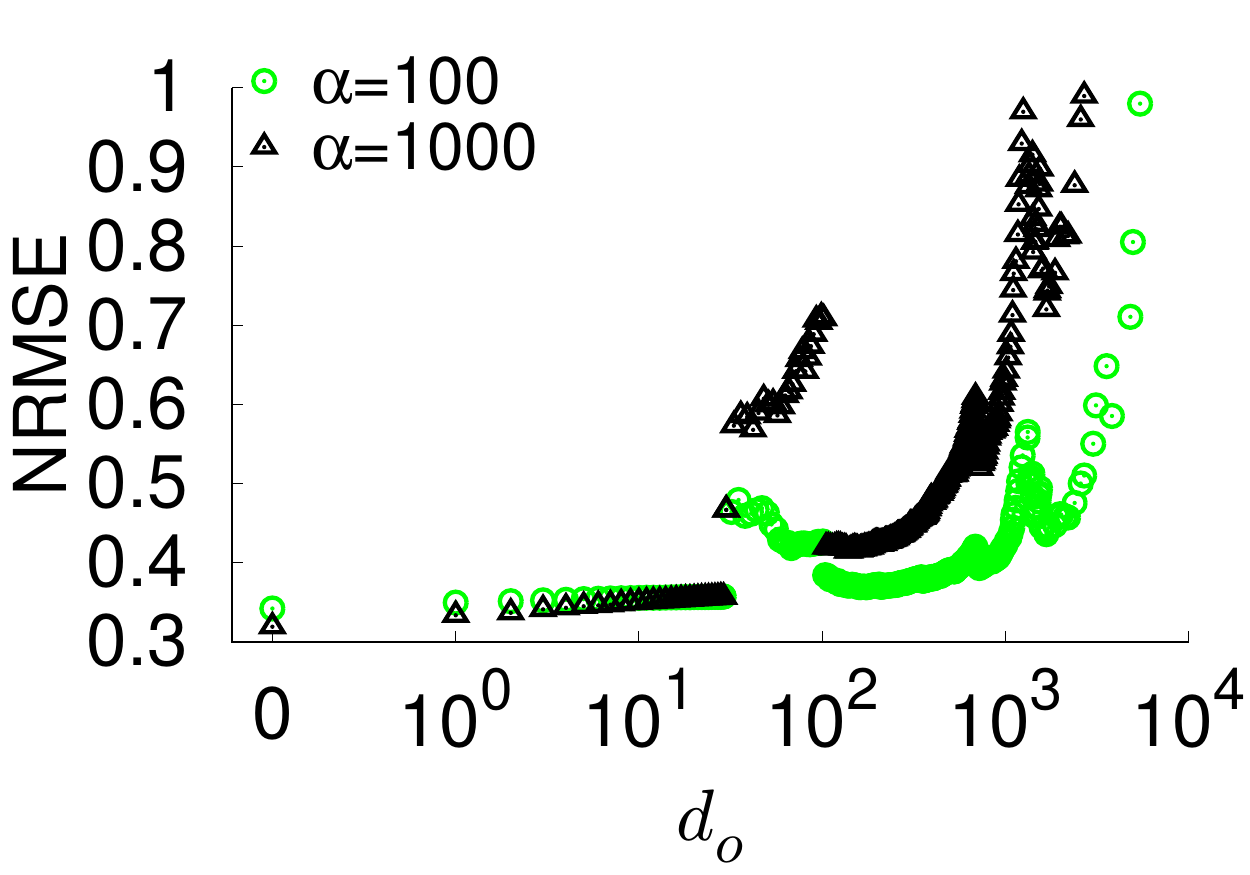}}
  \caption{\RW{T}\VS{A} degree distribution estimation and NRMSE analysis.
  \label{fig:mt_rwtvsa}}
\end{figure*}

Figure~\ref{fig:mt_rwtrwa} depicts the results of \RW{T}\RW{A}, and they are
similar to the results of \RW{T}\VS{A}.
First, from Figs.~\ref{f:mt_rwtrwa_est_di} and~\ref{f:mt_rwtrwa_est_do}, we
observe that \RW{T}\RW{A} can also provide unbiased estimates of the in- and
out-degree distributions.
Second, from Figs.~\ref{f:mt_rwtrwa_di_B} and~\ref{f:mt_rwtrwa_do_B}, we can find
that as sampling budget increases, the estimation error decreases accordingly for
both in- and out-degree estimations.
Last, from Figs.~\ref{f:mt_rwtrwa_di_ab} and~\ref{f:mt_rwtrwa_do_ab}, we find that
when jumping probability increases (by increasing $\alpha$ and $\beta$), the NRMSE
also decreases.

\begin{figure*}
  \centering
  \subfloat[in-degree est. ($\alpha=\beta=0.1$)\label{f:mt_rwtrwa_est_di}]{%
    \includegraphics[width=.25\linewidth]{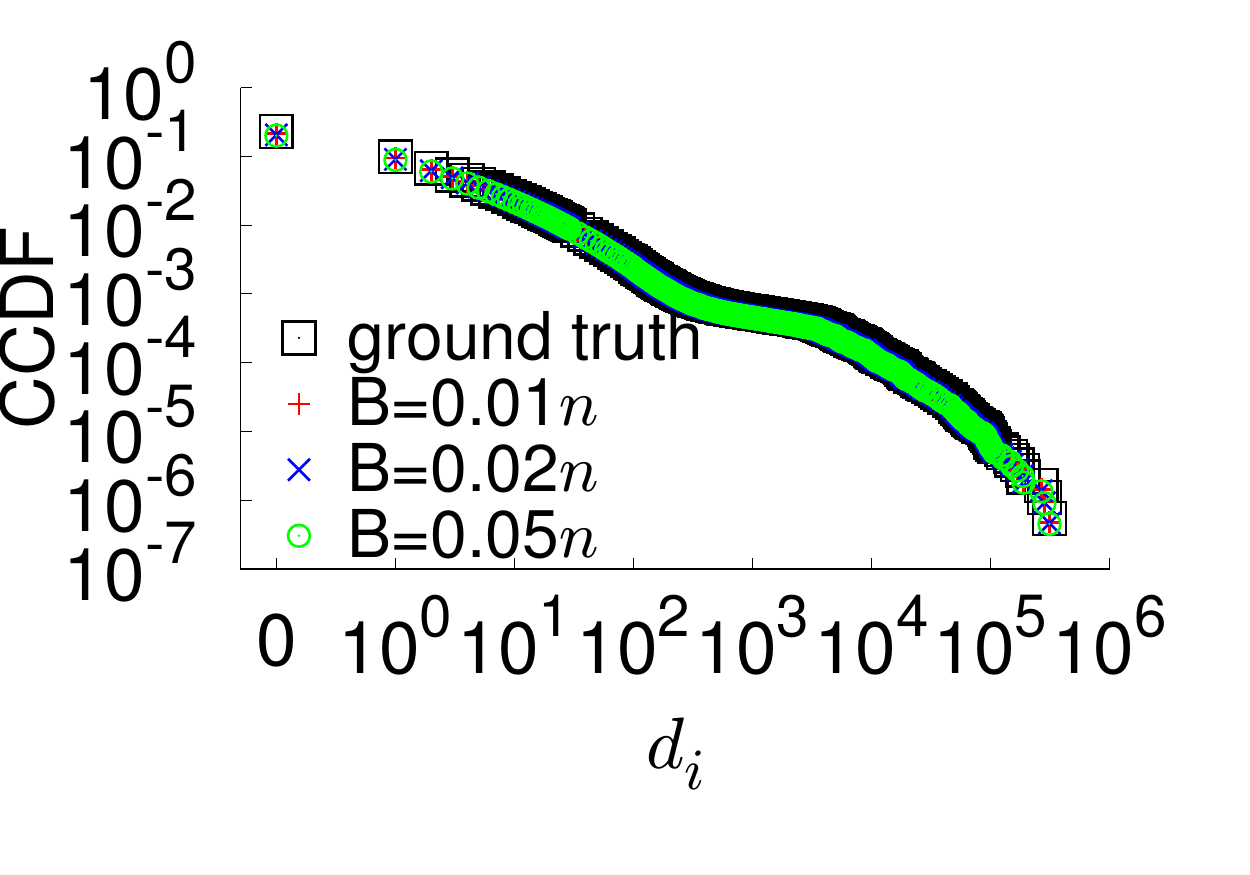}}
  \subfloat[CCDF NRMSE ($\alpha=\beta=0.1$)\label{f:mt_rwtrwa_di_B}]{%
    \includegraphics[width=.25\linewidth]{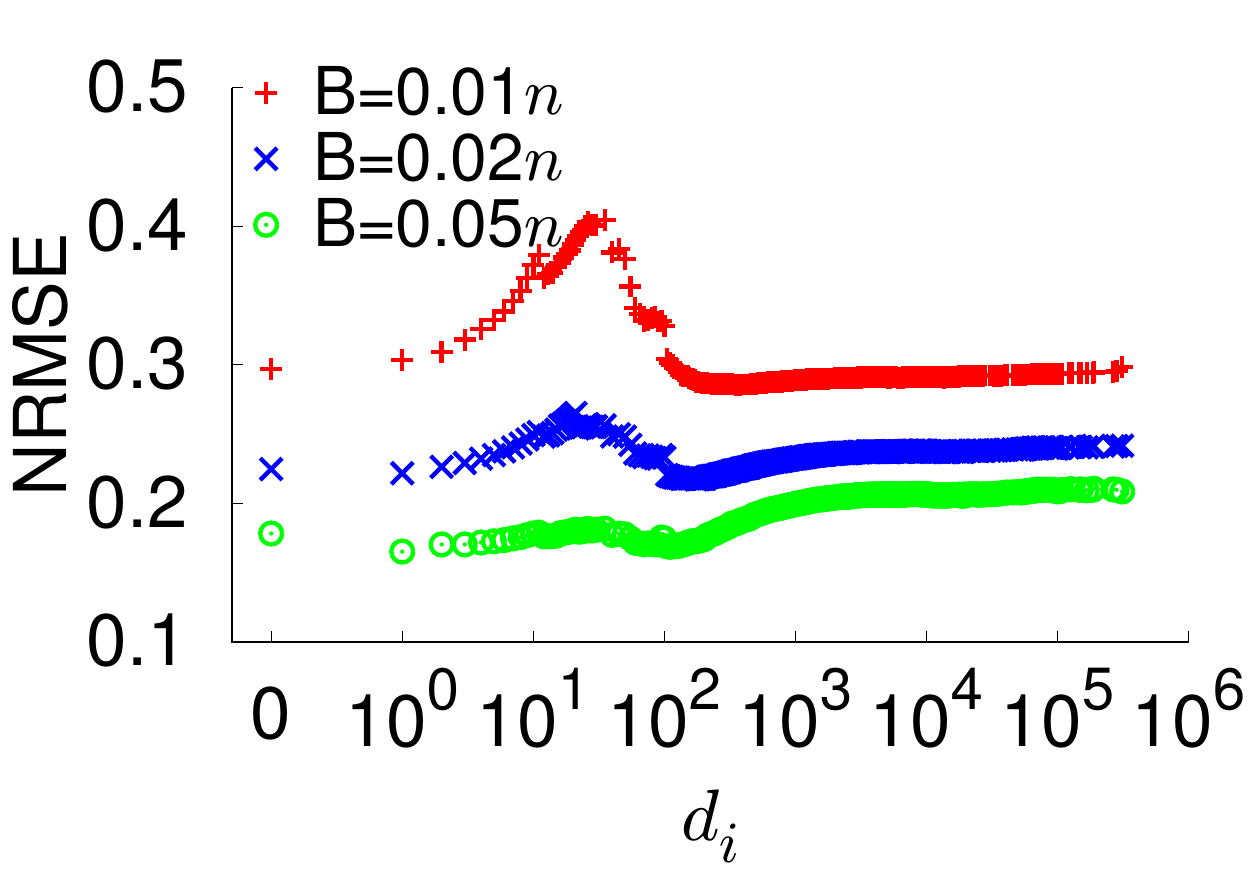}}
  \subfloat[CCDF NRMSE ($B=0.01n$)\label{f:mt_rwtrwa_di_ab}]{%
    \includegraphics[width=.25\linewidth]{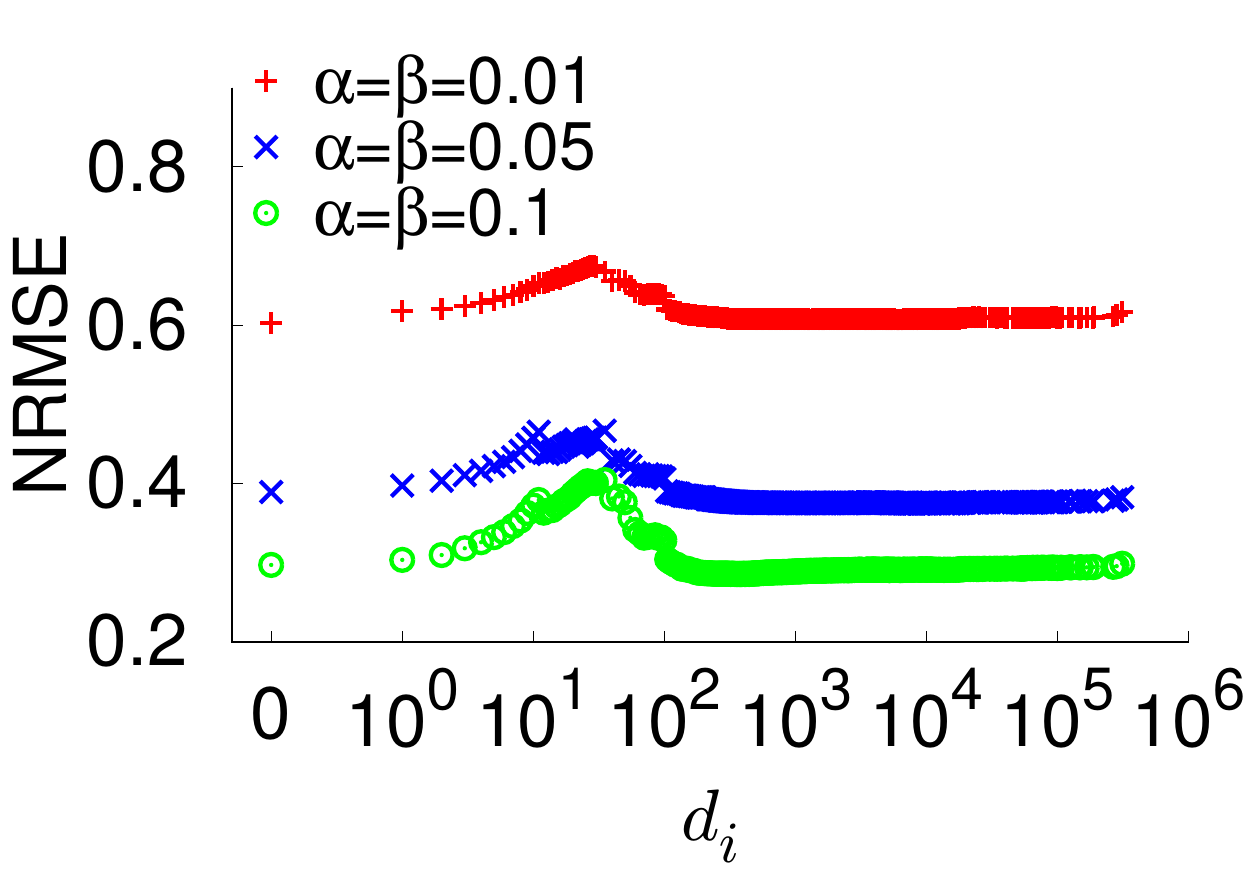}}
  \subfloat[too frequent jumping ($B\!=\!0.01n$)]{%
    \includegraphics[width=.25\linewidth]{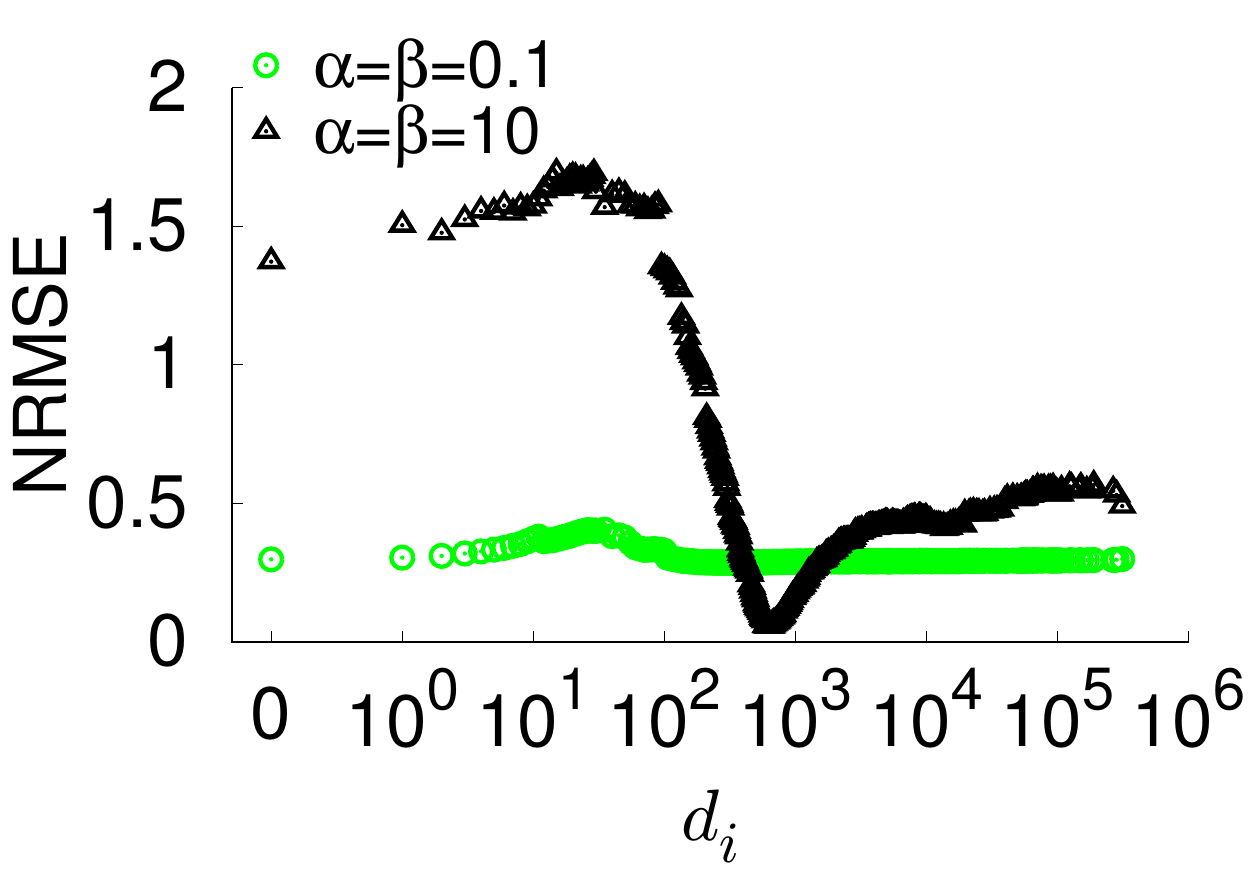}} \\
  \subfloat[out-degree est. ($\alpha=\beta=0.1$)\label{f:mt_rwtrwa_est_do}]{%
    \includegraphics[width=.25\linewidth]{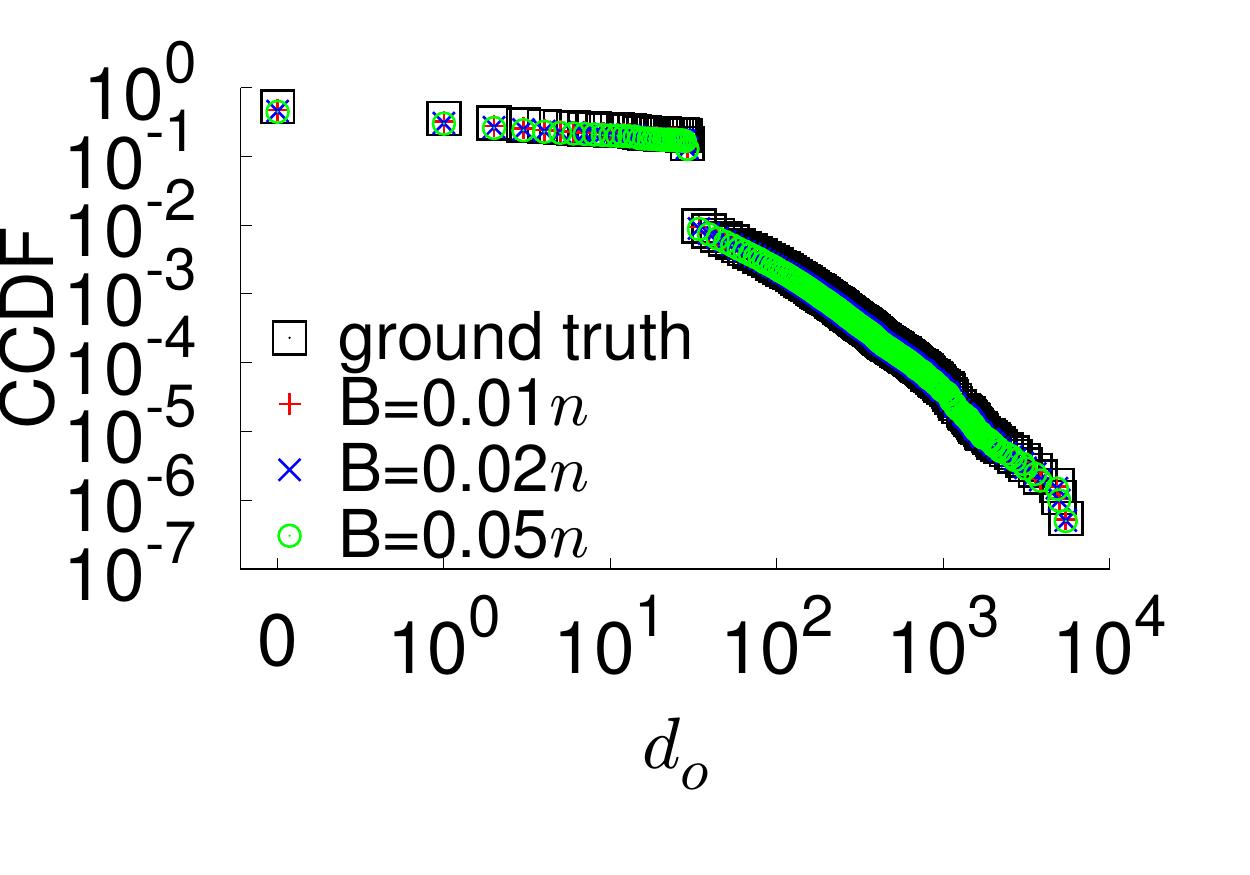}}
  \subfloat[CCDF NRMSE ($\alpha=\beta=0.1$)\label{f:mt_rwtrwa_do_B}]{%
    \includegraphics[width=.25\linewidth]{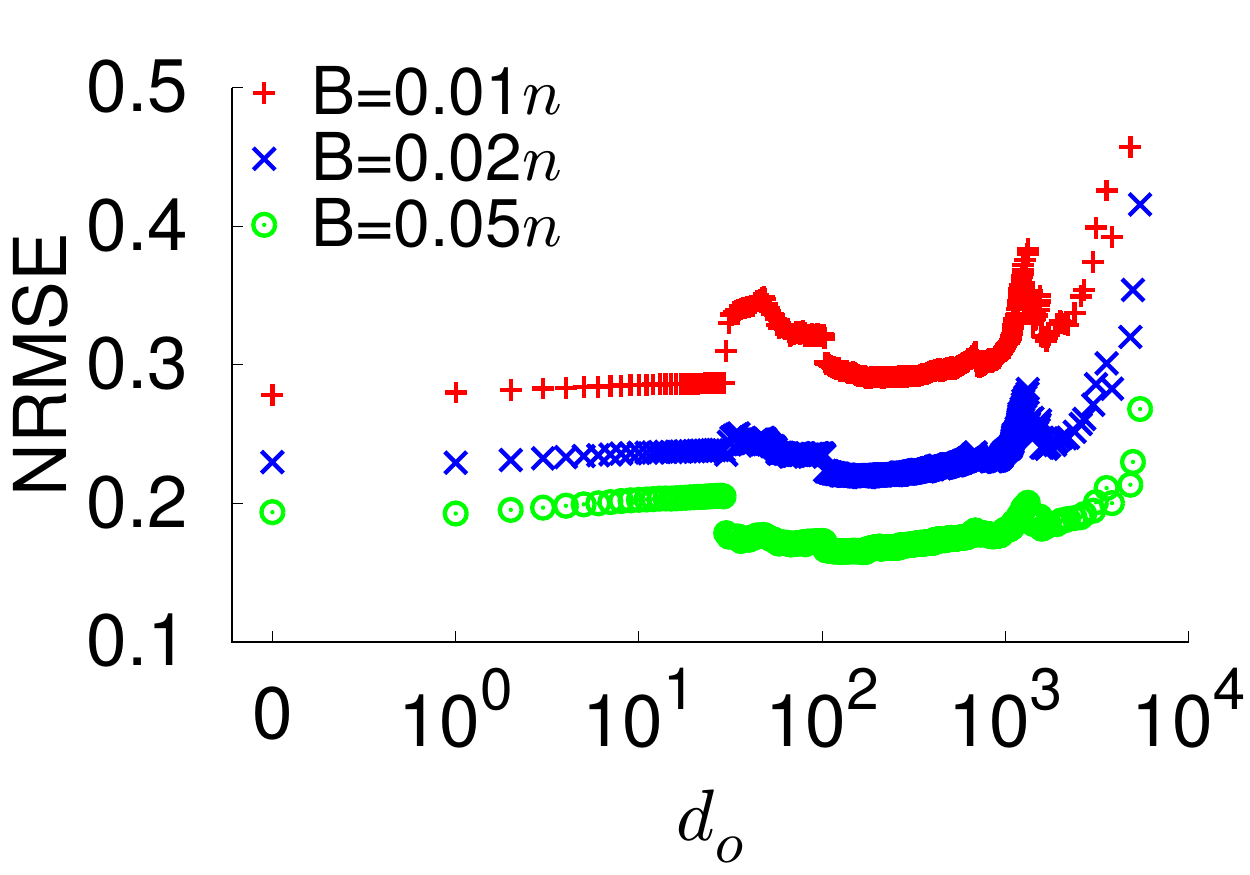}}
  \subfloat[CCDF NRMSE ($B=0.01n$)\label{f:mt_rwtrwa_do_ab}]{%
    \includegraphics[width=.25\linewidth]{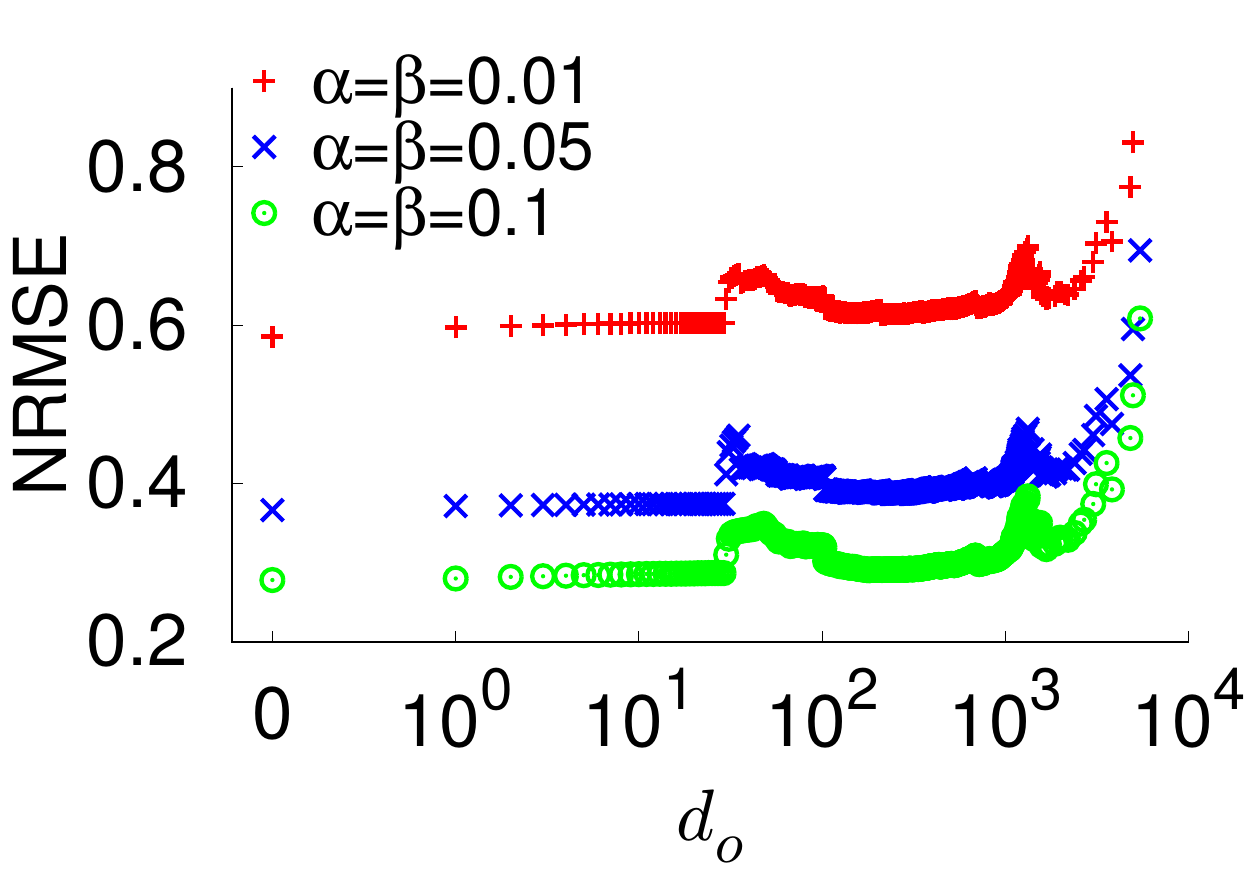}}
  \subfloat[too frequent jumping ($B\!=\!0.01n$)]{%
    \includegraphics[width=.25\linewidth]{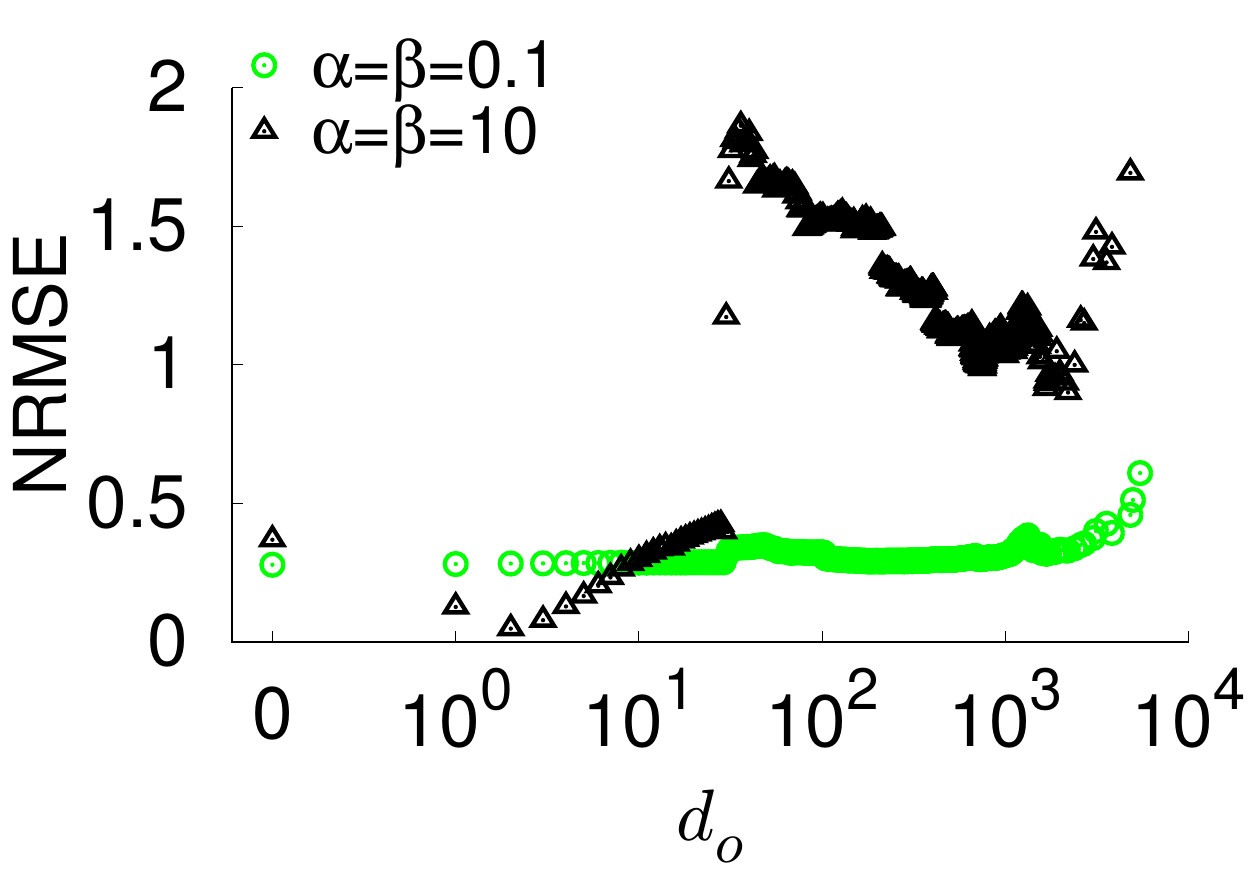}}
  \caption{\RW{T}\RW{A} degree distribution estimation and NRMSE analysis.
  \label{fig:mt_rwtrwa}}
\end{figure*}

However, it is worth noting that $\alpha$ and $\beta$ should not be too large for
both \RW{T}\VS{A} and \RW{T}\RW{A}.
Because we know that when $\alpha\rightarrow\infty$, \RW{T}\VS{A} becomes \VS{A},
which is biased on the Mtime dataset, and hence causes large NRMSE.
Similar behavior happens to \RW{T}\RW{A}, too.

\section{\textbf{Related Work}} \label{sec:related_work}

We briefly review some related literature in this section.

Graph sampling methods, especially random walk based graph sampling methods, have
been widely used to characterize large-scale complex networks.
These applications include, but are not limited to, estimating peer statistics in
peer-to-peer networks~\cite{Gkantsidis2006,Massoulie2006}, uniformly sampling
users from OSNs~\cite{Gjoka2010,Gjoka2011,Lee2012b,Xu2014}, characterizing
structure properties of large-scale
networks~\cite{Katzir2011,Hardiman2013,Seshadhri2013,Wang2014a}, and measuring
statistics of point-of-interests on maps~\cite{Wang2014}.
The above literature is mostly concerned with sampling methods that seek to
\emph{directly} sample nodes (or samples) in target graphs (or some sample
spaces).
However, direct sampling is not always efficient as we argued in this work.

When the target graph (or sample space) can not be directly sampled or direct
sampling is inefficient, several methods based on graph manipulation have been
proposed to improve sampling efficiency.
For example, Gjoka et al.~\cite{Gjoka2011a} study an approach to improve sampling
efficiency through building a \emph{multigraph} using different kinds of relations
(i.e., different types of edges) that exist on an OSN.
A multigraph is better connected than any individual graph formed by only one kind
of relations.
Therefore, the random walk can converge fast on this multigraph.
Zhou et al.~\cite{Zhou2013} exploit several criteria to rewire the target graph
on-the-fly to increase the graph conductance~\cite{Sinclair1989} and reduce mixing
time of a random walk.
Our method differs from theirs in that we do not manipulate target graphs.
We study a new approach that utilizes a widely existed two-layered network
structure to assist sampling on target graph indirectly.

Birnbaum and Sirken~\cite{Birnbaum1965} designed a survey method for estimating
the number of diagnosed cases of a rare disease in a population.
Directly sampling patients of a rare disease from the huge human population is
obviously inefficient, so they studied how to sample hospitals so as to sample
patients indirectly.
Their method motivates us to design the \VS{A} method.
However, as we pointed out, \VS{A} method cannot sample nodes that are not
connected to auxiliary graph, and we overcome this problem by designing
\RW{T}\VS{A} and \RW{T}\RW{A} methods.
Our work also complements existing sampling methods related to random walk with
jumps~\cite{Avrachenkov2010,Ribeiro2012b,Xu2014} by removing the necessity of
uniform node sampling on target graphs.

\section{Conclusion} \label{sec:conclusion}

When graphs become large in scale, sampling methods become necessary tools in the
study of characterizing their properties.
Among these sampling methods, random walk-based crawling methods are effective and
are gaining popularity.
However, if the graph under study is not well connected, random walk-based graph
sampling methods suffer from the slow mixing problem.
In this work, we observe that a graph usually does not exist in isolation.
In many applications, the target graph is accompanied with an auxiliary graph and
a bipartite graph, and they together form a better connected two-layered network
structure.
This new viewpoint brings extra benefits to the graph sampling framework.
We design three sampling methods to measure the target graph from this new
viewpoint, and these methods are demonstrated to be effective on both synthetic
and real datasets.
Therefore, our method complements existing methods in the literature of graph
sampling.

\bibliographystyle{IEEEtran}

\end{document}